\documentclass[a4papper,11pt]{article}
\usepackage[a4paper, total={6.5in, 9.5in}]{geometry}
\usepackage{amsmath,amsfonts,amssymb,amsthm,yfonts,tensor,dsfont}
\usepackage{authblk}
\usepackage{tikz-cd}
\newcommand\isoeq{\mathrel{\stackrel{\makebox[0pt]{\mbox{\normalfont\footnotesize iso}}}{=}}}
\newtheorem{manifold}{Definition}[subsection]
\newtheorem{liederiv}[manifold]{Definition}
\newtheorem{lievector}[manifold]{Proposition}
\newtheorem{evolution}{Proposition}[subsection]
\newtheorem{symplchart}[evolution]{Proposition}
\newtheorem{symplmanif}{Definition}[subsection]
\newtheorem{symplvec}[symplmanif]{Definition}
\newtheorem{symplperp}[symplmanif]{Lemma}
\newtheorem{sympltype}[symplmanif]{Definition}
\newtheorem{evenlagr}[symplmanif]{Lemma}
\newtheorem{symplframe}[symplmanif]{Proposition}
\newtheorem{darboux}[symplmanif]{Theorem}
\newtheorem{symplectomorphism}[symplmanif]{Definition}
\newtheorem{localhamilton}[symplmanif]{Proposition}
\newtheorem{poissonalg}[symplmanif]{Proposition}
\newtheorem{lieaction}{Definition}[subsection]
\newtheorem{moment}[lieaction]{Definition}
\newtheorem{noether}[lieaction]{Theorem}
\newtheorem{elementary}[lieaction]{Definition}

\newtheorem{bracohom}[lieaction]{Proposition}
\newtheorem{coadorbit}[lieaction]{Proposition}
\newtheorem{coad}[lieaction]{Proposition}
\newtheorem{classelem}[lieaction]{Proposition}
\newtheorem{presympl}[lieaction]{Definition}
\newtheorem{reduction}[lieaction]{Proposition}
\theoremstyle{definition}
\newtheorem{su2}[lieaction]{Example}
\newtheorem{poincare}[lieaction]{Example}
\newtheorem{automorphism}{Definition}[subsection]
\newtheorem{exactgroup}[automorphism]{Proposition}
\newtheorem{isompreq}[automorphism]{Proposition}
\newtheorem{prequantization}[automorphism]{Definition}
\newtheorem{preqdirac}[automorphism]{Proposition}
\newtheorem{preqcrit}[automorphism]{Proposition}
\newtheorem{preqreduction}[automorphism]{Proposition}
\newtheorem{preqsu2}[automorphism]{Example}
\newtheorem{realpolarization}{Definition}[subsection]
\newtheorem{connectionleaf}[realpolarization]{Proposition}
\newtheorem{polarizationcot}[realpolarization]{Proposition}
\newtheorem{polarizationcoord}[realpolarization]{Corollary}
\newtheorem{polinomial}[realpolarization]{Definition}
\newtheorem{polinomialindeed}[realpolarization]{Proposition}
\newtheorem{cxpolarization}[realpolarization]{Definition}
\newtheorem{complexlagrangian}[realpolarization]{Proposition}
\newtheorem{strongint}[realpolarization]{Definition}
\newtheorem{gencxlagr}[realpolarization]{Proposition}
\newtheorem{compatiblecoisotropic}[realpolarization]{Definition}
\newtheorem{reducpolarize}[realpolarization]{Proposition}
\newtheorem{stronginteg}[realpolarization]{Proposition}
\newtheorem{polsection}[realpolarization]{Definition}
\newtheorem{quantumsphere}[realpolarization]{Example}
\newtheorem{wavequation}[realpolarization]{Example}
\newtheorem{eulerlagrange}{Proposition}[subsection]
\newtheorem{omegafield}[eulerlagrange]{Proposition}
\newtheorem{classical}[eulerlagrange]{Example}
\newtheorem{fieldsymplectic}[eulerlagrange]{Example}
\newtheorem{cxg}{Proposition}[subsection]
\newtheorem{sho}[cxg]{Example}
\newtheorem{fockfield}[cxg]{Example}
\newtheorem{projectj}{Proposition}[subsection]
\newtheorem{projectcoherent}[projectj]{Lemma}
\newtheorem{composepi}[projectj]{Proposition}
\newtheorem{metasymplectic}[projectj]{Proposition}
\newtheorem{lagrangiangrassmanian}[projectj]{Proposition}
\newtheorem{halform}[projectj]{Definition}
\newtheorem{metadone}[projectj]{Proposition}
\newtheorem{correctsho}[projectj]{Example}
\newtheorem{metastruc}{Definition}[subsection]
\newtheorem{metapol}[metastruc]{Proposition}
\newtheorem{stringdual}[metastruc]{Example}

\usepackage[
sorting=none
]{biblatex}
\begin{document}

\title{Geometric Quantization: Particles, Fields and Strings}

\author[1]{David S Berman}
\author[2]{Gabriel Cardoso}
\affil[1]{School of Physical and Chemical Sciences, Queen Mary University of London, Mile End Road, London E1 4NS, U.K.}
\affil[2]{Department of Physics and Astronomy, Stony Brook University, Stony Brook, New York 11794-3800, USA}

\date{\today}

\maketitle

\begin{abstract} 
These notes present an introduction to the method of geometric quantization. We discuss the main theorems in a style suitable for a theoretical physicist with an eye towards the physical motivation and the interpretation of the geometric construction as providing a solution to Dirac's axioms of quantization. We provide in detail the examples of free relativistic particles, their corresponding quantum fields, and the bosonic string using formalism of double field theory. Based on lectures written by Gabriel Cardoso.

\end{abstract}


\section{Introduction}\label{sec:intro}

Twentieth century physics is remarkable for its use of geometric methods. The most impressive examples are Riemannian geometry in the theory of general relativity and the description of the fundamental forces of nature as fibre bundles. Wigner famously spoke of the unreasonable effectiveness of mathematics in the natural sciences. He might have just as well as discussed the the unreasonable effectiveness of geometry. But where does quantum theory fit into a geometric description of nature? Geometry and quantum theory have often given the impression of being distinct, maybe even (as suggested by quantum gravity) as being incompatible. Rather than trying to tackle the difficult topic of making quantizing geometry we will examine how we can make quantizing geometric.

Historically, the origin of what we mean by quantization is the canonical quantization proceedure, which provided the derivation of the Schr\"odinger equation from classical mechanics through a prescription for replacing the canonical position and momentum variables by quantum operators. Later, this replacement was understood as a realization of Heisenberg's canonical commutation relations. The commutation relations of operators implied the quantum uncertainty principle of observables while also suggested the relationship between the algebra of quantum observables and the Poisson brackets of classical observables.

It was then Dirac \cite{dirac1981principles} who first proposed thinking of the quantization prescription in a more general sense: given a classical system, which in practice means a phase space and the relevant classical observables, how to construct the analogue quantum system, ie the Hilbert space and relevant quantum operators? In particular, the association $\mathcal{Q} : f\mapsto\hat{f}$ from functions on the classical space of states to operators on the quantum space of states should be such that
\begin{align}\label{diracintro}
\begin{split}
\bullet \, & \mathcal{Q} : f\mapsto\hat{f} \text{ is $\mathbb{R}$-linear} \\ 
\bullet \, & [\mathcal{Q}(f),\mathcal{Q}(g)] = -i\hbar \mathcal{Q}(\{f,g\})\\
\bullet \, & f \text{ is a constant function} \Rightarrow \mathcal{Q}(f) = f\mathds{1} \text{ acts by multiplication by $f$} \\
\end{split},
\end{align}
where $[\cdot,\cdot]$ denotes the commutator of linear maps and $\{\cdot,\cdot\}$ is the Poisson bracket. It was realised early on that such a map cannot be extended to all the classical observables, so one should also have some criteria to select a subalgebra of the observables to be quantized. Finally, one also expects that, if the action of some symmetry group on the classical phase space by canonical transformations reveals important physical properties of the system, this group should also act on the quantum space of states by (protectively) unitary transformations. Note that it is the quantum system which is thought of to be more fundamental, so that in principle there is no guarantee that reconstructing it through quantization is even possible. Surprisingly, however, it has been successfully applied to a wide class of problems, and is in practice the only way quantum theories are effectively constructed, from field theory to condensed matter physics.

Besides its practical importance in physics, quantization has sparked the interest of both the mathematics and the mathematical physics communities, because it leads to interesting technical questions. For example, it is common for physics problems to have symmetries, which typically appear as group actions. For such systems, quantization naturally relates to the theory of unitary representations. Thus there are now various ``methods'' of quantization, which attempt to solving the demands of quantization, reveal their mathematical structure, and explain obstructions and subtleties related to this procedure, like deformation quantization, BV formalism and, the subject of these notes, geometric quantization.

Geometric quantization uses the geometry of phase space to construct the quantum states and the operators coresponding to observables. The underlying geometry of the Hamiltonian formalism of mechanics is symplectic geometry. It is natural to ask how the quantization proceedure fits into this symplectic geometric picture of classical mechanics as a geoemtric construction. This was the question most notably spearheaded by Kostant \cite{kostant1970quantization} and Souriau \cite{souriau1970strucure}, but it was later developed further by many others. 

The method of consists in three parts: prequantization, which relies on the geometry of complex line bundles with connection and hermitian strucutre; quantization, which uses polarizations, a special type of integrable distribution present in symplectic manifolds; and finally the metaplectic correction, which involves extending the symplectic group of classical mechanics to its double cover (just as a spinor extends the rotation group to its double cover). As we will see, the quantization process brings all these geometric ingredients together to produce a general solution to Dirac's axioms of quantization. 

The notes will follow the logical structure beginning with classical symplectirc geometry and then using the geometry of complex line bundles to construct the prequantum bundle before introducing full quantization and the role of polarizations and finally the metaplectic correction. Along the way we will give a few with applications to a physically significant problems, namely the derivation of the wave equations of relativistic quantum mechanics, the corresponding free quantum fields, and the quantum string.

A general familiarity with manifolds, Lie derivatives, Lie groups and Lie algebras, connections on fibre bundles etc. at the level of \cite{nakahara} is desirable, but we tried to include some of the main definitions and use similar notations so it should be possible to look up the necessary concepts as they appear. For those interested in carrying on into more technical details of the quantization procedure, we strongly recommend \cite{woodhouse1997geometric,asymptotics,souriau1970strucure,sniatycki2012geometric,kirillov2001geometric}, which served as the main references for these notes.

\pagebreak

\section{Mechanics and Symplectic Geometry}\label{mechgeom}
From the perspective of geometric quantization a classical system is simply a symplectic manifold equipped with a Hamiltonian flow. We will now review how this notion arises naturally from a geometrization of classical mechanics. A couple of standard references for the topics in this section are \cite{bruhat,nakahara,goldstein,symplectic,marsden}. (In the whole text, we assume that all manifolds are smooth and make use of the Einstein summation convention unless otherwise stated.)

\subsection{Manifolds}

\begin{manifold}
An $m$-dimensional \emph{smooth manifold} is a second-countable Hausdorff topological space M with a smooth atlas, ie. a family $ \left \{ (U_{i}, \phi _{i}) \right \} $ of charts $(U_{i}, \phi _{i})$ such that 
\begin{description}
\item[(i)] $\left \{U_{i}\right \}$ is a family of open sets which covers M, that is, $ \cup U_{i} = M$;
\item[(ii)] for each $i$, $\phi _{i}:U_{i} \to \mathbb{R} ^{m}$ is a homeomorphism;
\item[(iii)] whenever $U_{i} \cap U_{j} \neq \varnothing$, the map $\phi _{i} \circ \phi _{j}^{-1}: \phi _{j}(U_{i} \cap U_{j}) \to \phi _{i}(U_{i} \cap U_{j})$ is infinitely differentiable.
\end{description}
\end{manifold}


This is a definition that is intuitive for the physicist to understand: a smooth manifold is a space which admits local coordinates, and the transition functions between such local coordinates are smooth. It generalises $\mathbb{R}^{n}$ in that not necessarily there is any global coordinate system which covers the entire manifold. One can then use the coordinates to define differentiability of functions between manifolds, vector fields and so on. Particularly, one can define the Lie derivative.

\begin{liederiv}\label{liederiv}
Let $T$ denote a smooth tensor field in $M$ and $X$ be a smooth vector field in $M$. If $\rho : M \times \mathbb{R} \to M$ denotes the flow of $X$, then the Lie derivative of $T$ along $X$ is the tensor field $\mathcal{L}_{X}T$ defined at each point by
\begin{equation}\nonumber
\mathcal{L}_{X}T|_{q} = \frac{d}{dt} (\rho_{- t})_{*}T|_{\rho_{ t}(q)}\Big|_{t=0} = \lim_{t \to 0} \frac{(\rho_{- t}(q))_{*}T|_{\rho_{ t}(q)} - T|_{q}}{t},
\end{equation}
where $\rho_{t}(q) = \rho(q, t)$ and $\rho_{*}$ is the associated differential map.
\end{liederiv}
The restriction of this operation to the linear space $V(M)$ of vector fields on $M$ gives it the structure of a Lie algebra.

\begin{lievector}\label{lievector}
Let $X, Y \in V(M)$. The \emph{Lie bracket} $[X,Y] \in V(M)$ defined by $[X,Y] = \mathcal{L}_{X}Y$ has the following properties:
\begin{description}
\item[(i)] (Linearity) \\
$[aX + bY,Z] =  a[X, Z] + b[Y,Z]$,
\item[(ii)] (Antisymmetry) \\
$[X,Y] = -[Y,X]$,
\item[(iii)] (Jacobi identity) \\
 $[[X,Y],Z] + [[Y,Z],X] + [[Z,X],Y] = 0$,
\end{description}
$\forall a, b \in \mathbb{R}, \forall X, Y \in V(M)$.
\end{lievector}

\begin{proof}
This is a good example of how local coordinates can come in handy. In a local chart $x^{\mu}$, the flow of $\rho$ of $X$ solves $\frac{d \rho^{\mu}(x, t)}{dt} = X^{\mu}(\rho(x, t))$ and therefore
\begin{equation}\nonumber
\rho^{\mu}_{\epsilon}(q) = \rho^{\mu}(q, \epsilon) = q^{\mu} + \epsilon X^{\mu}(q) + O(\epsilon ^2)
\end{equation}
Thus to first order in $\epsilon$
\begin{equation}\nonumber
\begin{split}
Y|_{\rho_{\epsilon}(q)} & = Y^{\mu}(q + \epsilon X(q))\partial_{\mu}|_{q + \epsilon X} \\
&= [Y^{\mu}(q)+\epsilon X^{\nu}(q)\partial_{\nu}Y^{\mu}(q)]\partial_{\mu}|_{q+ \epsilon X}
\end{split}
\end{equation}
Pushing this vector forward to $q$ gives
\begin{equation}\nonumber
\begin{split}
(\rho _{- \epsilon})_{*} Y|_{\rho_{\epsilon}(q)} & = [Y^{\mu}(q)+\epsilon X^{\nu}(q)\partial_{\nu}Y^{\mu}(q)](\rho _{- \epsilon})_{*}\partial_{\mu}|_{q+ \epsilon X} \\
& =  [Y^{\mu}(q)+\epsilon X^{\nu}(q)\partial_{\nu}Y^{\mu}(q)]\partial_{\mu}[x^{\nu} - \epsilon X^{\nu}(q)]\partial_{\nu}|_{q} \\
& =  [Y^{\mu}(q)+\epsilon X^{\nu}(q)\partial_{\nu}Y^{\mu}(q)][\delta^{\nu}_{\mu} - \epsilon \partial_{\mu}X^{\nu}(q)]\partial _{\nu}|_{q} \\
& = [Y^{\nu} + \epsilon (X^{\mu}(q)\partial_{\mu}Y^{\nu}(q) - Y^{\mu}(q)\partial_{\mu}X^{\nu}(q))]\partial _{\nu}|_{q}
\end{split}
\end{equation}
up to terms $O(\epsilon^2)$. Substituting this in Definition \ref{liederiv} we get the coordinate expression
\begin{equation}\nonumber
[X, Y] = (X^{\mu}\partial_{\mu}Y^{\nu} - Y^{\mu}\partial_{\mu}X^{\nu})\partial_{\nu},
\end{equation}
from which \textbf{(i)} and \textbf{(ii)} follow immediately while \textbf{(iii)} is a straightforward application of the Leibnitz rule.
\end{proof}

\subsection{Mechanics in $\mathbb{R}^{2n}$}\label{mecrn}

The phase space of a classical mechanics system has more than just a differentiable structure and this can be motivated by the fact that it comes with a special class of \emph{canonical} coordinates. Their significance appears in the Hamiltonian formulation in the following way. Take $\mathbb{R}^{2n}$ with coordinates $(p_{1}, ..., p_{n}, q^{1} ..., q^{n})$ (the convention of the position of the indices will be clear later), as the phase space. To write the expressions more compactly, we define $\boldsymbol{\zeta} = (p_{a}, q^{b})$, so that

\begin{equation}\nonumber
\frac{\partial f}{\partial \boldsymbol{\zeta}} := \left(\frac{\partial f}{\partial p_{1}}, ..., \frac{\partial f}{\partial p_{n}}, \frac{\partial f}{\partial q^{1}},..., \frac{\partial f}{\partial q^{n}}\right)^{T}
\end{equation}
for any function $f(p_{a}, q^{b})$ and let
\begin{equation}\nonumber
\mathbf{J} := 
\begin{pmatrix}
\mathbf{0} & \mathbf{1} \\
\mathbf{-1} & \mathbf{0}
\end{pmatrix}
\end{equation}
where $\mathbf{0}$ and $\mathbf{1}$ are, respectively, the zero and the identity $n \times n$ matrices. Then the Poisson bracket $\{f, g\}_{\boldsymbol{\zeta}}$ between the functions $f, g: \mathbb{R}^{2n} \to \mathbb{R}$ in the coordinates $(p_{a}, q^{b})$ is defined as

\begin{equation}\label{poisson}
\{f, g\}_{\boldsymbol{\zeta}} := \left( \frac{\partial f}{\partial \boldsymbol{\zeta}} \right)^{T}\mathbf{J}\left( \frac{\partial g}{\partial \boldsymbol{\zeta}} \right).
\end{equation}
Notice that the coordinates $\boldsymbol{\zeta} = (p_{a},q^{b})$ satisfy the \emph{fundamental Poisson brackets}: $\{q^{a}, q^{b}\}_{\boldsymbol{\zeta}} = 0 = \{p_{a}, p_{b}\}_{\boldsymbol{\zeta}}$ and $\{p_{a}, q^{b}\}_{\boldsymbol{\zeta}} = \delta^{b}_{a}$, which we summarize as

\begin{equation}\nonumber
\{\boldsymbol{\zeta}, \boldsymbol{\zeta}\}_{\boldsymbol{\zeta}} = \mathbf{J}.
\end{equation}
Then, given an expression for the Hamiltonian function $H(p, q)$ in these coordinates, the dynamics of the system is given by the curve $(p,q)(t)$ which solves Hamilton's equations

\begin{equation}\nonumber
\frac{dp_{a}}{dt}(t) = \{ H, p_{a}\}_{\boldsymbol{\zeta}}, \,\,\,\,\, \frac{dq^{b}}{dt}(t) = \{ H, q^{b}\}_{\boldsymbol{\zeta}},
\end{equation}
or, equivalently,
\begin{equation}\label{hamiltoneqn}
\frac{d \boldsymbol{\zeta}}{dt}(t) = \{H, \boldsymbol{\zeta}\}_{\boldsymbol{\zeta}}.
\end{equation}
In fact, the Hamiltonian generates the time evolution of all observables\footnote{We consider observables given by a function $f:\mathbb{R}^{2n}\to\mathbb{R}$, whose only time-dependence comes from evaluation on the phase space trajectory traced out by the system.} through the Poisson bracket:
\begin{evolution}
Let $f: \mathbb{R}^{2n} \to \mathbb{R}$ be a function in phase space (a classical observable). Then its restriction $f(\boldsymbol{\zeta}(t))$ to the trajectory of the mechanical system of Hamiltonian $H(\boldsymbol{\zeta})$  is given by

\begin{equation}\label{hamiltonobs}
\frac{df}{dt}(\boldsymbol{\zeta}(t)) = \{H,f\}_{\boldsymbol{\zeta}}.
\end{equation}
\end{evolution}
\begin{proof}
\begin{equation}\nonumber
\begin{split}
\frac{df}{dt}(p_{a}(t),q^{b}(t)) &= \frac{\partial f}{\partial p_{a}}\frac{dp_{a}}{dt} + \frac{\partial f}{\partial q^{b}}\frac{dq^{b}}{dt} = \frac{\partial f}{\partial p_{a}} \{ H, p_{a}\}_{\boldsymbol{\zeta}} + \frac{\partial f}{\partial q^{b}} \{ H, q^{b}\}_{\boldsymbol{\zeta}} \\
&= -\frac{\partial f}{\partial p_{a}}\frac{\partial H}{\partial q^{a}} + \frac{\partial f}{\partial q^{a}}\frac{\partial H}{\partial p_{b}} = \{H,f\}_{\boldsymbol{\zeta}}.
\end{split}
\end{equation}
\end{proof}

The condition for a second coordinate system $\boldsymbol{\eta} = (p'_{a}, q'^{b})$ to be such that the Hamiltonian dynamics is still expressed in terms of equations (\ref{hamiltoneqn}) and (\ref{hamiltonobs}) is the following.

\begin{symplchart}
Two coordinate systems $\boldsymbol{\zeta}=(p_{a},q^{b})$ and $\boldsymbol{\eta}=(p'_{a},q'^{b})$ are such that the Poisson brackets of any two arbitrary functions $f, g: \mathbb{R}^{2n} \to \mathbb{R}$ are equal, i.e.,
\begin{equation}\label{eqcanon}
\{f, g\}_{\boldsymbol{\zeta}} = \{f, g\}_{\boldsymbol{\eta}}, \,\,\,\, \forall f, g,
\end{equation}
if, and only if, the Jacobian matrix 
\begin{equation}\nonumber
\frac{\partial \boldsymbol{\zeta}}{\partial \boldsymbol{\eta}} = \left(
\begin{array}{c|c}
\frac{\partial p_{a}}{\partial p'_{b}} & \frac{\partial p_{a}}{\partial q'^{b}} \\ \hline
\frac{\partial q^{a}}{\partial p'_{b}} & \frac{\partial q^{a}}{\partial q'^{b}} 
\end{array}
\right)
\end{equation}
is an element of the symplectic group $SP(n, \mathbb{R})$, where
\begin{equation}\nonumber
SP(n, \mathbb{R}) = \{ \mathbf{A} \in GL(2n, \mathbb{R}) | \mathbf{A^{T}JA} = \mathbf{J}\}.
\end{equation}
\end{symplchart}

\begin{proof}
Clearly the Jacobian matrix should be invertible for the coordinate transformation to not be singular. Also, substituting (\ref{poisson}) in (\ref{eqcanon}), we have
\begin{equation}\nonumber
\begin{split}
\left( \frac{\partial f}{\partial \boldsymbol{\zeta}} \right)^{T}\mathbf{J}\left( \frac{\partial g}{\partial \boldsymbol{\zeta}} \right) &= \{f, g\}_{\boldsymbol{\zeta}} = \{f, g\}_{\boldsymbol{\eta}} = \left( \frac{\partial f}{\partial \boldsymbol{\eta}} \right)^{T}\mathbf{J}\left( \frac{\partial g}{\partial \boldsymbol{\eta}} \right) \\
&= \left( \frac{\partial f}{\partial \boldsymbol{\zeta}} \right)^{T}\left[ \left(\frac{\partial \boldsymbol{\zeta}}{\partial \boldsymbol{\eta}}\right)^{T} \mathbf{J}\left(\frac{\partial \boldsymbol{\zeta}}{\partial \boldsymbol{\eta}}\right)\right]\left( \frac{\partial g}{\partial \boldsymbol{\zeta}} \right), \,\,\,\,\, \forall f, g,
\end{split}
\end{equation}
and the claim follows.
\end{proof}

In particular, this shows that if the Jacobian matrix is symplectic then also in the $\boldsymbol{\eta} = (p'_{a}, q'^{b})$ coordinates we have
\begin{equation}\nonumber
\{\boldsymbol{\eta}, \boldsymbol{\eta}\}_{\boldsymbol{\eta}} = \mathbf{J}, \,\,\,\,\,\, \frac{d \boldsymbol{\eta}}{dt}(t) = \{ \boldsymbol{\eta}, H\}_{\boldsymbol{\eta}}, \,\,\,\,\,\, \frac{df}{dt}(\boldsymbol{\eta}(t)) = \{f,H\}_{\boldsymbol{\eta}},
\end{equation}
where the two last refer to the trajectory of the system. It follows that the new coordinate system is just as good as the old one to formulate classical mechanics. In fact, it can be better, in the sense that the equations of motion can be simpler in the new coordinates \footnote{An example is given by the Hamilton-Jacobi theory, where one uses this fact to bring Hamilton's equations into a trivial form (cf. chapter $10$ of \cite{goldstein})}. Thus one says that both coordinate systems are \emph{canonical} and the transformation relating them is said to be a \emph{canonical transformation}.

\subsection{Mechanics in a manifold}

A manifold generalizes $\mathbb{R}^n$ in that it admits \emph{local} coordinates which are consistently related by smooth diffeomorphisms. Similarly, a \emph{symplectic manifold} generalizes $\mathbb{R}^{2n}$ in that it admits a special class of local coordinate systems which are related to one another by canonical transformations. As a byproduct, a symplectic manifold comes with a coordinate-free generalization of classical mechanics, including geometrical definitions of canonical coordinates, Hamilton's equations, and Poisson brackets.

\begin{symplmanif}
A \emph{symplectic manifold} is a pair $(M, \omega)$ in which $M$ is a smooth manifold and $\omega$ is a closed nondegenerate two-form on $M$. In other words,
\begin{equation}\nonumber
\omega \in \Omega^2(M), \,\,\,\,\,\,\,\,\, d\omega = 0,
\end{equation}
and the map
\begin{equation}\nonumber
T_{m}M \to T^{*}_{m}M: X \to X\lrcorner \omega
\end{equation}
is a linear isomorphism at each $m \in M$, where the contraction $\lrcorner$ is the generalization to tensor fields of the map $V(M)\times\Omega^{1}(M):\to C^{\infty}(M):(X,\theta)\mapsto X\lrcorner\theta = \theta(X)$.

The two-form $\omega$ is called the \emph{symplectic structure} of $(M, \omega)$.
\end{symplmanif}

In particular, any (even-dimensional) vector space with a specified antisymmetric nondegenerate bilinear form is a symplectic manifold if we think of the components in some basis as chart coordinates. Hence the definition
\begin{symplvec}
A \emph{symplectic vector space} is a pair $(V, \omega)$, where $V$ is a vector space\footnote{In this subsection we will mainly refer to real vector spaces, but one should bear in mind the obvious generalization to the complex case.} and $\omega$ is an antisymmetric, nondegenerate bilinear form on $V$. That is $\omega(X, Y) = - \omega(Y,X), \,\,\, \forall X, Y \in V$ and $X \lrcorner \omega = 0 \Leftrightarrow X = 0$.
\end{symplvec}
These will appear both as phase spaces of linear systems as well as the tangent spaces of general symplectic manifolds. 

It will be useful to define the symplectic complement $F^{\perp}$ of a given subspace $F \subset V$ by
\begin{equation}\nonumber
F^{\perp} = \{X \in V | \omega(X, Y) = 0, \,\,\, \forall Y \in F\},
\end{equation}
which is also a subspace. The symplectic complement has the simple properties
\begin{symplperp}
If $F, G \subset V$ are subspaces of $V$ then
\begin{description}
\item[(i)]$F \subset G \Rightarrow F^{\perp} \supset G^{\perp}$
\item[(ii)]$(F^{\perp})^{\perp} = F$
\item[(iii)]$(F+G)^{\perp} = F^{\perp}\cap G^{\perp}$
\item[(iv)]$(F\cap G)^{\perp} = F^{\perp}+G^{\perp}$
\item[(v)]$\dim(V) = \dim(F) + \dim(F^{\perp})$
\end{description}
\end{symplperp}
\begin{proof}
\textbf{(i)} and \textbf{(ii)} follow directly from the definition. For \textbf{(iii)},
\begin{equation}\nonumber
\begin{split}
X \in (F+G)^{\perp} &\Leftrightarrow 0 = \omega(X, aY + bZ) = a\omega(X, Y)+b\omega(X, Z), \,\, \forall (a,b) \in \mathbb{R}^{2}, \forall Y\in F, \forall Z \in G  \\
&\Leftrightarrow \omega(X, Y) = \omega(X, Z) = 0, \,\, \forall Y \in F, \forall Z \in G \,\,\,\, \Leftrightarrow \,\,\,\, X \in F^{\perp}\cap G^{\perp},
\end{split}
\end{equation}
and analogously for \textbf{(iv)}. Finally, \textbf{(v)} follows from the nondegeneracy of $\omega$.
\end{proof}

We will make use of the following classification of subspaces with respect to the operation of symplectic complement:
\begin{sympltype}
A subspace $F \subset V$ is defined to be
\begin{description}
\item[(i)]\emph{isotropic} $\Leftrightarrow F \subset F^{\perp}$
\item[(ii)]\emph{coisotropic} $\Leftrightarrow F \supset F^{\perp}$
\item[(iii)]\emph{symplectic} $\Leftrightarrow F \cap F^{\perp} = 0$
\item[(iv)]\emph{Lagrangian} $\Leftrightarrow F$ is maximal isotropic, that is, $F$ is isotropic and $\nexists G \supsetneq F$ such that $G$ is an isotropic subspace of $V$.
\end{description}
\end{sympltype}
It follows that, if $V$ is finite dimensional, a subspace $F \subset V$ is Lagrangian if, and only if, it is isotropic and $\dim(F) = \frac{1}{2}\dim(V)$ or, equivalently, $F = F^{\perp}$. We also find the following
\begin{evenlagr}
If $(V, \omega)$ is a finite-dimensional symplectic vector space, then $V$ is even-dimensional and contains a Lagrangian subspace.
\end{evenlagr}
\begin{proof}
Take a one-dimensional subspace $F$ of $V$. Since $\omega$ is nondegenerate and antisymmetric: $F \subset F^{\perp} \Rightarrow \dim(F) \leq \dim(F^{\perp}) = \dim(V) - \dim(F) \Rightarrow \dim(F) \leq \frac{1}{2}\dim(V)$. If the equality is satisfied, we have the claim. Otherwise, substitute $F$ by the two-dimensional subspace spanned by $F\cup \{X\}$, where $X$ is a vector in $F^{\perp}-F$, which is also isotropic, and repeat the process recursively. This has to terminate, since $V$ is finite-dimensional, at which point we have an isotropic subspace of dimension $\frac{1}{2}\dim(V)$.
\end{proof}

Notice that $\mathbb{R}^{2n}$, with coordinates $(p_{1}, \dots, p_{n}, q^{1}, \dots, q^{n})$, is a symplectic vector space with the symplectic structure $\omega((p_{a}, q^{b}), (p_{c}, q^{d})) = \frac{1}{2}(p_{a}q'^{a} - p'_{b}q^{b})$. It turns out that, up to the choice of a basis, this is the most general finite-dimensional symplectic vector space:
\begin{symplframe}\label{symplbasis}
Let $(V, \omega)$ be a $2n$-dimensional symplectic vector space. Then $V$ has a basis (called a \emph{symplectic frame}) $\{X^{1}, X^{2}, \dots, X^{n}, Y_{1}, Y_{2}, \dots, Y_{n}\}$ such that
\begin{equation}\nonumber
\omega(X^{a}, X^{b}) = 0, \,\,\, 2\omega(X^{a}, Y_{b}) = \delta ^{a}_{b}, \,\,\, \omega(Y_{a}, Y_{b}) = 0, \,\,\, \forall a, b \in \{1, \dots, n\}.
\end{equation}
\end{symplframe}
\begin{proof}
Let $F$ be a Lagrangian subspace of $V$ and let $G$ be some other $n$-dimensional subspace such that $V = G \oplus F$. Then the map $G \to F^{*}: Z \mapsto 2\omega(Z, \cdot)$ is linear and injective since, if $Z \in G$ and $2\omega(Z, Y) = 0, \forall Y \in F \Rightarrow X \in G \cap F^{\perp} = G \cap F = \{0\}$. Therefore this identifies isomorphically $G = F^{*}$. Take a basis $\{ Y_{1}, \dots, Y_{n}\}$ in $F$ and let $\{ Z^{1}, \dots, Z^{n}\}$ be the dual basis in $G = F^{*}$. Then $\omega(Y_{a},Y_{b}) = 0$ ($F$ is Lagrangian) and $2\omega(Z^{a}, Y_{b}) = \delta^{a}_{b}$. Let $\lambda ^{ab} = \omega(Z^{a}, Z^{b})$ and $X^{a} = Z^{a} + \lambda ^{ab}Y_{b}$. Thus $\lambda ^{ab} = -\lambda ^{ba}$ and
\begin{equation}\nonumber
\begin{split}
\omega(X^{a}, X^{b}) = \omega&(Z^{a}+\lambda^{ac}Y_{c}, Z^{b}+\lambda^{bd}Y_{d}) = \lambda^{ab}-\frac{1}{2}\lambda^{ab}+\frac{1}{2}\lambda^{ba} = 0, \\
&2\omega(X^{a}, Y_{b}) = 2\omega(Z^{a}+\lambda^{ac}Y_{c}, Y_{b}) = \delta^{a}_{b},
\end{split}
\end{equation}
so that $\{X^{a}, Y_{b}\}$ is a symplectic frame.
\end{proof}
If we parametrize $V$ by writing an element $X \in V$ as $X = p_{a}X^{a} + q^{b}Y_{b}$, then $V$ is identified with $\mathbb{R}^{2n}$ and $\omega(X, X') = \omega( p_{a}X^{a} + q^{b}Y_{b},  p'_{c}X^{c} + q'^{d}Y_{d} = \frac{1}{2}(p_{a}q'^{a} - p'_{b}q^{b})$. The convention of upper and lower indices is useful to remind of the identification $V = Q^{*} \oplus Q$ which was used in the proof. We call the coordinates in a symplectic frame \emph{canonical coordinates}.

The nonlinear analogue of proposition \ref{symplbasis} is that any symplectic manifold is covered by charts of local canonical coordinates. The basic example here is the cotangent bundle $M = T^{*}Q$ of some manifold $Q$ (the configuration space). It has a natural symplectic structure, the \emph{canonical two-form}, which can be defined in a coordinate-free way by $\omega = d\theta$, where the \emph{canonical one-form} $\theta$ is given, at each $m = (p, q), \,\, p \in T^{*}_{q}Q, \,\, q \in Q$ by
\begin{equation}\nonumber
X \lrcorner \theta = (\pi _{*}X)\lrcorner p, \,\,\, X \in T_{m}M,
\end{equation}
where $\pi$ is the bundle projection $T^*Q\to Q$. Then, as $m$ varies, this defines a smooth 1-form on $M$. If we choose coordinates $q^{1}, \dots, q^{n}$ in $Q$, then $M$ has the chart $p_{1}, \dots, p_{n}, q^{1}, \dots, q^{n}$ given by expressing each $p \in T^{*}_{q}Q$ as $p = p_{a}dq^{a}$. In these coordinates, we see that $\theta$ and $\omega$ have the simple forms
\begin{equation}\nonumber
\theta = p_{a}dq^{a}, \,\,\,\,\,\, \omega = dp_{a}\wedge dq^{a}.
\end{equation}
Analogously, for an arbitrary manifold $(M, \omega)$, we will call \emph{local canonical coordinates} on the open neighbourhood $U \subset M$ a chart $p_{a}, q^{b}$ in which $\omega = dp_{a}\wedge dq^{a}$. As we have just seen, if $M = T^{*}Q$ then such coordinates exist by construction. That these always exist locally is a consequence of the Darboux-Weinstein theorem.
\begin{darboux}\label{darbweins}
Let $N$ be a submanifold of $M$ and let $\omega_{0}$ and $\omega_{1}$ be two non-degenerate closed two-forms on $M$ such that $\omega_{0}|_{N} = \omega_{1}|_{N}$. Then there exists a neighbourhood $U$ of $N$ and a diffeomorphism $f: U \to M$ such that $f(n) = n, \forall n \in N$ and $f^{*}\omega_{1} = \omega_{0}$.\footnote{For a proof of this formulation of the Darboux-Weinstein theorem, see chapter 4 of \cite{asymptotics}.}
\end{darboux}
It follows that, if we take $N$ to be a point $m$ in the symplectic manifold $(M, \omega)$, it is possible, by proposition \ref{symplbasis}, to choose coordinates $\{r_{a}, s^{b}\}$ in a neighbourhood of $N$ such that $\omega = dr_{a}\wedge ds^{a}$ at $m = N$. Application of theorem \ref{darbweins} with $\omega_{0} = \omega$ and $\omega_{1} = dr_{a}\wedge ds^{a}$ gives a diffeomorphism $f$ in a neighbourhood of $m$ which we use to define the local coordinates $p_{a} = r_{a} \circ f, \,\, q^{b} = s^{b} \circ f$ and the theorem guarantees that
\begin{equation}\nonumber
\omega = \omega_{0} = f^{*}\omega_{1} = f^{*}dr_{a}\wedge ds^{a} = d(r_{a} \circ f)\wedge d(s^{a} \circ f) = dp_{a}\wedge dq^{a},
\end{equation}
in a neighbourhood of $m$.

To see how this definition of local canonical coordinates relates to the one given in section \ref{mecrn}, suppose $U, V \subset M$ are intersecting open subsets of the $2n$-dimensional symplectic manifold $(M,\omega)$ which admit local canonical coordinates  $\boldsymbol{\zeta} = (p_{a}, q^{b})$ and $\boldsymbol{\eta} = (p'_{a}, q'^{b})$, respectively. Then, on $U \cap V$,
\begin{equation}\nonumber
\begin{split}
\frac{1}{2}d\boldsymbol{\zeta}^{T}\wedge \mathbf{J}d\boldsymbol{\zeta} &= dp_{a}\wedge dq^{a} = \omega = dp'_{a}\wedge dq'^{a} = \frac{1}{2}d\boldsymbol{\eta}^{T}\wedge \mathbf{J}d\boldsymbol{\eta}  \\
&= \frac{1}{2}d\boldsymbol{\zeta}^{T}\wedge \left[ \left(\frac{\partial \boldsymbol{\eta}}{\partial \boldsymbol{\zeta}}\right)^{T}\mathbf{J}\left(\frac{\partial \boldsymbol{\eta}}{\partial \boldsymbol{\zeta}}\right)\right] d\boldsymbol{\zeta} \,\, \Rightarrow \,\, \left(\frac{\partial \boldsymbol{\eta}}{\partial \boldsymbol{\zeta}}\right) \in SP(n, \mathbb{R}),
\end{split}
\end{equation}
where we used the notation $\mathbf{a}^{T}\wedge \mathbf{b} = a_{i}\wedge b^{i}$. So introducing the symplectic form induces a preferred choice of charts which cover the manifold $M$ in such a way that the transition functions are \emph{canonical transformations}.

The natural symmetries of symplectic manifolds are diffeomorphisms which preserve the symplectic structure.
\begin{symplectomorphism}
A \emph{symplectomorphism} between two symplectic manifolds $(M, \omega)$ and $(N, \sigma)$ is a diffeomorphism $\rho : M \to N$ that preserves the symplectic structure, i.e., $\rho^{*}\sigma = \omega$.
\end{symplectomorphism}

Also, we call \emph{symplectic automorphism} or \emph{canonical transformation} of $(M, \omega)$ a symplectomorphism from $(M, \omega)$ to itself. These form a group and its infinitesimal generators form a Lie subalgebra of the one in proposition \ref{lievector}.

\begin{localhamilton}\label{localderived}
 \emph{Locally Hamiltonian vector fields in $M$}, which we denote by, $V^{LH}$, are defined by 
\begin{equation}\nonumber
X \in V^{LH}(M) \Leftrightarrow \mathcal{L}_{X}\omega = 0,
\end{equation}
form a Lie algebra with respect to $[X, Y] = \mathcal{L}_{X}Y$. The \emph{Hamiltonian vector fields} $V^H(M)$, defined through Hamilton's equation as
\begin{equation}\label{hamiltoneq}
X_{h} \in V^{H}(M) \Leftrightarrow \exists h \in C^{\infty}(M) \,\,\, \text{such that} \,\,\, X_{h}\lrcorner \omega + dh = 0,
\end{equation}
form a Lie subalgebra of $V^{LH}(M)$, and the derived algebra $[V^{LH}(M), V^{LH}(M)]$ is contained in $V^H(M)$. 
\end{localhamilton}
\begin{proof}
If $X, Y \in V^{LH}(M)$, then
\begin{equation}\nonumber
\begin{split}
&\mathcal{L}_{aX+bY}\omega = a\mathcal{L}_{X}\omega + b\mathcal{L}_{Y}\omega = 0 \\
&\mathcal{L}_{[X, Y]}\omega = \mathcal{L}_{X}\mathcal{L}_{Y}\omega - \mathcal{L}_{Y}\mathcal{L}_{X}\omega = 0
\end{split}\,\,\,,
\end{equation}
so that $V^{LH}(M)$ is a subalgebra of $V(M)$. Also, if $X, Y \in V^{LH}(M)$, then $\mathcal{L}_{X}\omega = 0 = \mathcal{L}_{Y}\omega = d(Y\lrcorner \omega)$\footnote{We shall use very often that $\mathcal{L}_{X}\alpha = X\lrcorner d\alpha + d(X\lrcorner\alpha), \,\,\forall X\in V(M), \,\, \forall \alpha\in \Omega^{p}(M)$.}, so that
\begin{equation}\label{liepoisson}
\begin{split}
[X, Y]\lrcorner \omega &= (\mathcal{L}_{X}Y)\lrcorner\omega = \mathcal{L}_{X}(Y\lrcorner\omega) - Y\lrcorner\mathcal{L}_{X}\omega = X\lrcorner d(Y\lrcorner\omega) + d(X\lrcorner(Y\lrcorner\omega)) \\
&= -df, \,\,\,\,\,\,\,\,\, f = 2\omega(X, Y),
\end{split}
\end{equation}
and hence, from the linearity of (\ref{hamiltoneq}), $[V^{LH}(M), V^{LH}(M)] \subset V^{H}(M)$. That $V^H(M)$ is a subalgebra follows from the next result, proposition \ref{poissonalgebra}.
\end{proof}
\noindent Note that, if $M$ is simply connected ($H^1(M) = \{0\}$), $V^{LH}(M) = V^{H}(M)$ since in this case
\begin{equation}\nonumber
X \in V^{LH}(M) \Leftrightarrow 0 = \mathcal{L}_{X}\omega = d(X\lrcorner\omega) \Leftrightarrow X\lrcorner \omega = -dh, \,\, \text{for some} \,\, h \in C^{\infty}(M).
\end{equation}
Secondly, we remark that, in local canonical coordinates, $\omega = dp_{a}\wedge dq^{a}$ and $dh = \frac{\partial h}{\partial p_{a}}dp_{a} + \frac{\partial h}{\partial q^{b}}dq^{b}$, so that, if the tangent vector to a curve $(p(t), q(t))$ in $M$ satisfies (\ref{hamiltoneq}), then
\begin{equation}\nonumber
\dot{q}^{a}(t) = \frac{\partial h}{\partial p_{a}}, \,\,\,\,\,\,\,\,\,\,\,\, \dot{p}_{b}(t) = - \frac{\partial h}{\partial q^{b}},
\end{equation}
which is the standard form of Hamilton's equations.

A further consequence of the existence of a symplectic structure is that one can define the \emph{Poisson bracket} in a natural, coordinate-independent way.
\begin{poissonalg}\label{poissonalgebra}
The functions on a symplectic manifold $(M, \omega)$ form a Lie algebra with respect to the \emph{Poisson bracket}, defined by 
\begin{equation}\nonumber
f, g \in C^{\infty}(M) \Rightarrow \{f, g\} = X_{f}(g),
\end{equation}
where $X_{f}\lrcorner\omega + df = 0$. Furthermore, Hamilton's equation provides a Lie algebra isomorphism
\begin{equation}\label{hamiltonpoisson}
V^{H}(M) \isoeq C^{\infty}(M)/\mathbb{R},
\end{equation}
where $\mathbb{R}$ represents the constant functions on $M$.
\end{poissonalg}
\begin{proof}
Since $2\omega(X_{f}, X_{g}) = - X_{f}\lrcorner (X_{g}\lrcorner \omega) = X_{f}\lrcorner dg = \{f,g\}$, equation (\ref{liepoisson}) shows that $X_{\{f,g\}} = [X_{f},X_{g}]$. Also,
\begin{equation}\nonumber
\begin{split}
\{f,g\} &= X_{f}\lrcorner dg = -X_{f}\lrcorner(X_{g}\lrcorner \omega) = X_{g}\lrcorner(X_{f}\lrcorner \omega) = -X_{g}\lrcorner df = -\{g,f\}, \\
\{f,ag+bh\} &= X_{f}\lrcorner d(ag+bh) = aX_{f}\lrcorner dg + bX_{f}\lrcorner dh \\
&= a\{f,g\}+b\{f,h\},
\end{split}
\end{equation}
with $a,b\in\mathbb{R}$. For the Jacobi identity, note that
\begin{equation}\nonumber
\begin{split}
\{f,\{g,h\}\} &= X_{f}\lrcorner d(X_{g}\lrcorner dh) = -\mathcal{L}_{X_{f}}[X_{g}\lrcorner(X_{h}\lrcorner\omega)] \\
&= -[X_{f},X_{g}]\lrcorner(X_{h}\lrcorner\omega)+ [X_{f},X_{h}]\lrcorner(X_{g}\lrcorner\omega) = X_{\{f,g\}}\lrcorner dh - X_{\{f,h\}}\lrcorner dg \\
&= \{\{f,g\}h\} - \{\{f,h\}g\}
\end{split}
\end{equation}
Finally, note that the kernel of the homomorphism $f \mapsto X_{f}$ is given by the constants.
\end{proof}
As expected, in local canonical coordinates,
\begin{equation}\nonumber
\begin{split}
X_{f} &= \frac{\partial f}{\partial p_{a}}\frac{\partial}{\partial q^{a}}-\frac{\partial f}{\partial q^{b}}\frac{\partial}{\partial p_{b}}\\
\Rightarrow \{f,g\} &= X_{f}(g) =  \frac{\partial f}{\partial p_{a}}\frac{\partial g}{\partial q^{a}}-\frac{\partial f}{\partial q^{b}}\frac{\partial g}{\partial p_{b}}.
\end{split}
\end{equation}
Finally, we note that, because the time evolution of the system is given by the integral lines of $X_{h}$, for $h$ the energy Hamiltonian, the measured value of an obsearvable $f\in C^{\infty}(M)$ satisfies
\begin{equation}\nonumber
\frac{df}{dt}(p(t),q(t)) = X_{h}(f) = \{h, f\},
\end{equation}
which is the original expression (\ref{hamiltonobs}), but without any reference to a particular choice of coordinates.

\pagebreak

\section{Free Elementary Particles}\label{particles}

In this section we present the quantization of free elementary relativistic particles. We use this example both to motivate the general ideas of geometric quantization as well as to show how they are explicitly applied. The goal is simple: to find the Hilbert spaces of states of free relativistic particles, a problem which is typically presented in terms of representation theory. Geometric quantization however uses as an input a symplectic manifold associated to a classical system and so one is now led to ask, what is an elementary relativistic particle at the classical level? The necessary concept is that of an elementary system consisting of a symplectic manifold equipped with a transitive action of a symmetry group given by symplectic diffeomorphisms. The connection to representation theory is through Kirillov's orbit method \cite{kirillov1962unitary,kirillov2004lectures}. The physical interpretation and subsequent generalizations served as one of the key results for geometric quantization \cite{kostant1970quantization}. We also refer to the standard references \cite{symplectic,marsden,souriau1970strucure,simms1976lectures,woodhouse1997geometric,kirillov2001geometric} from which we drew most of the material of this section.

\subsection{Elementary systems}

We saw in the previous section that many of the main tools of classical mechanics could be cast into the language of symplectic geometry. In particular, the symmetries of classical mechanics, canonical transformations, correspond to symplectic diffeomorphisms of phase space. We now look at actions of a Lie algebra (and eventually Lie groups) on the phase space manifold as generators of such symmetries.

\begin{lieaction}
Let $(M, \omega)$ be a symplectic manifold and $\mathfrak{g}$ a real Lie Algebra. A \emph{canonical action} of $\mathfrak{g}$ on $(M,\omega)$ is a Lie algebra homomorphism $\mathfrak{g} \to V^{LH}(M)$.
\end{lieaction}

Where $V^{LH}$ are locally Hamiltonian vecotr fields as defined in the previous chapter.
This definition says that it is possible to find a subset of the infinitesimal canonical transformations which comes with the algebraic structure of a certain Lie algebra. A further step is to consider whether there are Hamiltonians which generate these particular transformations, ie. whether one can lift the homomorphism $\mathfrak{g} \to V^{LH}$ to the Poisson algebra $C^{\infty}(M)$.

\begin{moment}
If, for the canonical action $\mathfrak{g} \to V^{LH}(M): A \mapsto X_{A}$, it is possible to construct a linear map $\mathfrak{g}\to C^{\infty}(M):A\mapsto h_{A}$ such that, $\forall A, B \in \mathfrak{g}$,
\begin{description}
\item[(i)]$X_{A}\lrcorner\omega + dh_{A} = 0$,
\item[(ii)]$h_{[A,B]} = \{h_{A},h_{B}\}$,
\end{description}
then the map $A\mapsto h_{A}$ is said to be a \emph{Hamiltonian} for the action $A\mapsto X_{A}$, and the dual map
\begin{equation}\nonumber
\mu : M \to \mathfrak{g}^{*}: m\mapsto f_{m},
\end{equation}
where $f_{m}(A) = h_{A}(m)$, is called a \emph{moment} for the action.
\end{moment}
An action which admits a Hamiltonian is often called a \emph{Hamiltonian action} or a moment map. One of the reasons why the moment map is interesting is Noether's theorem:
\begin{noether}
Let $\mathfrak{g} \to V^{LH}(M): A \mapsto X_{A}$ be a Hamiltonian action of $\mathfrak{g}$ on $(M, \omega)$ and $\mu : M \to \mathfrak{g}^{*}: m\mapsto f_{m}$ its moment map. If, additionally, the action preserves the Hamiltonian of the system, then the moments are constant. In other words,
\begin{equation}\nonumber
X_{A}(h) = 0, \,\,\forall A \in \mathfrak{g}\,\,\,\, \Rightarrow \,\,\,\, \mu \circ \rho _{t} = \mu, \,\,\forall t \in \mathbb{R},
\end{equation}
where $\rho _{t}$ is the flow of $X_{h}$.
\end{noether}
\begin{proof}
At any point $m \in M$, the linear functional $\mu \circ \rho _{t}$ is given by
\begin{equation}\nonumber
(\mu \circ \rho _{t})(m)(A) = [\mu (\rho _{t}(m))](A) = f_{\rho _{t}(m)}(A) = h_{A}(\rho _{t}(m)), A \in \mathfrak{g},
\end{equation}
which is constant since the flow of $X_{h}$ preserves the $h_{A}$'s. Indeed,
\begin{equation}\nonumber
\begin{split}
\frac{d}{dt}[(\mu \circ\rho _{t})(m)(A)] &= \frac{d}{dt}[(h_{A}(\rho_{t}(m))] = [X_{h}(h_{A})](\rho_{t}(m)) = -\{h_{A},h\}(\rho_{t}(m)) \\
&= -[X_{A}(h)](\rho_{t}(m)) = 0, \,\,\forall A \in \mathfrak{g},\,\,\forall m \in M,\,\,\forall t \in \mathbb{R}.
\end{split}
\end{equation}
Finally, since $\mu\circ\rho_{0} = \mu$, solving the above first-order ODE gives $\mu\circ\rho_{t} = \mu, \forall t$.
\end{proof}

In the cases we will be interested in, the canonical action of the Lie algebra is the infinitesimal form of an action of the corresponding Lie group where that Lie group is generated by finite canonical transformations. Now we can make precise what we mean by a classical elementary system: a phase space with the action of the symmetry group $G$ given by canonical transformations under which all states are equivalent.
\begin{elementary}
Let $G$ be a Lie group with Lie algebra $\mathfrak{g}$ which acts (on the right) on a symplectic manifold $(M,\omega)$ by symplectomorphisms. In other words, each $g \in G$ determines a diffeomorphism $g:M\to M:m\mapsto mg$ such that
\begin{equation}\nonumber
g^{*}\omega = \omega, \,\,\, (mg)g' = m(gg'), \,\, \forall g, g' \in G, \,\,\forall m \in M.
\end{equation}
We then call the action \emph{Hamiltonian} whenever the corresponding infinitesimal action of $\mathfrak{g}$ is Hamiltonian and, in this case, the corresponding moment map $\mu:M\to\mathfrak{g}^{*}$ is said to be a \emph{moment} for the action of $G$. If, in addition, the action of $G$ is \emph{transitive}, $(M,\omega)$ is an \emph{elementary system} with symmetry $G$.
\end{elementary}

We first show that elementary systems with Lie group symmetry exist by giving a concrete realization of them as orbits of the coadjoint action of the Lie group in the dual of its Lie algebra. From now on, we will always assume that the Lie group $G$ is connected unless explicitly stated otherwise. Consider the adjoint action of $G$ on itself, given by
\begin{equation}\nonumber
Ad : G\times G \rightarrow G : (g,h) \mapsto Ad_{g}(h):=ghg^{-1}.
\end{equation}
Notice that, for any $g \in G$, the action $Ad_{g}$ fixes the identity $e$. Hence its derivative at $e$,
\begin{equation}\label{adjoint}
ad : G\times \mathfrak{g} \rightarrow \mathfrak{g} : (g,A) \mapsto ad_{g}A:=\frac{d}{dt}(ge^{tA}g^{-1})\Big|_{t=0},
\end{equation}
 is a linear representation of the group on the vector space of its Lie algebra, which we will also refer to as the adjoint action. Here, $e^{A}$ denotes the exponential map in the group manifold. Indeed, neglecting terms of order higher than one in $t$ on the Baker-Campbell-Hausdorff formula (which cancel on evaluation at t=0),
\begin{equation}\nonumber
\begin{split}
ad_{g}(xX + yY) &= \frac{d}{dt}(ge^{txX}gg^{-1}e^{tyY}g^{-1})\Big|_{t=0} = x\frac{d}{dt}(ge^{tX}g^{-1})\Big|_{t=0} + y\frac{d}{dt}(ge^{tY}g^{-1})\Big|_{t=0} \\
&= xad_{g}X + yad_{g}Y.
\end{split}
\end{equation}
The coadjoint action is the dual representation
\begin{equation}\label{coadjoint}
ad^{*} : G\times \mathfrak{g}^{*} \rightarrow \mathfrak{g}^{*} : (g,f) \mapsto ad^{*}_{g}f, \,\,\, \text{where} \,\,\, ad^{*}_{g}f(A) = f(ad_{g}A).
\end{equation}
This action splits $\mathfrak{g}^{*}$ in orbits of the form $M = \{ad^{*}_{g}f|g \in G\}$. On each one, the action of $G$ is obviously transitive. Moreover, since we assume that $G$ is connected, each orbit is also connected and, because all the elements of $M$ are related by the action of $G$, the tangent space at any $f$, $T_{f}M$, is spanned by the generators of the coadjoint action in $M$. The explicit form of these vectors is given by the infinitesimal form of equations (\ref{adjoint}) and (\ref{coadjoint}):
\begin{equation}\nonumber
\begin{split}
[ad^{*}_{e^{tA}}f](B) &= f(ad_{e^{tA}}B)=f\left[\frac{d}{ds}(e^{tA}e^{sB}e^{-tA})\Big|_{s=0}\right]\\
&=f\left[\frac{d}{ds}(e^{tA}e^{sB}e^{-tA})\Big|_{s=t=0}+t\frac{d}{dt}\frac{d}{ds}(e^{tA}e^{sB}e^{-tA})\Big|_{s=t=0}+O(t^{2})\right]\\
&=f(B+t[A,B]+O(t^{2}))=[f+tX_{A}|_{f}+O(t^{2})](B), \,\,\forall B \in \mathfrak{g},
\end{split}
\end{equation}
where in the last line we have used that $\mathfrak{g}^{*}$ is a vector space identify $T_{f}\mathfrak{g}^{*}\sim\mathfrak{g}^{*}$ such that $X_{A}|_{f}(B)=f([A,B])$. We can therefore fully define a differential form on $M$ by giving its values on the vectors $X_{A}|_{f}$ at each point $f \in M$.
\begin{coadorbit}
The two-form $\omega$ which, at each $f\in M$, is given by
\begin{equation}\label{symplorbit}
\omega(X_{A}|_{f},X_{B}|_{f})=\frac{1}{2}f([A,B]),
\end{equation}
is a well-defined symplectic structure on $M = \{ad^{*}_{g}f|g\in G\}$ and it is invariant under the coadjoint action of $G$.
\end{coadorbit}
\begin{proof}
That the definition gives a well-defined two-form on $M$, in the sense that $\omega(X_{A},X_{B})|_{f}$ depends only on the value of the $X_{A}$ fields at the point $f$, is seen from
\begin{equation}\nonumber
\begin{split}
X_{A}|_{f}=X_{A'}|_{f} \,\,\,\,\Rightarrow \,\,\,\,\omega(X_{A},X_{B})|_{f}-\omega(X_{A'},X_{B})|_{f} &= \frac{1}{2}f([A,B])-\frac{1}{2}f([A',B]) \\
&= \frac{1}{2}(X_{A}|_{f}-X_{A'}|_{f})(B) = 0, \,\, \forall B \in \mathfrak{g}.
\end{split}
\end{equation}
To prove that it is also invariant under the coadjoint action, we define, for each $A \in \mathfrak{g}$, the function $h_{A}:M \rightarrow \mathbb{R}:f\mapsto h_{A}(f):=f(A)$. Then
\begin{equation}\nonumber
\begin{split}
(X_{A}\lrcorner dh_{B})(f) = \frac{d}{dt}[h_{B}(f+tX_{A}|_{f})]\Big|_{t=0} &=\frac{d}{dt}[(f+tX_{A}|_{f})(B)]\Big|_{t=0} = f([A,B]) = h_{[A,B]}(f) \\
&\Rightarrow \,\,\,\, X_{A}\lrcorner dh_{B} = h_{[A,B]},
\end{split}
\end{equation}
and, since $2\omega(X_{B},X_{A})(f) = -f([A,B]) = -h_{[A,B]}(f),\,\, \forall f$, this gives
\begin{equation}\nonumber
X_{A}\lrcorner (X_{B}\lrcorner\omega + dh_{B}) = 2\omega(X_{B},X_{A}) + h_{[A,B]} = 0, \,\,\forall A,B \in \mathfrak{g} \,\,\,\, \Rightarrow\,\,\,\, X_{B}\lrcorner\omega + dh_{B} = 0, \,\,\forall B\in \mathfrak{g},
\end{equation}
where the last implication comes from using once again that the $X_{A}|_{f}$'s span $T_{f}M$. Hence
\begin{equation}\nonumber
0 = \mathcal{L}_{X_{A}}(X_{B}\lrcorner\omega + dh_{B}) = X_{B}\lrcorner\mathcal{L}_{X_{A}}\omega + [X_{A},X_{B}]\lrcorner\omega + d(X_{A}\lrcorner dh_{B}) = X_{B}\lrcorner\mathcal{L}_{X_{A}}\omega, \,\,\forall A,B\in\mathfrak{g},
\end{equation}
since $[X_{A},X_{B}] = X_{[A,B]}$. It follows that $\omega$ is invariant under the flows of the $X_{A}$'s and thus under the coadjoint action.

To show that this form is a symplectic structure, we point out that
\begin{equation}\nonumber
\omega(X_{A},X_{B})|_{f} = 0, \,\, \forall X_{B}|_{f} \,\,\,\, \Leftrightarrow \,\,\,\, f([A,B]) = 0, \,\, \forall B\in\mathfrak{g} \,\,\,\,\Leftrightarrow\,\,\,\, X_{A}|_{f} = 0,
\end{equation}
so it is nondegenerate. Finally, closure comes from
\begin{equation}\nonumber
0 = \mathcal{L}_{X_{A}}\omega = d(X_{A}\lrcorner\omega) + X_{A}\lrcorner d\omega = -d^{2}h_{A} + X_{A}\lrcorner d\omega = X_{A}\lrcorner d\omega, \,\,\forall A\in\mathfrak{g}.
\end{equation}
\end{proof}
Notice that actually the map $\mathfrak{g} \rightarrow C^{\infty}(M): A \mapsto h_{A}$ used in the proof is itself a Hamiltonian for the action. In summary, we have that
\begin{coad}
Each orbit $M$ of the coadjoint action of a Lie group $G$ on the dual of its Lie algebra $\mathfrak{g}$ is an elementary system with symmetry $G$ and with moment map given by the inclusion $M\hookrightarrow\mathfrak{g}^{*}$.
\end{coad}

In fact, in the cases we will be interested in, these are the \emph{only} elementary systems.
\begin{classelem}
Let $(M',\omega ')$ be an elementary system with symmetry $G$ and let $\mu : M' \rightarrow\mathfrak{g}^{*}$ be a moment. Then $\mu$ is a surjective local symplectomorphism from $(M',\omega ')$ onto a coadjoint orbit $(M,\omega)\subset\mathfrak{g}^{*}$.
\end{classelem}
\begin{proof}
Let $A\mapsto X_{A}'$ and $A\mapsto X_{A}$ denote the actions in $M'$ and $M$, respectivelly, and likewise for the Hamiltonians $h_{A}$ and $h_{A}' = h_{A}\circ\mu$.

Since $\mu$ is a moment, $X_{A}'\lrcorner dh_{B}' = h_{[A,B]}', \,\, \forall A,B\in\mathfrak{g}$. Taking this into account, we have
\begin{equation}\nonumber
\begin{split}
[(X_{A}'\lrcorner d\mu)(m')](B) &= \frac{d}{dt}h_{B}'(m'+tX_{A}')\Big|_{t=0} = (X_{A}'\lrcorner dh_{B}')(m') = h_{[A,B]}'(m') \\
&= h_{[A,B]}(\mu(m')) = [\mu(m')]([A,B]) = [X_{A}|_{\mu(m')}](B) \\
&= [(X_{A}\circ\mu)(m')](B), \,\, \forall m'\in M', \,\, \forall B\in \mathfrak{g}\\
\Rightarrow & \,\,\,\, \,\,\,\, X_{A}'\lrcorner d\mu = X_{A}\circ \mu,
\end{split}
\end{equation}
so that we have, at any $m'\in M'$,
\begin{equation}\nonumber
\mu_{*}(X_{A}'|_{m'}) = \frac{d}{dt}\mu(m+tX_{A}')\Big|_{t=0} = (X_{A}'\lrcorner d\mu)(m') = X_{A}|_{\mu(m')}.
\end{equation}
One sees that $\mu_{*}X_{A}' = X_{A}, \,\, \forall A\in \mathfrak{g}$ and thus, since $G$ is connected, $\mu(m'g) = ad^{*}_{g}\mu(m'), \,\, \forall m'\in M', \forall g\in G$.

Furthermore, the action of $G$ on $M'$ is transitive, so $\mu(M') = \mu(\{m'g|g\in G\}) = \{ad^{*}_{g}\mu(m')|g\in G\}$, which is a coadjoint orbit. Also because the action on $M'$ is transitive, the vector fields $X_{A}'$ span the tangent spaces at each point of $M'$, and hence
\begin{equation}\nonumber
2\omega '(X_{A}', X_{B}') = h_{[A,B]}' = h_{[A,B]}\circ\mu = 2\omega(X_{A},X_{B})\circ\mu
\end{equation}
implies that $\omega ' = \mu^{*}\omega$. Now both $\omega$ and $\omega '$ are nondegenerate, so $\mu_{*}$ is injective at each point and therefore a surjective local symplectic diffeomorphism, as claimed.

\end{proof}

This gives a correspondence between elementary systems \emph{with a momentum map} and coadjoint orbits. It turns out that for some Lie algebras of central interest to physics the momentum map for a given elementary system is unique.

\begin{bracohom}\label{corsesi}
Let $\mathfrak{g}$ be a Lie algebra such that $[\mathfrak{g},\mathfrak{g}] = \mathfrak{g}$ and $H^{2}\mathfrak{g} = 0$. Then every canonical action of $\mathfrak{g}$ is Hamiltonian and has a unique moment.
\end{bracohom}
To avoid a digression into the cohomology of Lie algebras \cite{chevalley1948cohomology} we refer to \cite{woodhouse1997geometric} for a proof of this result. In particular, in the case of semisimple Lie algebras, Whitehead's lemmas imply that the above conditions hold \cite{whitehead1937certain}. We see therefore that in cases like that of $SO(3)$ or Lorentz symmetry, there exists a unique way of assigning classical observables to the generators of the symmetry transformations. The moment map thus recovers the classical and relativistic angular momenta unambiguously. More than that, the constructed moment map associates in a unique way a coadjoint orbit to each elementary system with $SO(3)$ or Lorentz symmetry. This is a classification of the classical phase space analogues of elementary particles.

Aside from knowing that the phase spaces of elementary particles are coadjoint orbits, we want to explicitly construct them. A useful result is the following reduction procedure \cite{guillemin1982moments}.

\begin{presympl}\label{pres}
Let $C$ be a smooth manifold and let $\sigma$ be a closed two-form of constant rank on $C$, that is, such that the dimension of
\begin{equation}\nonumber
K_{m} = \{X|X\lrcorner\sigma = 0\}\subset T_{m}C
\end{equation}
is constant as $m$ varies over $C$. Then $K$ is a distribution\footnote{A real distribution on $M$ is a sub-bundle of $TM$.} on $C$, which we will call the \emph{characteristic distribution} of $\sigma$ and $(C,\sigma)$ is said to be a \emph{presymplectic manifold}.
\end{presympl}

Note now that the identity
\begin{equation}\label{a116}
\begin{split}
d\alpha(X_{i},X_{j},...,X_{m}) = X_{[i}(\alpha(X_{j},X_{k},...,X_{m]})-\frac{p}{2}&\alpha([X_{[i},X_{j}],X_{k},...,X_{m]}), \,\,\\ &\forall \alpha\in\Omega^{p}(M), \,\, \forall X_{i},X_{j},...\in V(M)
\end{split}
\end{equation}
(square-bracketed indices are antisymmetrized) implies that, for $X,Y\in V_{K}(C)$ and any $Z\in V(C)$,
\begin{equation}\nonumber
\sigma ([X,Y],Z) = -3d\sigma (X,Y,Z) = 0 \,\,\,\,\Rightarrow\,\,\,\, [X,Y]\in V_{K}(C),
\end{equation}
since $\sigma$ is closed. Therefore $K$ is an integrable distribution\footnote{This is Frobenius' theorem: if a distribution $K$ satisfies $[X,Y]\in K, \,\, \forall X,Y\in K$, then it is integrable. An integrable distribution is also called a foliation.}. We will call the presymplectic manifold $(C,\sigma)$ $reducible$ if its characteristic foliation $K$ is reducible\footnote{If K is a foliation on $C$, there are sumanifolds on $C$ (the leaves of $K$) whose tangent bundles are given by $K$. $C/K$ is the space of leaves. If this is Hausdorff, the foliation is said to be reducible.}. In this case, the space of leaves $C/K$ is a Hausdorff manifold and $\sigma$ projects to a well defined symplectic structure $\underline{\omega}$ on $C/K$. Indeed, closure follows from the closure of $\sigma$ and nondegeneracy from the definition of $C/K$. Finally, for any $X\in V_{K}(C)$,
\begin{equation}\nonumber
X\lrcorner\sigma = 0 = X\lrcorner d\sigma.
\end{equation}
We call the symplectic manifold $(C/K, \underline{\omega})$ the \emph{reduction} of $(C, \sigma)$. In our discussion, this construction will enter through the following proposition.

\begin{reduction}\label{reductionorbit}
Let $(M,\omega)$ be a coadjoint orbit in the Lie algebra $\mathfrak{g}^{*}$ of some Lie group and let $C$ be a manifold with: an action $\mathfrak{g}\rightarrow V(C): A\mapsto X_{A}'$, a surjection $\pi :C\rightarrow M$ and a 1-form $\theta '\in\Omega ^{1}(C)$ such that
\begin{description}
\item[(i)] $\pi ^{-1}(m)$ is connected for each $m\in M$
\item[(ii)] $\pi _{*}X_{A}' = X_{A}$
\item[(iii)] for each $m'\in C$, $X_{A}'\lrcorner\theta '(m') = [\pi (m')](A)$ (remember that $\pi (m') \in M\subset \mathfrak{g}^{*}$)
\end{description}
Then there is a symplectic diffeomorphism between $(M,\omega)$ and the reduction of $(C,d\theta ')$.
\end{reduction}
\begin{proof}
Let $\phi : C\to C/K$ denote the reduction map, i.e., $\phi$ maps $m'\in C$ to the leaf of $K$ through $m'$, where $K$ is the characteristic foliation of $d\theta '$. We know that $(C/K, \underline{\omega})$, where $\underline{\omega}$ is the projection of $d\theta '$, is a symplectic manifold. We first show that the action of $\mathfrak{g}$ on $C$ projects to a Hamiltonian action on $(C/K, \underline{\omega})$. For this, recall the definition of the symplectic structure on $M$ from (\ref{symplorbit}). We have
\begin{equation}\nonumber
\begin{split}
[X_{B}'\lrcorner (X_{A}'\lrcorner d\theta ')](m') &= [X_{B}'\lrcorner (\mathcal{L}_{X_{A}'}\theta ' - d(X_{A}'\lrcorner \theta '))](m') \\
&= [X_{A}'(X_{B}'\lrcorner\theta ')](m') - ([X_{A}',X_{B}']\lrcorner\theta ')(m') - [X_{B}'(X_{A}'\lrcorner\theta ')](m') \\
&= X_{A}'([\pi (m')](B)) - [\pi (m')]([A,B]) - X_{B}'([\pi (m')](A)) \\
&= [\pi (m')]([A,B]) = [X_{B}\lrcorner (X_{A}\lrcorner\omega)](\pi (m')) \\
\Rightarrow d\theta ' = \pi^{*}\omega 
\end{split}
\end{equation}
on $\text{span}\{X_{A}'\}$. It follows that $\theta '$ is invariant under the action of $\mathfrak{g}$:
\begin{equation}\nonumber
\begin{split}
(\mathcal{L}_{X_{A}'}\theta ')(m') &= (X_{A}'\lrcorner d\theta ')(m') + [d(X_{A}'\lrcorner \theta ')](m') = [\pi^{*}(X_{A}\lrcorner\omega)](m') + d(h_{A}[\pi(m')]) \\
&= [\pi^{*}(X_{A}\lrcorner\omega + dh_{A})](m') = 0,
\end{split}
\end{equation}
where $h_{A}$ is the Hamiltonian in $M$. Now, for any $Y\in V_{K}(C)$ and any $A\in\mathfrak{g}$,
\begin{equation}\nonumber
[X_{A}',Y']\lrcorner d\theta ' = \mathcal{L}_{X_{A}'}(Y'\lrcorner d\theta ') - Y'\lrcorner d(\mathcal{L}_{X_{A}'}\theta ') = 0 \,\,\,\,\Rightarrow \,\,\,\, [X_{A}',Y']\in V_{K}(C),
\end{equation}
so that $X_{A}'$ projects to a well defined $\underline{X_{A}} = \phi_{*}X_{A}' \in V(C/K)$. Clearly $[\underline{X_{A}},\underline{X_{B}}] = \underline{X_{[A,B]}}$. Consider now the function $h_{A}' = X_{A}'\lrcorner \theta ' \in C^{\infty}(C)$. For any $Y'\in V_{K}(C)$,
\begin{equation}\nonumber
Y'(h_{A}') = Y'\lrcorner d(X_{A}'\lrcorner \theta ') = - Y'\lrcorner (X_{A}'\lrcorner d\theta ') = 0,
\end{equation}
since $Y'\lrcorner d\theta ' = 0$. Therefore, $h_{A}' = \underline{h_{A}}\circ\phi$ for some well-defined $\underline{h_{A}}\in C^{\infty}(C/K)$. Moreover, $\underline{h_{A}}$ generates $\underline{X_{A}}$:
\begin{equation}\label{invpotential1}
0 = \mathcal{L}_{X_{A}'}\theta ' = X_{A}'\lrcorner d\theta ' + d(X_{A}'\lrcorner\theta ') = \phi^{*}(\underline{X_{A}}\lrcorner\underline{\omega} + d\underline{h_{A}}),
\end{equation}
and,
\begin{equation}\label{invpotential2}
[\underline{h_{A}},\underline{h_{B}}]\circ\phi = \phi^{*}[\underline{X_{A}}(\underline{h_{B}})] = X_{B}'\lrcorner\mathcal{L}_{X_{A}'}\theta ' + [X_{A}',X_{B}']\lrcorner\theta ' = h_{[A,B]}' = \underline{h_{[A,B]}}\circ\phi,
\end{equation}
so $\underline{h_{A}}$ is a Hamiltonian. Hence there is a moment $\mu : C/K \rightarrow \mathfrak{g}^{*}$. To complete the proof, note that, by \textbf{(iii)},
\begin{equation}\nonumber
\begin{split}
[\mu(\phi(m'))](A) = \underline{h_{A}}(\phi(m')) = & h_{A}'(m') = (X_{A}'\lrcorner\theta ')(m') = [\pi(m')](A), \,\,  \forall A\in \mathfrak{g}, \,\,\forall m' \in C \\
& \Rightarrow\,\,\,\, \mu\circ\phi = \pi .
\end{split}
\end{equation}
In other words, the diagram
\begin{equation}\nonumber
\begin{tikzcd}
& (C,d\theta ') \arrow[dl, "\phi"'] \arrow[dr, "\pi"] & \\
(C/K, \underline{\omega}) \arrow[rr, "\mu"'] & & (M,\omega)
\end{tikzcd}
\end{equation}
commutes. Finally, $\mu$ is a symplectic diffeomorphism as a consequence of \textbf{(i)}, of $ker(\phi_{*}|_{m'}) = ker(\pi_{*}|_{m'}), \,\, \forall m'\in C$, and of
\begin{equation}\nonumber
\phi^{*}\underline{\omega} = d\theta ' = \pi^{*}\omega = (\mu\circ\phi )^{*}\omega = \phi^{*}(\mu^{*}\omega) \Rightarrow \underline{\omega} = \mu^{*}\omega.
\end{equation}
\end{proof}

\begin{su2}\emph{(Rotational symmetry)}\label{sphereso3}

We finally apply these results to some examples. Before jumping into the case of Poincar\'e symmetry, however, it is instructive to see how the above strategy works in the case of rotational symmetry. 

To start with, note that proposition \ref{corsesi} implies that any canonical action of $so(3)$ has associated a unique moment. Let us check what this moment is for the action of the group $SO(3)$ on $(T^{*}\mathbb{R}^{3} \simeq \mathbb{R}^{6}, d\theta)$ by rotations, where $\theta$ is the canonical one-form. First let $(p_{a},q^{b})$ be coordinates on a symplectic frame. Then an element $g\in SO(3)$ acts by $\mathbf{q} = (q^{a})\mapsto \mathbf{gq} = (\tensor{g}{^a_b}q^{b})$ and $\mathbf{p} = (p^{a})\mapsto \mathbf{gp} = (\tensor{g}{^a_b}p^{b})$, where $\mathbf{q}$ and $\mathbf{p}$ are column vectors (we raise and lower indices with the identity and employ the Einstein summation convention). Note that this leaves the canonical one-form $\theta$ invariant:
\begin{equation}\nonumber
\mathbf{p}^{T}d\mathbf{q}\mapsto (\mathbf{gp})^{T}d(\mathbf{gq}) = \mathbf{p}^{T}\mathbf{g}^{T}\mathbf{g}d\mathbf{q} = \mathbf{p}^{T}d\mathbf{q} .
\end{equation}
Then, because $\theta$ is also a symplectic potential potential (i.e., $\omega=d\theta$), calculations analogous to (\ref{invpotential1}) and (\ref{invpotential2}) show that $h_{A} = X_{A}\lrcorner\theta, \,\, \forall A\in so(3)$ is the Hamiltonian. The Lie algebra $so(3)$ is spanned by the $3\times 3$ matrices $L^{i}, \,\, i\in \{1,2,3\}$, where $(L^{i})^{jk} = -\epsilon^{ijk}$ \cite{hall2015lie}. Notice that this basis satisfies $[L^{i},L^{j}] = \tensor{\epsilon}{^i^j_k}L^{k}$. To find the explicit form of the vectors $X_{A}$, note that the action of $SO(3)\ni e^{tA}, \,\, A \in so(3)$ in the $((\mathbf{p}^{T})_{a},(\mathbf{q})^{b})$ coordinates can be written as
\begin{equation}\nonumber
e^{tA}(\mathbf{p}^{T},\mathbf{q}) = ((e^{t\mathbf{A}}\mathbf{p})^{T},e^{t\mathbf{A}}\mathbf{q}) = (\mathbf{p}^{T},\mathbf{q}) + t(-(\mathbf{p})^{T}\mathbf{A},\mathbf{Aq}) + O(t^2),
\end{equation}
which we recognise as the flow of
\begin{equation}\nonumber
X_{A} = -p_{b}\tensor{A}{^b_a}\frac{\partial}{\partial p_{a}} + \tensor{A}{^a_b}q^{b}\frac{\partial}{\partial q^{a}}.
\end{equation}
Hence,
\begin{equation}\nonumber
h_{A} = \left(-p_{b}\tensor{A}{^b_a}\frac{\partial}{\partial p_{a}} + \tensor{A}{^a_b}q^{b}\frac{\partial}{\partial q^{a}}\right)\lrcorner p_{c}dq^{c} = \tensor{A}{^a_b}q^{b}p_{a}.
\end{equation}
In particular,
\begin{equation}\nonumber
h_{(L^{i})} = \tensor{(L^{i})}{^a_b}q^{b}p_{a} = \epsilon_{iab}q^{a}p^{b} = (\mathbf{q}\times\mathbf{p})_{i},
\end{equation}
thus the momentum map, in the specific case of rotational symmetry in $T^*\mathbb{R}^3$, recovers uniquely the familiar expression for the angular momentum.

We turn now to the construction of the elementary systems with symmetry $SO(3)$. First notice that, for any $g\in SO(3)$, $det(\mathbf{g}) = 1$ and $\mathbf{g}^{T}\mathbf{g} = \mathbf{1}$ imply, by Cramer's formula for the inverse,
\begin{equation}\nonumber
\begin{split}
\tensor{(g)}{_k^l} = \tensor{(g^{-1})}{^l_k} = \frac{1}{(3-1)!}\epsilon_{ki_{2}i_{3}}\epsilon^{lj_{2}j_{3}}\tensor{g}{^{i_{2}}_{j_{2}}}\tensor{g}{^{i_{3}}_{j_{3}}} \\
\Rightarrow \epsilon^{kmn}\tensor{g}{_k^l} = \frac{1}{2}\epsilon^{kmn}\epsilon_{ki_{2}i_{3}}\epsilon^{lj_{2}j_{3}}\tensor{g}{^{i_{2}}_{j_{2}}}\tensor{g}{^{i_{3}}_{j_{3}}} = \epsilon^{lj_{2}j_{3}}\tensor{g}{^m_{j_{2}}}\tensor{g}{^n_{j_{3}}}.
\end{split}
\end{equation}
We employ this in the following: let us map $so(3) \ni A = a_{i}L^{i}\mapsto (a_{i})\in\mathbb{R}^{3}$. Then, under this identification, the adjoint action looks like
\begin{equation}\nonumber
\begin{split}
(ad_{g}(a_{i}L^{i}))^{jk} &= \left(\frac{d}{dt}[ge^{t(a_{i}L^{i})}g^{-1}]\Big|_{t=0} \right)^{jk}= (g(a_{i}L^{i})g^{T})^{jk} = -\tensor{g}{^j_m}(a_{i}\epsilon^{imn})\tensor{g}{^k_n} \\
&= -a_{i}\epsilon^{ljk}\tensor{g}{_l^i} = ((\mathbf{ga})_{l}L^{l})^{jk}\\
& \,\,\,\, \Rightarrow  \,\,\,\, ad_{g}(a_{i}L^{i}) = (\mathbf{ga})_{i}L^{i},
\end{split}
\end{equation}
so that this choice of basis identifies the adjoint action on $so(3)$ with the action on $\mathbb{R}^{3}$ by rotations that we have just discussed. By introducing a dual basis in $so(3)^{*}$, we see that the coadjoint action can be pictured in the same way. It becomes clear that this mapping sends the coadjoint orbits in $so(3)^{*}$ to spheres centered at the origin in $\mathbb{R}^{3}$. Indeed, the sphere is a standard example of a symplectic manifold, the symplectic form given by the volume form divided by its radius, and the $SO(3)$-action is canonical and transitive over it.

From the physical point of view, we will see shortly that the sphere of radius $s$ is the classical phase space for the rotational degrees of freedom of an elementary particle of spin $s$. It will then be instructive to first apply geometric quantization to the sphere. With this goal in mind, we now use its interpretation as a coadjoint orbit in $so(3)^{*}$ to reconstruct it as a reduction by using proposition \ref{reductionorbit}. The reason to do this is that we will end up with a parametrization of the sphere in complex coordinates which provides a natural choice of a \emph{polarization}, a necessary ingredient for quantization (see subsection \ref{subsecquant}).

First we point out that the coadjoint orbits of the action of $G$ on $\mathfrak{g}^{*}$ are completely determined by the corresponding Hamiltonian action of $\mathfrak{g}$ on $\mathfrak{g}^{*}$, as far as our definitions go. This was already hinted at, for example, by the way in which we defined the symplectic structure on the orbits in equation (\ref{symplorbit}). An implication is that, if two Lie groups have isomorphic Lie algebras, the coadjoint orbits are the same, even though the groups themselves might not be isomorphic. Therefore, recall that the Lie algebra of $SU(2)$, which is given by the $2\times 2$ anti-hermitian\footnote{We use the mathematicians' convention that the exponential map from $Lie(G)$ to $G$ is $e^{tA}$. If we used the physicists convention $e^{itA}$, the Lie algebra of $SU(2)$ would consist of \emph{hermitian} matrices.} matrices, is generated by $\{i\sigma^{j}\}$, where $\sigma^{j}$ are the Pauli matrices

\begin{equation}\nonumber
\sigma^{1} = \left( \begin{array}{cc}
0 & 1 \\
1 & 0
\end{array} \right), \qquad
\sigma^{2} = \left( \begin{array}{cc}
0 & -i \\
i & 0
\end{array} \right), \qquad
\sigma^{3} = \left( \begin{array}{cc}
1 & 0 \\
0 & -1
\end{array} \right).
\end{equation}

It is straightforward to verify that 
\begin{equation}\label{sigmaproduct}
\sigma^{i}\sigma^{j} = \delta^{ij}\mathbf{1} + i\tensor{\epsilon}{^i^j_k}\sigma^{k} \qquad \text{and hence that}\qquad
\left[\frac{-i\sigma^{i}}{2}, \frac{-i\sigma^{j}}{2}\right] = \tensor{\epsilon}{^i^j_k}\left(\frac{-i\sigma^{k}}{2}\right),
\end{equation}
so that the linear map $\varphi: so(3)\rightarrow su(2)$ which satisfies $\varphi(L^{j}) = -i\sigma^{j}/2$ is a Lie algebra isomorphism. We conclude that the coadjoint orbits of $SO(3)$ are equal to those of $SU(2)$. Then proposition \ref{reductionorbit} allows one to construct the coadjoint orbits as the reduction of the group manifold $SU(2)$ itself.

For any Lie group $G$, each $g\in G$ determines two diffeomorphisms $G\rightarrow G$, given by
\begin{equation}\nonumber
\rho_{g}: g'\mapsto g'g\; \text{(right translation)} \quad \text{and} \quad \lambda_{g}:g'\mapsto gg' \; \text{(left translation)}.
\end{equation}
One can then define the left-invariant vector fields by
\begin{equation}\nonumber
L_{A}(e) = A, \qquad \lambda_{g*}(L_{A}) = L_{A}, \quad \forall A \in \mathfrak{g} = T_{e}G, \; \forall g\in G,
\end{equation}
where $e$ is the group identity. Then the flow of each $L_{A}$ is the one-parameter subgroup $(g,t)\mapsto ge^{tA}$. In the case of a matrix Lie group this can be seen from
\begin{equation}\nonumber
\frac{d}{dt}(ge^{tA}) = ge^{tA}A = \lambda_{ge^{tA}*}A = L_{A}(\lambda_{ge^{tA}}e) = L_{A}(ge^{tA}).
\end{equation}
If we make use of the relation
\begin{equation}\nonumber
e^{sA}e^{tB} \thicksim e^{st[A,B]}e^{tB}e^{sA},
\end{equation}
which holds up to second order in $s$ and $t$ as a consequence of the Baker-Campbell-Hausdorff formula, we may evaluate $[L_{A},L_{B}]$. Let $g\in G$ and $f: G\rightarrow \mathbb{R}$. Then
\begin{equation}\nonumber
\begin{split}
([L_{A},L_{B}](f))(g) &= \frac{d}{ds}\frac{d}{dt}[f(ge^{st[A,B]}e^{tB}e^{sA}) - f(ge^{sB}e^{tA})]\Big|_{t=s=0} = \frac{d}{dr}f(ge^{r[A,B]})\Big|_{r=0} \\
&= (L_{[A,B]}(f))(g).
\end{split}
\end{equation}
As this holds for any $f$ and $g$, we conclude that $\mathfrak{g}$ acts on $G$ by $A\mapsto L_{A}$. Pick an element $f\in \mathfrak{g}^{*}$. It determines a right-invariant one-form $\theta_{f}\in \Omega^{1}(G)$ by
\begin{equation}\nonumber
\theta_{f}(e) = f, \qquad \rho_{g}^{*}(\theta_{f}) = \theta_{f}, \quad \forall g\in G.
\end{equation}
The claim is that the coadjoint orbit through $f\in su(2)^{*}$ is the reduction of $(SU(2), d\theta_{f})$. To see how this follows from proposition \ref{reductionorbit}, consider the map $\pi : SU(2)\rightarrow (M_{f},\omega_{f}): g\mapsto ad_{g}^{*}f$, where $(M_{f},\omega_{f})$ is the coadjoint orbit through $f\in su(2)^{*}$. This is clearly surjective and, for any $B\in su(2)^{*}$, 
\begin{equation}\nonumber
\begin{split}
(\pi_{*}L_{A}|_{g})(B) &= ([\pi\circ\lambda_{g}]_{*}A)(B) = \frac{d}{dt}(ad_{ge^{tA}}^{*}f)\Big|_{t=0}(B) = \frac{d}{dt}f[ad_{ge^{tA}}(B)]\Big|_{t=0} \\
&= f(g[A,B]g^{-1}) = [ad_{g}^{*}f]([A,B]) = X_{A}\Big|_{ad_{g}^{*}f}(B),
\end{split}
\end{equation}
so that the action of $su(2)$ on $SU(2)$ does project to the coadjoint action on the orbits. Here we have used again the identification of $T_{f}su(2)^{*}$ with $su(2)^{*}$. To verify item \textbf{(iii)}, note that
\begin{equation}\nonumber
L_{A}|_{g} = \lambda_{g*}A = (\rho_{g}\circ\rho_{g^{-1}})_{*}\lambda_{g*}A = \rho_{g*}(\rho_{g^{-1}}\circ\lambda_{g})_{*}A = \rho_{g*}ad_{g}A,
\end{equation}
so that
\begin{equation}\nonumber
(L_{A}\lrcorner\theta_{f})(g) = f(ad_{g}A) = (ad_{g}^{*}f)(A) = [\pi(g)](A), \quad \forall A\in su(2), \; \forall g\in G.
\end{equation}
To summarize, this shows that the spheres centered at the origin are symplectic diffeomorphic to the reductions of $(SU(2),d\theta_{f})$, for $f\in su(2)^{*}$. To make the correspondence more explicit, recall that we have chosen bases in $so(3)$ and $su(2)$ such that
\begin{equation}\label{r3su2}
\mathbb{R}^{3}\ni \mathbf{a}\mapsto a_{i}L^{i}\mapsto \frac{-ia_{j}\sigma^{j}}{2} \in su(2).
\end{equation}
We then see the dual space of $\mathbb{R}^{3}$ as $\mathbb{R}^{3}$, with the pairing given by $\langle\mathbf{a},\mathbf{b}\rangle = \mathbf{a}^T\mathbf{b}$, and map this to $su(2)^{*}$. Now, the manifold $SU(2)$ can be embeded in $\mathbb{C}^{2}$ as the 3-sphere by taking $(z^{0},z^{1})$ in
\begin{equation}\nonumber
SU(2)\ni g = \left( \begin{array}{cc}
z^{0} & z^{1} \\
-\bar{z}^{1} & \bar{z}^{0}
\end{array} \right), \qquad z^{0}\bar{z}^{0} + z^{1}\bar{z}^{1} = 1
\end{equation}
as holomorphic coordinates. Combining this with (\ref{r3su2}), we send $\mathbb{R}^{3}$ to the tangent to the sphere at $(1,0)\in\mathbb{C}^{2}$. Explicitly,
\begin{equation}\nonumber
\begin{split}
e^{\frac{-ita_{j}\sigma^{j}}{2}} &= \cos\left(\frac{t|\mathbf{a}|}{2}\right)\mathbf{1} - \frac{i}{|\mathbf{a}|}\sin\left(\frac{t|\mathbf{a}|}{2}\right)a_{j}\sigma^{j} \\
&=\left(\begin{array}{cc}
\cos\left(\frac{t|\mathbf{a}|}{2}\right) - \frac{ia_{3}}{|\mathbf{a}|}\sin\left(\frac{t|\mathbf{a}|}{2}\right) &
-\frac{(a_{2}+ia_{1})}{|\mathbf{a}|}\sin\left(\frac{t|\mathbf{a}|}{2}\right) \\
\frac{(a_{2}-ia_{1})}{|\mathbf{a}|}\sin\left(\frac{t|\mathbf{a}|}{2}\right) &
\cos\left(\frac{t|\mathbf{a}|}{2}\right) + \frac{ia_{3}}{|\mathbf{a}|}\sin\left(\frac{t|\mathbf{a}|}{2}\right)
\end{array}\right),
\end{split}
\end{equation}
which we see as a curve $(z^{0}(t),z^{1}(t))$ through $(1,0)$. Taking the derivative, we find the vector corresponding to $\mathbf{a}$. The 2-sphere of radius $s$ is the orbit in $\mathbb{R}^{3}$ of the vector $(0,0,-s)$. Thus it is the reduction of $(SU(2),d\theta_{f})$, where the value of $\theta_{f}$ at $(1,0)$ (the identity in $SU(2)$) is fixed by demanding that it be a real one-form such that the diagram
\begin{equation}\nonumber
\begin{tikzcd}
\mathbb{R}^{3}\ni \mathbf{a} \arrow[d, mapsto, "{\cdot (0,0,-s)}"'] \arrow[r, mapsto] & \frac{-ia_{3}}{2}\frac{\partial}{\partial z^{0}} +\frac{ia_{3}}{2}\frac{\partial}{\partial \bar{z}^{0}} - \frac{(a_{2}+ia_{1})}{2}\frac{\partial}{\partial z^{1}} - \frac{(a_{2} - ia_{1})}{2}\frac{\partial}{\partial\bar{z}^{1}} \in T_{(1,0)}S^{3}\arrow[dl, mapsto, "{\lrcorner\theta_{f}\big|_{(1,0)}}"]\\
-sa_{3}&
\end{tikzcd}
\end{equation}
commutes. This gives $\theta_{f}|_{(1,0)} = is(d\bar{z}^{0} - dz^{0})$, and we determine its value on the rest of $S^{3}$ by remembering that it is right-invariant, so that $\theta_{f}|_{g \in SU(2)} = \rho_{g}^{*}\theta_{f}|_{e\in SU(2)}$. In the defined complex coordinates, multiplication on the right by the element of $SU(2)$ with coordinates $(z^{0},z^{1})$ acts by
\begin{equation}\nonumber
(w^{0},w^{1})\mapsto (z^{0}w^{0} - \bar{z}^{1}w^{1}, z^{1}w^{0} + \bar{z}^{0}w^{1}),
\end{equation}
and thus, transforming the components of $\theta_{f}$ by the inverse of the Jacobian, we find that, at any $(z^{0},z^{1})\in S^{3}$,
\begin{equation}\label{potencialsu2}
\begin{split}
\theta_{f} = is(z^{0}d\bar{z}^{0}+z^{1}d\bar{z}^{1}-\bar{z}^{0}dz^{0}-\bar{z}^{1}dz^{1}) \\
\Rightarrow d\theta_{f} = 2is(dz^{0}\wedge d\bar{z}^{0} + dz^{1}\wedge d\bar{z}^{1}).
\end{split}
\end{equation}
Finally, we have that the sphere of radius $s$ is the reduction of the presymplectic manifold 
\begin{equation}\label{s3tos2}
(S^{3},2is(dz^{0}\wedge d\bar{z}^{0} + dz^{1}\wedge d\bar{z}^{1})).
\end{equation}
This $S^3\to S^2$ is the famous Hopf fibration \cite{hopf1964abbildungen}.
\end{su2}

\begin{poincare}\label{poincareorbits}\emph{(Poincar\'e symmetry)}

Free relativistic particles correspond to the elementary systems with Poincar\'e symmetry $P$. Although it can also be seen as a matrix group, we will follow the approach in \cite{woodhouse1997geometric} and parametrize the Poincar\'e group by isometries $\rho :\mathbb{M}\rightarrow\mathbb{M}$ of Minkowski space $(\mathbb{M},\eta)$, where $\eta$ is the constant metric of signature $+---$ and we take the group composition to be such that $P$ acts on $\mathbb{M}$ on the right\footnote{See also \cite{souriau1970strucure,bacry1967space,penrose1968twistor,arens1971classical}.}. Thus one can think of its Lie algebra as the $\eta$-preserving vector fields, ie,
\begin{equation}\label{poincarealgebra}
\mathfrak{p} = \{X\in T\mathbb{M}|\mathcal{L}_{X}g = 0\}.
\end{equation}
For a coordinate expression, notice that
\begin{equation}\nonumber
0 = \mathcal{L}_{X}g_{ab} = \nabla_{a}X_{b} - \nabla_{b}X_{a} \,\,\,\,\Rightarrow\,\,\,\, X^{a} = q^{b}\tensor{L}{_b^a}+T^{a},
\end{equation}
where $L$ and $T$ are constants and $L_{ba} = -L_{ab}$. We have used the flatness of the metric. Then, if $X = (q^{b}\tensor{L}{_b^a}+T^{a})\partial_{a}$ and $Y = (q^{b}\tensor{M}{_b^a}+N^{a})\partial_{a}$,
\begin{equation}\nonumber
\begin{split}
&rX+sY = [q^{b}(r\tensor{L}{_b^a}+s\tensor{M}{_b^a}) + (rT^{a}+sN^{a})]\partial_{a} \in \mathfrak{p}, \quad \forall r,s \in \mathbb{R} \\
&[X,Y] = \mathcal{L}_{X}Y = (X^{b}\partial_{b}Y^{a} - Y^{b}\partial_{b}X^{a})\partial_{a} = [q^{c}(\tensor{L}{_c^b}\tensor{M}{_b^a} - \tensor{M}{_c^b}\tensor{L}{_b^a}) + (T^{b}\tensor{M}{_b^a} - N^{b}\tensor{L}{_b^a})]\partial_{a} \in \mathfrak{p},
\end{split}
\end{equation}
 and these vectors form a Lie subalgebra of $V(\mathbb{M})$ (see \ref{lievector}). Furthermore, the second equation implies that $[\mathfrak{p},\mathfrak{p}] = \mathfrak{p}$, because every antisymmetric matrix is the bracket of two antisymmetric matrices of the same order and every vector $(P^{a})\in \mathbb{R}^{n}$ can be writen as $MT-LN$ for fixed antisymmetric matrices $M$ and $L$ by choosing suitable vectors $(T^{a}),(N^{b})\in \mathbb{R}^{n}$. Although $\mathfrak{p}$ is not semisimple, it is still true that $H^{2}\mathfrak{p} = 0$ (for a proof, see \cite{woodhouse1997geometric}) and, by \ref{corsesi}, there is always one and only one way of associating a moment to a given canonical action of the Poincar\'e algebra on a symplectic manifold. Moreover, if the action is transitive, there is a unique canonical map of this symplectic manifold to one of the coadjoint orbits. Hence the coadjoint orbits give all of the elementary systems with Poincar\'e symmetry. The physical interpretation of a classical system having such a transitive action of the group of spacetime isometries is that it does not have any structure other than its spacetime structure \cite{souriau1970strucure}. These are the elementary relativistic particles.

The most general linear funcional acting on $\mathfrak{p}$ can be written as
\begin{equation}\nonumber
f(q^b\tensor{L}{_b^a}+T^{a}) = -\frac{1}{2}M^{ab}L_{ab}-p_{a}T^{a}
\end{equation}
for some constants $M^{ab} = -M^{ba}$ and $p_{a}$. If we recquire that this pairing be invariant under Lorentz transformations and translations of the origin in $\mathbb{M}$ (under which $L$ and $T$ transform in the obvious way), i.e.,
\begin{equation}\nonumber
\frac{1}{2}\tilde{M}^{ab}\tilde{L}_{ab}+\tilde{p}_{a}\tilde{T}^{a} = \frac{1}{2}M^{ab}L_{ab}+p_{a}T^{a},
\end{equation}
then we discover that the components of $M$ and $p$ transform as tensors under Lorentz transformations, in the way suggested by their indices, but under a change of origin $x\mapsto x+K$ one must take
\begin{equation}\label{Mtransf}
(M^{ab},p_{c})\mapsto (M^{ab}+p^{a}K^{b}-K^{a}p^{b},p_{c}),
\end{equation}
where we recognize the transformation law for the components of the total angular momentum $M$ if we take $p$ to be the four momentum. Furthermore, this shows that we may adopt a characterisation of $f\in \mathfrak{p}$ which is independent of the choice of origin in $\mathbb{M}$ if we trade $M$ and $T$ for the tensor
\begin{equation}\label{tensorpoincare}
f^{ab} = M^{ab} + p^{a}x^{b} - x^{a}p^{b},
\end{equation}
so that its value at $x\in \mathbb{M}$ gives the total angular momentum about $x$. This aligns well with our intuition about what the moment map should be in this case. One sees that no information is lost since
\begin{equation}\nonumber
\begin{split}
f(0) = M \\
\frac{1}{3}\nabla_{b}f^{ab} = \frac{1}{3}(p^{a}\nabla_{b}x^{b} - p^{b}\nabla_{b}x^{a}) = p^{a},
\end{split}
\end{equation}
so we might just as well start from $f$ and recover $M$ and $p$. Thus we think of $\mathfrak{p}$ as the vector fields of the form (\ref{poincarealgebra}) and $\mathfrak{p}^{*}$ as the tensor fields of the form (\ref{tensorpoincare}) on $\mathbb{M}$, with the pairing given by
\begin{equation}\label{pairf}
f(X) = -\frac{1}{2}f^{ab}\nabla_{a}X_{b} - \frac{1}{3}X_{a}\nabla_{b}f^{ab},
\end{equation}
avoiding explicit reference to the origin in $\mathbb{M}$ (as this expression for the pairing is constant throughout $\mathbb{M}$). In this approach, one can find the coadjoint action in the following way. Let $\rho_{t}$ be the flow generated by $X\in \mathfrak{p}$. Then vector field $X$ is invariant under the isometry $\rho_{t}$ for each $t$ since, for any $m\in\mathbb{M}$,
\begin{equation}\nonumber
\rho_{t*}X(m) = \frac{d}{ds}\rho_{t}(\rho_{s}(m))\Big|_{s=0} = \frac{d}{ds}\rho_{s}(\rho_{t}(m))\Big|_{s=0} = X(\rho_{t}(m)),
\end{equation}
ie., $\rho_{t*}X = X$. Now the adjoint action of $P$ on itself acts on this isometry by $\rho_{t}\mapsto \rho\rho_{t}\rho^{-1},\,\, \forall \rho\in P$. This is true for any real $t$, so one may speak of the flow $\rho\rho_{t}\rho^{-1}$, which is generated by the vector field $(\rho_{*}^{-1})(X)$. Indeed, this vector field is invariant under the flow, as
\begin{equation}\nonumber
(\rho\rho_{t}\rho^{-1})_{*}(\rho_{*}^{-1}X) = \rho_{*}^{-1}\rho_{t*}\rho_{*}\rho_{*}^{-1}X = \rho_{*}^{-1}X.
\end{equation}
The inversion on the order of composition of maps arises because we assumed from the begining that $P$ acts on $\mathbb{M}$ on the right. We conclude that the derivative of the adjoint action in the group is given by $ad_{\rho}X = \rho_{*}^{-1}X$. Hence by (\ref{pairf}) the coadjoint action should be given simply by
\begin{equation}\nonumber
ad_{\rho}^{*}f = \rho_{*}f.
\end{equation}
Here, the $*$ symbol means dual map on the LHS while it means the differential map on the RHS. We shall consider first the coadjoint orbits with respect to the identity component $P_{0}$ of the Poincar\'e group, leaving the discrete transformations of parity and time-reversal to be implemented later. Thus at this stage the coajoint orbits are seen to be of the form $\{\rho_{*}f|\rho\in P_{0}\}$.

Because of the tensor character of the coefficients of $f$ under the coadjoint action, we see that $m^{2}:=p_{a}p^{a}$ is constant throughout each orbit. Consider the case $m^{2}>0$ (massive particles). Then the Pauli-Lubanski vector,
\begin{equation}\nonumber
S^{a} = \frac{1}{2}\epsilon^{abcd}p_{b}M_{cd}
\end{equation}
is orthogonal to $p$ ($\epsilon$ denotes the Levi-Civita symbol). Again, because of the way in which the group acts on $f\in\mathfrak{p}^{*}$, the length of $S$ is invariant (the antisymmetrization of indices of M in the definition of $S$ guarantees that this also holds for translations (\ref{Mtransf})). This implies that the spin $s$, defined through $m^{2}s^{2} = -S_{a}S^{a}$, is also a constant on each orbit. The simplest situation is when $s=0$ (scalar particle), which happens when the Pauli-Lubanski vector vanishes ($S$ can't be light-like since $p_{a}p^{a}>0$). This means that $p^{[a}f^{bc]} = 0$ on all of $\mathbb{M}$. Since $\mathbb{M}$ has trivial topology, this implies $f^{ab} = 2p^{[a}w^{b]}$ for some four-vector field $w$. Thus we have
\begin{equation}\nonumber
M^{ab}+2p^{[a}x^{b]} = f^{ab} = 2p^{[a}w^{b]} \Rightarrow 2p^{[a}(w-x)^{b]} = M^{ab} = \text{constant},
\end{equation}
which is solved if there exists a constant four-vector $q^{a}$ such that $w - x + q = \lambda(x)p, \,\forall x$, since then $2p^{[a}(w-x)^{b]} = -2p^{[a}q^{b]}$. Thus there is a unique timelike geodesic, namely $\{x = q + \lambda p| \, \lambda \in \mathbb{R}\}$, where $f$ = 0:
\begin{equation}\nonumber
f^{ab}|_{x = q + \lambda p} = 2p^{[a}w^{b]}|_{x = q + \lambda p} = 2p^{[a}(x - q + \tilde{\lambda} p)^{b]}|_{x = q + \lambda p} = 0.
\end{equation}
We call this the centre-of-mass world-line, since it is the locus of vanishing total angular momentum. Notice that, given $m^{2}>0$ and the centre-of-mass worldline, $p$ can be recovered as the tangent four-vector normalized so that $p_{a}p^{a} = m^2$ and $f$ can be recovered as $f^{ab} = 2p^{[a}x^{b]} - 2p^{[a}q^{b]}$, where $q$ is an arbitrary point on the centre-of-mass worldline. However, this fixes only the direction of $p$, but not its orientation (whether it is future- or past-pointing), so there is a two-to-one relation between the coadjoint orbits in $\mathfrak{p}^{*}$ with positive $m^{2}$ and zero spin and the timelike geodesics in $\mathbb{M}$. The action of $P_{0}$ on the coadjoint orbits becomes its natural action on $\mathbb{M}$, and we see that any timelike geodesic can be related to any other through one of its elements, but the causality of $p$ cannot be changed. Thus there are two orbits with spin zero for each $m^{2}>0$. We denote them by $M_{0m}^{+}$($p$ is future-pointing) and $M_{0m}^{-}$($p$ is past-pointing).

Finally, we employ proposition \ref{reductionorbit}, just as in the example of the sphere. First, each centre-of-mass worldline allows for a parametrization $(p_{a},q^{b})$ subject to the equivalence relation $(p_{a},q^{b})\sim (p_{a},q^{b}+\lambda p^{b}), \,\, \forall \lambda\in\mathbb{R}$. Hence consider the projection
\begin{equation}\nonumber
\pi : C_{0m} = \{(p,q)\in T^{*}\mathbb{M}| p_{a}p^{a} = m^{2}\} \longrightarrow C_{0m}/\sim.
\end{equation}
Because of the above discussion, the 7-dimensional hypersurface $C_{0m}$ has two components depending on the time orientation of $p$. The restriction of $\pi$ to one of the components gives the surjection of proposition \ref{reductionorbit} onto one of the orbits. To see how this works, note that we defined $P_{0}$ by its action on $\mathbb{M}$: $q\mapsto \rho (q)$ so that it acts naturally on $T^{*}\mathbb{M}$ by $(p, q)\mapsto (\rho_{*}p,\rho (q))$. But this agrees with how $P_{0}$ acts on $\mathfrak{p}^{*}$ when $\mathfrak{p}^{*}$ is parametrized by $(p,q)$ as above. Thus the infinitesimal generators $X \in \mathfrak{p}$ lift to 
\begin{equation}\nonumber
X' = X^{a}\frac{\partial}{\partial q^{a}} - p_{b}\nabla_{a}X^{b}\frac{\partial}{\partial p_{a}}
\end{equation}
on $T^*\mathbb{M}$, as is clear from remembering that $\nabla_{a}X^{b} = \tensor{L}{_a^b}$, and this is tangent to $C_{0m}$ since $e^{X}\in P_{0}$ preserves the length of $p$. Therefore $\mathfrak{p}$ acts on $C_{0m}$ by $X\mapsto X'$ and $\pi _{*}X'$ gives the correct action in the coadjoint orbit. Finally, let us evaluate the pairing $f(X)$ at a point in the centre-of-mass worldline. There $f^{ab} = 0$, $\frac{1}{3}\nabla_{b}f^{ab} = p^{a}$ and thus, by (\ref{pairf}),
\begin{equation}\nonumber
f(X) = -p_{a}X^{a} = X'\lrcorner\theta ',
\end{equation}
where $\theta ' = -p_{a}dq^{a}|_{C_{0m}}$. Thus we conclude that the orbits $M_{0m}^{+}$ and $M_{0m}^{-}$ are the reductions of the two components of the presymplectic manifold $(C_{0m}, dq^{a}\wedge dp_{a})$. From the physical point of view, this is the classical phase space of a relativistic massive scalar particle.

Let us now repeat the procedure when $m^{2}, s^{2}>0$. In this case, there is no geodesic on which $f=0$ (which one interprets as meaning that a spinning particle has nonzero angular momentum around any event). However, one can define the centre-of-mass worldline by $x^{a} = m^{-2}M^{ab}p_{b}+\lambda p^{a}$, since over it
\begin{equation}\nonumber
f^{ab}p_{b}|_{CM} = M^{ab}p_{b}+p^{a}(m^{-2}M^{bc}p_{c}+\lambda p^{b})p_{b}-p^{b}(m^{-2}M^{ac}p_{c}+\lambda p^{a})p_{b} = 0,
\end{equation}
remembering that $p_{a}p^{a} = m^{2}$. This means that the orbital component of the angular momentum vanishes on this geodesic. On this line, $f$ becomes
\begin{equation}\nonumber
f^{ab}|_{CM} = M^{ab}+\frac{p^{a}p_{c}M^{bc}}{m^{2}}-\frac{p^{b}p_{c}M^{ac}}{m^{2}} = \frac{p^{e}p_{c}M^{fg}}{2m^{2}}(\epsilon^{abcd}\epsilon_{defg}) = \frac{\epsilon^{abcd}p_{c}S_{d}}{m^{2}},
\end{equation}
which is a constant. Therefore, after fixed the centre-of-mass worldline, the only freedom left to fix in order to specify a point $f\in\mathfrak{p}^{*}$ is the direction of $S$ (remember that $S_{a}S^{a}=-m^{2}s^{2}$). The only constraint on the direction of $S_{a}$ is that $p_{a}S^{a}=0$, so that $S$ can be any vector in the three-space of the centre-of-mass rest frame (where $p = (m,0,0,0)$). Now, since $S$ is transforms as a four-vector under the coadjoint action, it follows that, even after $p$ is fixed to be in the time direction, there is always a Lorentz transformation which rotates the direcion of $S$ into any other, so all of the elements on the dual of the Poincar\'e Lie algebra which differ only by the direction of $S$ lie in the same coadjoint orbit. We see that each orbit is specified by the centre-of-mass worldline, by the sign of $p_{0}$, and by the direction of $S$. Therefore there are again two orbits for each value of $m$, $s$, which we denote $M_{sm}^{+(-)}$, which have the structure of bundles over $M_{0m}^{+(-)}$ with fibre $S^{2}$. This connects with example \ref{sphereso3} in that we model the phase space corresponding to the spin degrees of freedom as a sphere of radius $s$.

The connection can be made more explicit using a spinor parametrization of the orbits \cite{woodhouse1997geometric}. For this one uses the following mapping from Minkowski space $\mathbb{M}$ to $2\times 2$ hermitian matrices, which in turn are equivalent to $\mathbb{S}\otimes\bar{\mathbb{S}}$, where $\mathbb{S}$ is the space of two-spinors: it takes the four-vector $(X_{\mu})$ to
\begin{equation}\nonumber
\begin{split}
(X_{A\bar{A}}) =\left(\begin{array}{cc}
X_{0}+X_{3} &
X_{1}-iX_{2} \\
X_{1}+iX_{2} &
X_{0}-X_{3}
\end{array}\right).
\end{split}
\end{equation}
This then can be extended to a spinor representation of tensors of any order\footnote{The spinor indices $A, \bar{A}$ transform, respectively, under the $SL(2,\mathbb{C})$ spinor representation of the Lorentz group and its conjugate representation.}. We refer to \cite{penrose1984spinors} for the extensions of this story and proofs of the necessary results. Note that
\begin{equation}\nonumber
f^{ab}|_{CM} = \frac{\epsilon^{abcd}p_{c}S_{d}}{m^{2}} = \left(\star\frac{2pS}{m^{2}}\right)^{ab},
\end{equation}
where $\star$ is the Hodge dual, $(\star f)^{ab} = \frac{1}{2}\epsilon^{abcd}f_{cd}$. Now, it is a result from spinor theory that any bivector has a spinor equivalent of the form
\begin{equation}\nonumber
\phi^{AB}\epsilon^{\bar{A}\bar{B}}+\psi^{\bar{A}\bar{B}}\epsilon^{AB},
\end{equation}
with $\phi^{AB}$ and $\psi^{AB}$ symmetric. In case the bivector is real, $\psi^{\bar{A}\bar{B}} = \bar{\phi}^{\bar{A}\bar{B}}$. Moreover, the dual bivector is then given by
\begin{equation}\nonumber
-i\phi^{AB}\epsilon^{\bar{A}\bar{B}}+i\psi^{\bar{A}\bar{B}}\epsilon^{AB}.
\end{equation}
We further note that any $n$-index symmetric spinor can be written as the symmetrized product of $n$ one-index spinors. Therefore, we may express the spinor equivalent of the real bivector $2pS/m^{2}$ as
\begin{equation}\label{fsp1}
-sz^{(A}w^{B)}\epsilon^{\bar{A}\bar{B}}-s\bar{z}^{(\bar{A}}\bar{w}^{\bar{B})}\epsilon^{AB},\,\,\,\text{with}\,\,\, w^{A} = \pm \frac{\sqrt{2}}{m}p^{A\bar{A}}\bar{z}_{\bar{A}},
\end{equation}
where the sign in the expression for $w^{A}$ is the same as the sign of $p_{0}$. Since $f^{ab}|_{CM}$ is the dual of this, it is given by
\begin{equation}\label{fsp2}
isz^{(A}w^{B)}\epsilon^{\bar{A}\bar{B}}-is\bar{z}^{(\bar{A}}\bar{w}^{\bar{B})}\epsilon^{AB},
\end{equation}
and in particular one can check that equations (\ref{fsp1}) and (\ref{fsp2}) give the correct normalization $S_{a}S^{a}=-m^{2}s^{2}$. Additionally, note that although $w$ is fixed by (\ref{fsp1}) there is an additional freedom in choosing $z$ reflected by the fact that $(z,w)$ and $(\lambda z,\lambda^{-1}w)$ give the same $f$. We fix this by also imposing $z_{A}w^{A}=1$.

Hence, a specified orbit with $m^{2},s^{2}>0$ is given by a timelike geodesic (just as in the scalar case) and a two-component spinor $z^{A}$, defined up to a phase. Therefore it is the quotient of $T^{*}\mathbb{M}\times\mathbb{S}$ by the equivalence relation $(p_{a},q^{b},z^{C})\sim (p_{a}, q^{b}+\lambda p^{b},e^{i\phi}z^{C}), \,\, \forall \lambda, \phi \in \mathbb{R}$. Again, we consider the projection
\begin{equation}\label{csm}
\pi : C_{sm} = \{(p,q, z)\in T^{*}\mathbb{M}\times\mathbb{S}| p_{a}p^{a} = m^{2}, \sqrt{2}p_{A\bar{A}}z^{A}\bar{z}^{\bar{A}}=\pm m\} \longrightarrow C_{sm}/\sim.
\end{equation}
This time the hypersurface $C_{sm}$ is $9$-dimensional, and has two components which differ by the time orientation of $p$. To apply proposition \ref{reductionorbit}, we recall that $P_{0}$ has a natural action on $T^{*}\mathbb{M}\times\mathbb{S}$ by $(p, q, z)\mapsto (\rho_{*}p,\rho (q),\rho_{*}z)$. Again, this agrees with the coadjoint action of $P_{0}$ when $\mathfrak{p}^{*}$ is parametrized by $(p,q,z)$ as above. The infinitesimal form of this action is given by the generators  
\begin{equation}\nonumber
X' = X^{a}\frac{\partial}{\partial q^{a}} - p_{b}\nabla_{a}X^{b}\frac{\partial}{\partial p_{a}} + z^{A}\frac{1}{2}\nabla_{A\bar{B}}X^{B\bar{B}}\frac{\partial}{\partial z^{B}}+\bar{z}^{\bar{A}}\frac{1}{2}\bar{\nabla}_{\bar{A}B}\bar{X}^{\bar{B}B}\frac{\partial}{\partial\bar{z}^{\bar{B}}},
\end{equation}
which is clear if one remembers that $\nabla_{a}X^{b} = \tensor{L}{_a^b}$. The action preserves the lengths of both $p$ and $z$, so it preserves $C_{sm}$ and projects to the coajoint action on the orbits. To obtain the symplectic structure from proposition \ref{reductionorbit}, we only need the potential satisfying condition {\bf (iii)}. This is the restriction to $TC_{sm}$ of 
\begin{equation}\label{thetaspin}
\theta ' =\pm \frac{\sqrt{2}is}{m}p_{A\bar{A}}(z^{A}d\bar{z}^{\bar{A}}-\bar{z}^{\bar{A}}dz^{A})-p_{a}dq^{a}.
\end{equation}
To see this, recall that the pairing on $\mathfrak{p}^{*}\times\mathfrak{p}$ is given by $f(X) = -\frac{1}{2}f^{ab}\nabla_{a}X_{b} - \frac{1}{3}X_{a}\nabla_{b}f^{ab}$. The second term is given by $-p_{a}X^{a} = X'\lrcorner (-p_{a}dq^{a})$, just as in the scalar case, while the first term is given, on the centre-of-mass worldline, by
\begin{equation}\nonumber
\begin{split}
-\frac{1}{2}f^{ab}\nabla_{a}X_{b}|_{CM} & \sim \pm \frac{1}{2}\frac{is\sqrt{2}}{m}p_{A\bar{A}}\left[z^{A}\bar{z}^{\bar{C}}\bar{\nabla}_{\bar{C}B}\bar{X}^{\bar{A}B}-\bar{z}^{\bar{A}}z^{C}\nabla_{C\bar{B}}X^{A\bar{B}}\right] \\
&= \left( z^{A}\frac{1}{2}\nabla_{A\bar{B}}X^{B\bar{B}}\frac{\partial}{\partial z^{B}}+\bar{z}^{\bar{A}}\frac{1}{2}\bar{\nabla}_{\bar{A}B}\bar{X}^{\bar{B}B}\frac{\partial}{\partial\bar{z}^{\bar{B}}}\right) \lrcorner \left[ \pm \frac{\sqrt{2}is}{m}p_{A\bar{A}}(z^{A}d\bar{z}^{\bar{A}}-\bar{z}^{\bar{A}}dz^{A})\right],
\end{split}
\end{equation}
using (\ref{fsp1}) and (\ref{fsp2}). Thus we find that the orbits $M_{sm}^{+(-)}$ are the reductions of the two components of the presymplectic manifold $(C_{sm}, d\theta '|_{C_{sm}})$ with $\theta '$ given by equation (\ref{thetaspin}). From the physical point of view, this is the classical phase space of a relativistic massive spinning particle.

Consider now a spining particle at rest, i.e., take the subspace of $M_{sm}^{+}$ in which $q$ is fixed and $(p_{\mu})=(m,0,0,0)$. Substituting this in (\ref{csm}), we see that it should be the reduction of the submanifold labelled by the $z^{A}$ such that
\begin{equation}\nonumber
\begin{split}
m = \sqrt{2}p_{A\bar{A}}z^{A}\bar{z}^{\bar{A}} = 
\sqrt{2}
\left(\begin{array}{cc}
z^{0} &
z^{1}
\end{array}\right)
\frac{1}{\sqrt{2}}
\left(\begin{array}{cc}
m &
0 \\
0 &
m
\end{array}\right)
\left(\begin{array}{cc}
\bar{z}^{0} \\
\bar{z}^{1}
\end{array}\right)
\,\,\,\,\Rightarrow\,\,\,\, z^{0}z^{1}+\bar{z}^{0}\bar{z}^{1}=1.
\end{split}
\end{equation}
Similarly, the restriction to this submanifold of $d\theta '$ is
\begin{equation}\nonumber
\frac{\sqrt{2}is}{m}p_{A\bar{A}}(dz^{A}\wedge d\bar{z}^{\bar{A}}-d\bar{z}^{\bar{A}}\wedge dz^{A}) = 2is(dz^{0}\wedge d\bar{z}^{0}+dz^{1}\wedge d\bar{z}^{1}).
\end{equation}
This connects the phase space of a spinning particle of spin $s$ to the two-sphere of radius $s$, and simplifies the `reduction' procedure to that of the previous example.

The third type of orbit we look at is the one in which $p_{a}p^{a}=0$ and $S^{a}=s p^{a}$ for some constant $s$, which we call helicity of the orbit (note that, although the sign of the spin is non-physical, because it is $s^{2}$ which parametrizes the orbits, the sign of the helicity is a relevant, labelling distinct orbits). Once again, we use the same technique as in the cases above. For a fixed value of $s$, there are again two orbits depending on the sign of $p_{0}$. In the one where $p_{0}>0$, the proportionality of $S^{a}$ and $p^{a}$ gives a constraint
\begin{equation}\label{ps}
\epsilon^{abcd}p_{b}M_{cd}=2s p^{a}
\end{equation}
which, in turn, implies that the spinor representation of $f^{ab}=M^{ab}+p^{a}x^{b}-p^{b}x^{a}$ can be written as\footnote{See chapter 6 of \cite{penrose1984spinors}.}
\begin{equation}\label{ftwistor}
iz^{(A}\bar{\pi}^{B)}\epsilon^{\bar{A}\bar{B}}-i\bar{z}^{(\bar{A}}\pi^{\bar{B})}\epsilon^{AB},
\end{equation}
where $\bar{\pi}_{A}\pi_{\bar{A}}=p_{A\bar{A}}$ and $z^{A}=\omega^{A}-ix^{A\bar{A}}\pi_{\bar{A}}$, with $\omega^{A}$ a constant spinor. Finally, relation (\ref{ps}) implies the normalization $z^{A}\bar{\pi}_{A}+\bar{z}^{\bar{A}}\pi_{\bar{A}}=2s$. Hence the two spinors $\omega^{A}$ and $\pi^{\bar{A}}$ specify the tensor field $f^{ab}$. Conversely, equation (\ref{ftwistor}) specifies $(\omega^{A},\pi_{\bar{A}})$ up to the equivalence $(\omega^{A},\pi_{\bar{A}})\sim (e^{i\phi}\omega^{A},e^{i\phi}\pi_{\bar{A}}), \,\,\, \phi\in\mathbb{R}$. So we want to use the projection
\begin{equation}\nonumber
\pi : C_{s0}=\{(\omega,\pi)\in\mathbb{S}\times\bar{\mathbb{S}}|\omega^{A}\bar{\pi}_{A}+\bar{\omega}^{\bar{A}}\pi_{\bar{A}}=2s\}\Longrightarrow C_{s0}/\sim .
\end{equation}
Again we apply proposition \ref{reductionorbit}: the adjoint action of $P_{0}$ on the orbits, which is given by $\rho_{*}f$, becomes, in the parametrization chosen, $ad_{\rho}^{*}(\omega^{A},\pi_{\bar{A}}) = (\rho_{*}\omega^{A},\rho_{*}\pi_{\bar{A}})$. Notice, however, that although under Lorentz transformations it is given by $\tensor{L}{_a^b}=\nabla_{a}X^{b}, \,\, X\in \mathfrak{p}$ just as before, when considering translations by some four-vector $V^{a}$, $z^{A}$ is invariant, while the definition of $\omega$ gives $\omega^{A'}=z^{A'}+i(x+T)^{A\bar{A}}\pi_{\bar{A}} = \omega^{A}+iT^{A\bar{A}}\pi_{\bar{A}}$. Therefore, the infinitesimal form of this action on $\mathbb{S}\times\bar{\mathbb{S}}$ is given by
\begin{equation}\nonumber
X'=2\Re\left(\omega^{A}\frac{1}{2}\nabla_{A\bar{B}}X^{B\bar{B}}\frac{\partial}{\partial \omega^{B}}-\pi_{\bar{A}}\frac{1}{2}\bar{\nabla}_{\bar{A}B}\bar{X}^{\bar{B}B}\frac{\partial}{\partial \pi_{\bar{B}}}+iT^{A\bar{A}}\pi_{\bar{A}}\frac{\partial}{\partial\omega^{A}}\right),
\end{equation}
where $T^{a}=X^{a}(0)$ is the translation part of $X\in\mathfrak{p}$. Again, these induce flows preserving $C_{s0}$ and
\begin{equation}\nonumber
f(X) = -\frac{1}{2}f^{ab}\nabla_{a}X_{b}-\frac{1}{3}X_{a}\nabla_{b}f^{ab} \sim -i\omega^{A}\bar{\pi}^{B}\frac{1}{2}\nabla_{A\bar{B}}X^{B\bar{B}}+i\bar{\omega}^{\bar{A}}\pi^{\bar{B}}\frac{1}{2}\bar{\nabla}_{\bar{A}B}\bar{X}^{\bar{B}B}-\pi_{\bar{A}}\bar{\pi}_{A}T^{A\bar{A}},
\end{equation}
which is given by $X'\lrcorner\theta'$ for
\begin{equation}\nonumber
\theta' = i\bar{\pi}_{A}d\omega^{A}-i\pi_{\bar{A}}d\bar{\omega}^{\bar{A}}-i\omega^{A}d\bar{\pi}_{A}+i\bar{\omega}^{\bar{A}}d\pi_{\bar{A}}.
\end{equation}
Thus we conclude that this orbit, denoted $M_{s0}^{+}$, is the symplectic reduction of 
\begin{equation}\nonumber
(C_{s0},-id\omega^{A}\wedge d\bar{\pi}_{A}+id\bar{\omega}^{\bar{A}}\wedge d\pi_{\bar{A}}|_{C_{s0}}).
\end{equation}

There are other coadjoint orbits in $\mathfrak{p}^{*}$ which do not fit in any of the types investigated. These describe particle dymanics of types not observed in nature, some of them having, for example, negative $m^{2}$ and thus moving faster than light.

In order to obtain elementary systems with respect to the full Poincar\'e group $P$, it is necessary to implement the discrete symmetries of time reversal and spatial reflection. In doing this, we choose the action so as to make it agree at the quantum level with the usual conventions from quantum field theory. In particular, it should happen that some of these transformations should become anti-unitary operators, which is the case if we take the classical action to be anti-canonical, ie., such that $\rho^{*}\omega = -\omega$. In light of this, we update the definition of elementary systems to mean a transitive action of the symmetry group by either canonical or anti-canonical transformations. For the case of the elementary particles, we then adopt the definition
\begin{equation}\label{pacts}
\text{ad}_{\rho}^{*}f = \xi\rho_{*}, \,\, \rho\in P,
\end{equation}
where $\xi = 1$ if $\rho$ preserves the arrow of time and $\xi = -1$ if it reverses. In the case of massive particles of arbitrary spin (scalar particles included), this is preserves both $M_{sm}^{+}$ and $M_{sm}^{-}$, with the transformations that reverse time acting anti-canonically. An additional symmetry, however, is also preserved by quantization: consider the symplectic manifold obtaining by changing the sign of the symplectic structure on $M_{sm}^{-}$, so that this orbit is given by the reduction of the component on which $p_{0}<0$ of $(C_{sm},-d\theta '|_{C_{sm}})$. In this case,
\begin{equation}\nonumber
C:M_{sm}^{+}\to M_{sm}^{-}:f\mapsto -f
\end{equation}
is a symplectic diffeomorphism. Upon quantization, this is recognized as the charge quantization symmetry, so that $M_{sm}^{+}$ and $M_{sm}^{-}$ are seen to describe the phase spaces of a massive spinning relativistic particle and its associated antiparticle, respectively, the two being canonically equivalent and interchanged by charge conjugation. Therefore we shall consider the total symplectic manifold $M_{sm}$, given by the two-component reduction of the presymplectic manifold $(C_{sm},\omega ')$, where
\begin{equation}\label{omegaspin}
\omega ' =
\begin{cases}
d\theta ' &, \,\,\,\,\, \text{if} \,\,\,\, p_{0}>0\\
-d\theta ' &, \,\,\,\,\, \text{if} \,\,\,\, p_{0}<0
\end{cases}
,
\end{equation}
in both the massive scalar and spinning cases, which is an elementary system with respect to $P\times\mathbb{Z}_{2}$, the group generated by the isometries of flat space-time and charge conjugation.

In the massless case, we again define the action of the whole of $P$ by equation (\ref{pacts}). However, charge conjugation does not give rise to an independent quantum symmetry, so the actual total phase space should be taken to be the two-component symplectic manifold $M_{s0}=M_{s0}^{+}\cup M_{-s0}^{+}$, on which again the elements of $P$ act transitively by canonical or anti-canonical transformations as determined by whether or not they preserve time-orientation, the transformations that reverse only space orientation exchanging the two components. Physically, the two components are given by the two helicity states of the particle differing by the sign of $s$.

\end{poincare}

\pagebreak

\subsection{Prequantization}
Given the classical phase space and the functions which make up the observables of interest, how do we construct the Hilbert space of quantum states and the operators which are the quantum counterparts of the relevant physical observables? This is the question of quantization, and it was Dirac who first laid out the rules which provide the guidelines for possible answers to this question. \cite{dirac1981principles}. A geometrical interpretation of his axioms is the following. One starts with a symplectic manifold $(M,\omega)$, where the symplectic structure endows the smooth functions on $M$ with the structure of a Lie algebra (the Poisson bracket, proposition \ref{poissonalgebra}), which we denote by $C^{\infty}(M)$. Hamilton's equation then provides a morphism from this to the algebra of Hamiltonian vector fields, $V^{H}(M)$ which generate symplectomorphisms of phase space. We can then recast equation (\ref{hamiltonpoisson}) as describing the exact sequence
\begin{equation}\label{algebraexact}
0 \longrightarrow \mathbb{R} \longrightarrow C^{\infty}(M) \longrightarrow V^{H}(M) \longrightarrow 0,
\end{equation}
where $\mathbb{R}$ is seen as the abelian Lie algebra of constant functions. In this way classical observables generate \emph{classical symmetries}: flows on $M$ which preserve the symplectic structure. Conversely, in the quantum system the states are normalised vectors in a Hilbert space, $\mathcal{H}$ and the observables form a subalgebra $\mathcal{O}$ of $\text{gl}(\mathcal{H})$, consisting of operators which generate \emph{quantum symmetries}: flows on $\mathcal{H}$ preserving the Hermitian structure. One then asks that there be an association $\mathcal{Q}: C^{\infty}(M)\rightarrow \mathcal{O}$ such that
\begin{align} \label{dirac}
\begin{split}
\bullet \, & \mathcal{Q} : C^{\infty}(M)\rightarrow \mathcal{O} \text{ is $\mathbb{R}$-linear}. \\ 
\bullet \, & [\mathcal{Q}(f),\mathcal{Q}(g)] = -i\hbar \mathcal{Q}(\{f,g\})\\
\bullet \, & f \text{ is a constant function} \Rightarrow \mathcal{Q}(f) = f\mathds{1}_{\mathcal{H}} \text{ acts by multiplication by $f$}, \\
\end{split}
\end{align}
where $[\cdot,\cdot]$ denotes the commutator of linear maps. We will very often use the notation $\hat{f} := \mathcal{Q}(f)$. Note that, even though the classical dynamics is essentially governed by $V^{H}(M) = C^{\infty}(M)/\mathbb{R}$, it is not enough to `quantize' this algebra, as the action of the constants is relevant at the quantum level (for example, to implement the uncertainty principle).

As we will briefly discuss, Dirac's quantization rules cannot be fully realized for all classical observables and there are a number of subtleties involved, but in geometric quantization the procedure of \emph{prequantization} gives a first step towards the solution. It starts by constructing a line bundle with hermitian structure and compatible connection over the symplectic manifold $(M,\omega)$, whose curvature form is given by $\omega$. Then sections of this bundle will work as wavefunctions, giving the quantum states, and the action of the quantum observables on such states can be guessed from the action of the classical observables on the base manifold. This subsection will discuss how this prescription solves a few of Dirac's requirements.

Notice, however, that the so-constructed wavefunctions depend on all coordinates of phase space (both position and momentum, for example). This is at odds with our experience of quantum mechanics where the wavefunction depends only on position of momentum. This choice of determining sections which depend on half of the coordinates is the next step of quantization, and will be taken up in subsection \ref{subsecquant}.

We start from a few properties of Hermitian line bundles with connection. For simplicity we will often restrict the discussion to the case when the base manifold is connected and simply connected. For generalisations of the procedure we refer to \cite{kostant1970quantization}. Let $\pi : B\rightarrow M$ be a complex line bundle with connection $\nabla$ and Hermitian structure $(\cdot ,\cdot )$. For $X\in V(M)$ and $s:M\rightarrow B$ a smooth section, we will often think of the connection as $\nabla _{X}s = X\lrcorner Ds$, where $D: C^{\infty}(B) \rightarrow C^{\infty}(T^{*}M \otimes B)$ is such that
\begin{equation}\nonumber
\begin{split}
D(s + s') = Ds + Ds'\\
D(fs) = df\otimes s + fDs
\end{split}\,\,\, ,
\end{equation}
for any $f\in C^{\infty}(M),\,\, s, s' \in C^{\infty}(B)$. Given a local section $s\in C^{\infty}(B|_U)$, we can define a local potential one-form $\Gamma _{s}$ by $Ds = \Gamma _{s}\otimes s$. We also require that the connection should preserve the Hermitian structure, meaning that
\begin{equation}\nonumber
\nabla _{X}(s, s') = (\nabla _{X}s, s')+(s, \nabla _{X}s'), \,\, \forall X\in V(M), \,\, \forall s, s'\in C^{\infty}(B).
\end{equation}
This implies that, for $s$ a section of modulus one, the potential is pure imaginary,
\begin{equation}\nonumber
\Gamma _{s} = -i\frac{\theta_{s}}{\hbar}
\end{equation}
for some $\theta _{s} \in T^{*}M$. Even though the conection potentials are only defined locally, the curvature form $-i\omega / \hbar$, where $\omega=d\theta_{s}$, is independent of the local section $s$.
\begin{automorphism}
Let $B\rightarrow M$, $L\rightarrow N$ be line bundles with connection and compatible Hermitian structure. A morphism between $B$ and $L$ is a pair $(f,r)$, where $f: M\rightarrow N$ is differentiable and $r:m\mapsto r(m)\in \text{Hom}(L_{f(m)},B_m)$ depends smoothly on $m\in M$. If $s$ is a section of $L$, then the pullback $f^{*}s$ is the section of $B$ defined by
\begin{equation}\nonumber
f^{*}s(m) = r(m)s(f(m)).
\end{equation}
The morphism is said to preserve the connection if
\begin{equation}\nonumber
(\nabla _{X}f^{*}s)(m) = r(m)(\nabla _{f_{*}X}s)(f(m)),\,\, \forall m\in M,
\end{equation}
and to preserve the Hermitian structure if
\begin{equation}\nonumber
(f^{*}s, f^{*}s')(m) = (s, s')(f(m)),\,\, \forall m\in M,
\end{equation}
for any $X \in V(M)$ and $s, s'\in C^{\infty}(L)$. Note that if $f: B\rightarrow B$ is a fibre-preserving diffeomorphism and on each fibre $f:B_{m} \rightarrow B_{f(m)}$ is a linear isomorphism (in local charts, multiplication by a nonzero complex number) then $f$ determines a bundle morphism $(\bar{f},r)$ by $\bar{f}(m) = \pi f(s(m))$ for some section $s$ and $r(m) = (f|_m)^{-1}$. In this case, $f$ is called an \emph{automorphism} of $B$.
\end{automorphism}
Obviously, any automorphism preserving the connection also preserves the curvature form. In the case of interest, when $\omega$ is nondegenerate, it defines a symplectic structure on the base manifold $M$ and hence we see that the automorphisms of $B$ preserving the connection give rise to symplectomorphisms of the base $(M,\omega)$. Physically, it determines a canonical transformation of the classical system. In fact, the following holds:
\begin{exactgroup}
Let $\pi :B\rightarrow M$ be a Hermitian line bundle with compatible connection and nondegenerate curvature form over a connected, simply connected base $M$, $H$ be the group of symplectomorphisms of the base $(M,\omega)$ and $P$ the group of automorphisms of $B$ preserving the connection and the Hermitian structure. Then there is an exact sequence of group morphisms
\begin{equation}\label{groupexact}
1 \longrightarrow S^{1} \longrightarrow P \longrightarrow H \longrightarrow 1 .
\end{equation}

\end{exactgroup}
\begin{proof}
This result means that, as we said above, each automorphism of the line bundle with connection and Hermitian structure gives a symplectomorphism of the base manifold and, moreover, each automorphism is uniquely determined by a choice of one such symplectomorphism plus a complex number of modulus one. Because the bundle automorphisms will correspond to quantum symmetries and the base symplectomorphisms will correspond to classical symmetries, this is an important result for quantization.

The idea of the proof is to use the fact that an automorphism should preserve parallel transport to show how it is determined by the canonical transformation of the base and a complex number. First, recall the concept of parallel transport: let $\gamma :(-1,1)\rightarrow M$ be a smooth path on $M$. Then $\gamma ^{*}B$ is a line bundle over $(-1,1)$ with the connection defined by: if $s$ is a local section over an open $U \subset M$ with connection potential $\Gamma _{s}$, then we take the connection potential associated to $\gamma ^{*}s$ to be $\Gamma _{\gamma ^{*}s} = \gamma ^{*}\Gamma _{s}$. Note that, since $(-1,1)$ is one-dimensional, $\Gamma _{\gamma ^{*}s} = A_{\gamma ^{*}s}dt$ for some function $A_{\gamma ^{*}s}$ of the standard coordinate $t$ on $(-1,1)$. We then say that a section $s' = \psi s$ is parallel along $\gamma$ if $\nabla _{\partial / \partial t}s' = 0$,
\begin{equation}\nonumber
0 = \nabla _{\partial / \partial t}(\psi s) = \frac{\partial}{\partial t}\lrcorner [d(\psi\circ\gamma) \otimes \gamma ^{*}s +(\psi\circ\gamma) A_{\gamma ^{*}s}dt\otimes \gamma ^{*}s] = \left[\frac{d}{dt}\psi + A_{\gamma ^{*}s}\psi \right]\gamma ^{*}s,
\end{equation}
So that giving $\psi (\gamma(0))$ and asking for $\psi (\gamma (t)) \gamma ^{*}s$ to define a section parallel along $\gamma$ is equivalent to an initial value problem whose only solution is
\begin{equation}\nonumber
\begin{split}
\psi (\gamma (t)) &= \psi (\gamma (0))\exp\left(-\int _{0}^{t} A_{\gamma ^{*}s}(t')dt'\right) = \psi (\gamma (0))\exp\left(-\int _{\gamma(0)}^{\gamma(t)} \Gamma _{s}\right) \\
&=\psi (\gamma (0))\exp\left(\frac{i}{\hbar}\int _{\gamma(0)}^{\gamma(t)} \theta _{s}\right).
\end{split}
\end{equation}
The relevant fact for the moment is that the solution is unique.

Now, let $(f,r)$ be an automorphism of $B$, $\gamma :(-1,1)\rightarrow M$ be a smooth curve on $M$ and $f\circ\gamma $ its image curve. Since the map preserves the connection, it should take parallel transport along $\gamma $ to parallel transport along $f\circ \gamma $. Hence, if $\psi_{1}(t) \in B_{f\circ\gamma (t)}$ is the image of $\psi_{1}(0)\in B_{f\circ\gamma (0)}$ by parallel transport along $f\circ\gamma $, than
\begin{equation}\label{parallel}
\psi_{2}(t) = r(\gamma (t))\psi_{1}(t)\in B_{\gamma (t)}
\end{equation}
is the image of $\psi_{2}(0)\in B_{\gamma (0)}$ by parallel transport along $\gamma$. Hence the requirement that the connection is preserved fixes $r(\gamma(t))$ uniquely up to the choice of a nonzero scalar $r_{0} = r(\gamma(0))$. Since $M$ is connected, we can use this to define $r$ by giving its value on any one particular point $m_{0}\in M$: for any other $m\in M$, we join $m_{0}$ to $m$ by a curve $\gamma$ and define $r(m)$ by equation (\ref{parallel}). To do this consistently, one has to verify that the value $r(m)$ thus defined is independent of the path $\gamma$ chosen. Equivalently, for $s$ some local unit section and any closed curve $\gamma$ based at $\gamma(0) = m_{0}$, one must have the consistency condition
\begin{equation}\nonumber
\exp\left(\frac{i}{\hbar}\oint_{\gamma}\theta_{s}\right)r(m_{0})\psi(f\circ\gamma(0)) = r(m_{0})\exp\left(\frac{i}{\hbar}\oint_{f\circ\gamma}\theta_{s}\right)\psi(f\circ\gamma(0)),
\end{equation}
for arbitrary $\psi$. The way this condition is stated, it depends on the local trivialization $s$. A way to fix this is to take a surface $\sigma$ spanning $\gamma$ ($\partial \sigma = \gamma$), which exists since $M$ is simply connected, so that applying Stokes' theorem one gets the condition
\begin{equation}\label{omegainteg}
\exp\left(\frac{i}{\hbar}\int_{\sigma}\omega\right) = \exp\left(\frac{i}{\hbar}\int_{f\circ\sigma}\omega\right) = \exp\left(\frac{i}{\hbar}\int_{\sigma}f^{*}\omega\right).
\end{equation}
Since the surface $\sigma$ is arbitrary and $f$ is continuous, it implies $f^{*}\omega = \omega$, ie., $f$ defines a symplectomorphism of the base $(M, \omega)$.

Finally, note that the requirement that $(f,r)$ should preserve the hermitian structure implies that the scalar $r_{0} = r(m_{0})$ be of unit-modulus. Hence every automorphism of $B$ determines a canonical transformation of the base $(M,\omega)$ and any two automorphisms defining the same canonical transformation differ by a multiplication by a unit-modulus complex number.
\end{proof}

The main idea of prequantization is to identify the exact sequence (\ref{algebraexact}) with the `infinitesimal version' of (\ref{groupexact}), i.e., with the exact sequence induced between the Lie algebras, which then realizes the Poisson bracket-commutator correspondence.

\begin{isompreq}
Let $\pi :B\to M$ be a hermitian line bundle with compatible connection such that its curvature two-form $-i\omega /\hbar$ is nondegenerate. Then $Lie(P) \isoeq C^{\infty}(M)$, where $P$ is the group of automorphisms of $B$ and $C^{\infty}(M)$ is the Poisson algebra on $(M, \omega)$.
\end{isompreq}
\begin{proof}
Let $f_{t}$ be a one-parameter family of automorphisms of $B$ and let $s$ be a local unit section. Since each $f_{t}$ preserves the hermitian structure, $f_{t}^{*}s = e^{i\alpha_{t}}s$. Now, since it also preserves the connection, we must have
\begin{equation}\nonumber
\begin{split}
\left[f_{t}^{*}\left(-\frac{i}{\hbar}\theta_{s}\right)\right]& \otimes f_{t}^{*}s = f_{t}^{*}Ds = D(f_{t}^{*}s) = D(e^{i\alpha_{t}}s) = \left(id\alpha_{t} - \frac{i}{\hbar}\theta_{s}\right)\otimes f_{t}^{*}s \\
&\Rightarrow \,\,\,\, f_{t}^{*}\theta_{s} = -\hbar d\alpha_{t} + \theta_{s} \,\,\,\, \Rightarrow \,\,\,\, \mathcal{L}_{\xi}\theta_{s} = -\hbar d\dot{\alpha},
\end{split}
\end{equation}
where $\xi $ is the infinitesimal generator of $f_{t}$ and $\dot{\alpha} = (d\alpha_{t}/dt)|_{t=0}$. Hence, if one defines $\phi_{\xi} = \hbar\dot{\alpha}+\xi\lrcorner\theta_{s}$, then
\begin{equation}\nonumber
\xi\lrcorner\omega + d\phi_{\xi} = [\xi\lrcorner d\theta_{s} + d(\xi\lrcorner\theta_{s})] + \hbar d\dot{\alpha} = 0,
\end{equation}
so that $f_{t}$ projects to the hamiltonian flow generated by $\phi_{\xi}$ on $(M,\omega)$. Note that $\phi_{\xi}$ is independent of the unit trivialization $s$: if we substitute $s\mapsto e^{iu}s$, then $\theta_{s}\mapsto du + \theta_{s}$ while
\begin{equation}\nonumber
f_{t}^{*}(e^{iu}s) = e^{i(f_{t}^{*}u - u)}e^{i\alpha_{t}}(e^{iu}s),
\end{equation}
so that $\alpha\mapsto\alpha + \xi\lrcorner du$. One can prove that the map $\xi\mapsto\phi_{\xi}$ is bijective \cite{kostant1970quantization}. This is the correspondence between generators of automorphisms of the line bundle and functions on the base symplectic manifold.

From the perspective of quantization, it is more natural to think of the inverse map: to each $f\in C^{\infty}(M)$, one associates the vector
\begin{equation}\label{liftham}
\xi_{f} = X_{f} - \dot{\alpha}\frac{\partial}{\partial\phi} = X_{f} + \frac{(X_{f}\lrcorner\theta_{s} - f)}{\hbar}\frac{\partial}{\partial\phi},
\end{equation}
where we adopt polar coordinates $z = re^{i\phi}$ on the fibres (once in the trivialization determined by s). Compare this with the previous map $\xi\mapsto\phi_{\xi} = h\dot{\alpha}+\xi\lrcorner\theta_{s} = h\dot{\alpha}+X_{\phi_{\xi}}\lrcorner\theta_{s}$. Linearity is clear from expression (\ref{liftham}) and that it is also a Lie algebra morphism follows from
\begin{equation}\nonumber
\begin{split}
[\xi_{f},\xi_{g}] &= [X_{f},X_{g}] + \frac{X_{f}(X_{g}\lrcorner\theta_{s} - g) - X_{g}(X_{f}\lrcorner\theta_{s}-f)}{\hbar}\frac{\partial}{\partial \phi} \\
&= X_{\{f,g\}} +\frac{([X_{f},X_{g}]\lrcorner\theta_{s}+X_{g}\lrcorner\mathcal{L}_{X_{f}}\theta_{s}-\{f,g\}) - (X_{g}\lrcorner d(X_{f}\lrcorner\theta_{s}) - X_{g}\lrcorner(-X_{f}\lrcorner d\theta_{s}))}{\hbar}\frac{\partial}{\partial\phi} \\
&= X_{\{f,g\}} + \frac{X_{\{f,g\}}\lrcorner\theta_{s} - \{f,g\}}{\hbar}\frac{\partial}{\partial\phi} \\
&= \xi_{\{f,g\}}.\nonumber
\end{split}
\end{equation}
\end{proof}

The construction gives a representation of the Poisson algebra $C^{\infty}(M)$ as the generators of automorphisms of $B$. Now, denoting by $\xi^{f}_{t}$ the flow of the vector field $\xi_{f}$ on $B$, we may define an action $\hat{\rho}^{f}_{t}$ on the sections $C^{\infty}(B)$ by
\begin{equation}\label{flowsection}
\xi^{f}_{t}(\hat{\rho}^{f}_{t}s(m)) = s(\rho^{f}_{t}m),
\end{equation}
where $\rho^{f}_{t}$ denotes the flow of $X_{f}$ in $M$. Since sections of this line bundle are like wavefunctions, the infinitesimal action of the flow $\hat{\rho}^f_t$ should be identified with the \emph{quantum operator} $\hat{f}$.

\begin{prequantization}
Let $B\to M$ be a hermitian line bundle with compatible connection and nondegenerate curvature form $-i\omega/\hbar$. Then we define the Hilbert space $\mathcal{H}$ of prequantization to be the space of square integrable sections $s:M\to B$ with inner product
\begin{equation}\nonumber
\langle s, s' \rangle = \int_{M}(s,s')\epsilon,
\end{equation}
where $\epsilon = (\omega/2\pi\hbar)^{n}$ is the canonical volume-form on the $2n$-dimensional symplectic manifold $(M,\omega)$. Given a classical observable $f\in C^{\infty}(M)$, the quantum observable $\hat{f}$ is defined as the generator of the action $\hat{\rho}^{f}_{t}$ through
\begin{equation}\label{absf}
\frac{d\hat{\rho}^{f}_{t}}{dt} = \frac{i}{\hbar}\hat{\rho}^{f}_{t}\hat{f}.
\end{equation}
\end{prequantization}

\begin{preqdirac}
Explicitly, the quantum operators are given by
\begin{equation}\label{expquant}
\hat{f}s = -i\hbar\nabla_{X_{f}}s + fs.
\end{equation}
The map $f\mapsto \mathcal{Q}(f) = \hat{f}$ satisfies the Dirac's quantization conditions (\ref{dirac}).
\end{preqdirac}
\begin{proof}
Let $s' = \psi s$. Then, by the definition (\ref{flowsection}),
\begin{equation}\label{flowsection2}
\begin{split}
\psi(\rho^{f}_{t}m)s(\rho^{f}_{t}m) &= s'(\rho^{f}_{t}m) = \xi^{f}_{t}(\hat{\rho}^{f}_{t}s'(m))\\
\Rightarrow [\hat{\rho}^{f}_{t}(\psi s)](m) &= \psi(\rho^{f}_{t}m)\exp\left(-i\int_{0}^{t}\dot{\alpha}dt'\right)s(m)\\
&=\psi(\rho^{f}_{t}m)\exp\left(-\frac{i}{\hbar}\int_{0}^{t}[(X_{f}\lrcorner\theta_{s} - f)(\rho^{f}_{t'}m)]dt'\right)s(m).
\end{split}
\end{equation}
Substituting in (\ref{absf}),
\begin{equation}\nonumber
\begin{split}
\frac{i}{\hbar}[\hat{\rho}^{f}_{t}\hat{f}(\psi s)](m) &= \frac{d}{dt}[\hat{\rho}^{f}_{t}(\psi s)(m)] = \frac{d}{dt}\left[\psi(\rho^{f}_{t}m)\exp\left(-\frac{i}{\hbar}\int_{0}^{t}[(X_{f}\lrcorner\theta_{s} - f)(\rho^{f}_{t'}m)]dt'\right)s(m)\right] \\
&= \left[X_{f}(\psi) - \frac{i}{\hbar}(X_{f}\lrcorner\theta_{s} - f)\psi\right](\rho^{f}_{t}m)\exp\left(-\frac{i}{\hbar}\int_{0}^{t}(X_{f}\lrcorner\theta_{s}-f)dt'\right)s(m)\\
&=\frac{i}{\hbar}\hat{\rho}^{f}_{t}\left \{\left[-i\hbar\left(X_{f}(\psi) - \frac{i}{\hbar}X_{f}\lrcorner\theta_{s}\psi\right) + f\psi\right]s\right \}(m)\\
&=\frac{i}{\hbar}\{\hat{\rho}^{f}_{t}[-i\hbar\nabla_{X_{f}}(\psi s) + f(\psi s)]\}(m),
\end{split}
\end{equation}
which proves equation (\ref{expquant}).

Let us consider the Dirac rules (\ref{dirac}), linearity is a consequence of the linearity of Hamilton's equation and of expression (\ref{expquant}); also, if $f = c$ is a constant, then
\begin{equation}\nonumber
X_{f}\lrcorner\omega = -df = 0 \,\, \Rightarrow \,\, X_{f} = 0 \,\, \Rightarrow \hat{f}s = fs.
\end{equation}
We evaluate the commutator explicitly, as this is the reason for introducing the $-i/\hbar$ factor in the expressions for the curvature form:
\begin{equation}\nonumber
\begin{split}
[\mathcal{Q}(f),\mathcal{Q}(g)]s &= \hat{f}\hat{g}s - \hat{g}\hat{f}s 
= -\hbar^{2}[\nabla_{X_{f}},\nabla_{X_{g}}]s -i\hbar [X_{f}(g)s - X_{g}(f)s + g\nabla_{X_{f}}s - f\nabla_{X_{g}}s] +\\
 & \hspace{6.7cm} + f[-i\hbar\nabla_{X_{g}}s + gs] - g[-i\hbar\nabla_{X_{f}}s + fs]\\
&= -\hbar^{2}\left(\nabla_{[X_{f},X_{g}]}s - \frac{2i\omega(X_{f},X_{g})}{\hbar}s\right)-i\hbar\{f,g\}s+i\hbar X_{g}\lrcorner df s\\
&=-i\hbar[-i\hbar\nabla_{X_{\{f,g\}}}s-X_{g}\lrcorner (X_{f}\lrcorner\omega + df)s+\{f,g\}s]\\
&=-i\hbar\mathcal{Q}(\{f,g\})s,
\end{split}
\end{equation}
where we wrote the curvature form in terms of $\omega$ in the second line and used Hamilton's equation for $f$ in the last step.

\end{proof}

Therefore, provided one can construct a prequantum bundle over phase space, the above procedure gives a solution of Dirac's axioms on the space of sections of this bundle. However, it is not true that one can find one such bundle for any given symplectic manifold. One extra necessary condition is the following. From the discussion preceding equation (\ref{omegainteg}), parallel transport along a closed curve $\gamma$ amounts to a transformation on the fibre given by multiplication by
\begin{equation}\nonumber
\exp \left(\frac{i}{\hbar}\oint_{\gamma}\theta_{s}\right),
\end{equation}
in a local frame $s$. If there are two surfaces $\sigma$ and $\sigma '$ intersecting only on the common boundary $\gamma$, this implies
\begin{equation}\nonumber
\begin{split}
\exp\left(\frac{i}{\hbar}\int_{\sigma}\omega\right) = \exp \left(\frac{i}{\hbar}\oint_{\gamma}\theta_{s}\right) = \exp\left(\frac{i}{\hbar}\int_{\sigma '}\omega\right) \,\,\,
\Rightarrow \,\,\, \exp\left(\frac{i}{\hbar}\oint_{\bar{\sigma} \cup \sigma '}\omega\right) = 1,
\end{split}
\end{equation}
where the bar in $\bar{\sigma}$ denotes inverting the orientation. Hence it is necessary that the integral of $\omega$ on any closed surface in $M$ be an integer multiple of $2\pi\hbar$, which is known in Physics as the \emph{Bohr-Sommerfeld quantization rule}. The more technical statement is

\begin{preqcrit}\label{critpreq}
Given a symplectic manifold $(M,\omega)$, there exists a prequantum bundle over $M$, i.e., a Hermitian line bundle with compatible connection whose curvature form is $-i\omega /\hbar$, if, and only if, $(2\pi\hbar)^{-1}\omega$ defines an integer 2nd cohomology class. Moreover, if this is satisfied, then the inequivalent choices of bundle and connection are parametrized by $H^{1}(M,S^{1})$.\footnote{A simple proof can be found in \cite{woodhouse1997geometric}.}
\end{preqcrit}
Because many of our examples deal with symplectic reduction, the following test will be more straightforward to apply:

\begin{preqreduction}\label{reducequantum}
Let $(M,\omega)$ be the reduction of some presymplectic manifold $(M', d\theta ')$, $\theta '\in \Omega^{1}(M')$. If, for any closed curve $\gamma $ contained in a leaf of the characteristic foliation of $d\theta '$,
\begin{equation}\label{intred}
\frac{1}{2\pi\hbar}\int_{\gamma}\theta '\in\mathbb{Z},
\end{equation}
then $(M,\omega)$ admits a prequantum bundle. Moreover, if $M'$ is simply connected, then this is also a necessary condition.
\end{preqreduction}

\begin{proof}
A direct proof of the implication can be given by explicitly constructing a prequantum bundle. Let $K_{m'} = \{X\in T_{m'}M'|X\lrcorner d\theta ' = 0\}$ be the characteristic foliation and $\pi :M'\to M$ be the reduction map. Then we define $B\to M$ to be the bundle whose fiber $B_{m}$ is the space of smooth complex functions $\psi :\pi^{-1}(m)\to \mathbb{C}$ such that, for any piecewise smooth path $\gamma :(-1,1)\times \pi^{-1}(m) \to \pi^{-1}(m): (t,m')\mapsto \gamma_{t}m'$,
\begin{equation}\label{leaves}
\psi (\gamma_{t}m') = \psi (m')\exp \left(\frac{i}{\hbar}\int_{m'}^{\gamma_{t}m'}\theta '\right).
\end{equation}
Clearly, evaluation at some point in the leaf $\pi^{-1}(m)$ provides an isomorphism $B_{m}\to \mathbb{C}$, so the fibres are one-dimensional. Note how this construction is well-defined because different choices of the path $\gamma $ differ by a factor of $\exp (i2\pi n), \,\, n \in \mathbb{Z}$ by assumption. Therefore sections of this bundle are given by colections of functions on all of the leaves of $K$, which then add up to complex functions $\psi : M'\to \mathbb{C}$, such that their restrictions to leaves of $K$ satisfy (\ref{leaves}). Now, complex functions on $M'$ are sections of the trivial bundle $M'\times \mathbb{C}$. Also, let $X\in V_{K}(M')$ be a vector tangent to the characteristic foliation and let $\gamma _{t}$ be its flow. Then, by (\ref{leaves}),
\begin{equation}\nonumber
[X(\psi)](m') = \frac{d}{dt}\psi(\gamma_{t}m') = \frac{d}{dt}\left[ \psi (m')\exp \left(\frac{i}{\hbar}\int_{m'}^{\gamma_{t}m'}\theta '\right) \right] = \frac{i}{\hbar}(X\lrcorner\theta ')\psi (m').
\end{equation}
This means that, if we take $\nabla '$ to be the connection on $M\times \mathbb{C}$ with potential $-i\theta '/\hbar $ in the trivialization $e : m'\mapsto (m', 1)$, then each section of our line bundle $B\to M$ is given by a function $\psi : M'\to \mathbb{C}$ such that
\begin{equation}\nonumber
\nabla '_{X}(\psi e) = \left[ X(\psi) - \frac{i}{\hbar}(X\lrcorner\theta ')\psi \right]e = 0, \,\,\, \forall X\in V_{K}(M').
\end{equation}

We define the connection on $B$ by using $\nabla '$. Let $\psi : M'\to \mathbb{C}$ be the function corresponding to a section $s\in C^{\infty}(B)$. Then $\nabla _{Y}s$ is the section corresponding to the function $\phi$ such that $(\phi e) = \nabla _{Z}'(\psi e)$, where $\pi _{*}Z = Y$. This is independent of the choice of $Z\in V(M')$ since
\begin{equation}\nonumber
\pi_{*}Z = \pi_{*}W \,\, \Rightarrow \,\, \exists X\in V_{K}(M')| W = Z + X \,\, \Rightarrow \,\, \nabla _{W}'(\psi e) = \nabla _{Z + X}'(\psi e) = \nabla_{Z}'(\psi e),
\end{equation}
since $\psi $ should define a section of $B$. To see that this definition maps sections of $B$ to sections of $B$, we use that the curvature of $\nabla '$ is $-i (d\theta ')/\hbar$, so that, if $X\in V_{K}(M')$, $Z\in V(M')$, and $\nabla_{X}(\psi e) = 0, \forall X\in V_{K}(M')$, then
\begin{equation}\nonumber
\nabla_{X}'[\nabla_{Z}' (\psi e) ] = \nabla_{Z}'[\nabla_{X}'(\psi e)] + \nabla_{[X,Z]}'(\psi e) - \frac{2i}{\hbar}d\theta '(X,Z)e = 0,
\end{equation}
since $X\in V_{K}(M')$ ($\Rightarrow X\lrcorner d\theta ' = 0$) and $[X,Z]\in V_{K}(M')$ as $K$ is integrable. Therefore $\nabla_{Z}(\psi e)$ also defines a section of $B$ if $(\psi e)$ does.

To calculate the curvature of $\nabla $ we need only that $\pi _{*}[Z,W] = [\pi_{*}Z, \pi_{*}W]$. For $Z, W\in V(M')$ and $s$ the section of $B$ associated with $\psi$,
\begin{equation}\nonumber
\nabla _{\pi_{*}Z}\nabla_{\pi_{*}W} s - \nabla _{\pi_{*}W}\nabla_{\pi_{*}Z} s - \nabla_{[\pi_{*}Z,\pi_{*}W]}s 
\end{equation}
should correspond to the function $\phi$, where
\begin{equation}\nonumber
\begin{split}
(\phi e) &= \nabla_{Z}'[\nabla_{W}' (\psi e) ] - \nabla_{W}'[\nabla_{Z}'(\psi e)] - \nabla_{[Z,W]}'(\psi e) = \frac{2i}{\hbar}d\theta '(Z,W)e \\
&= \frac{2i}{\hbar}\omega(\pi_{*}Z,\pi_{*}W)e,
\end{split}
\end{equation}
so the curvature is indeed given by the symplectic form on $(M,\omega)$. Finally, we take the hermitian structure on $B$ to be the one inherited from the standard one in $M'\times\mathbb{C}$, where $(e,e) = 1$.

Conversely, if $(M, \omega)$ admits a prequantum bundle, we have seen that necessarily the integral of $\omega$ on a closed surface is an integer multiple of $2\pi\hbar$. Let $\gamma$ be a closed curve on a leaf $\pi^{-1}(m)$ of $K$. If one assumes that $M'$ is simply connected, then one can find a surface $\sigma \subset M' $ such that $\partial \sigma = \gamma$. Now, the image of $\gamma \subset \pi^{-1}(m)$ is $m$, so $\pi(\sigma)\in M$ is closed. Thus
\begin{equation}\nonumber
\frac{1}{2\pi\hbar}\oint_{\gamma}\theta ' = \frac{1}{2\pi\hbar}\int_{\sigma}d\theta ' = \frac{1}{2\pi\hbar}\int_{\pi(\sigma)}\omega \in \mathbb{Z},
\end{equation}
so equation (\ref{intred}) follows.

\end{proof}

\begin{preqsu2}\emph{Quantization of spin}

The existence of a prequantum bundle implies a quantization of spin. Let us first take elementary systems with rotational symmetry. Recall from equations (\ref{potencialsu2}) and (\ref{s3tos2}), that each coadjoint orbit can be expressed as a sphere of radius $s$, which we think of as the symplectic reduction of 
\begin{equation}\nonumber
(S^{3},\omega_{f} = d\theta_{f}), \,\,\,\,\,\,\,\,\, \theta_{f} = is(z^{0}d\bar{z}^{0}+z^{1}d\bar{z}^{1}-\bar{z}^{0}dz^{0}-\bar{z}^{1}dz^{1}).
\end{equation}

First we find out what is the characteristic foliation $K$. Let $X^{z^{0}}\partial_{z^{0}} + X^{\bar{z}^{0}}\partial_{\bar{z}^{0}} + X^{z^{1}}\partial_{z^{1}} + X^{\bar{z}^{1}}\partial_{\bar{z}^{1}} \in V(\mathbb{C}^{2})$ be some vector tangent to $S^{3}$. Then the solution to
\begin{equation}\nonumber
\frac{d}{dt}z^{\alpha}(t) = X^{\alpha}(z(t))
\end{equation}
should be a curve in the sphere, ie.,
\begin{equation}\nonumber
\begin{split}
z^{0}(t)\bar{z}^{0}(t) + z^{1}(t)\bar{z}^{1}(t) = 1 & \,\,\,\, \Rightarrow \,\,\,\, 0 = \frac{d}{dt}(z^{0}\bar{z}^{0} + z^{1}\bar{z}^{1}) = 2Re(\bar{z}^{0}X^{z^{0}} + \bar{z}^{1}X^{z^{1}})\\
&\Rightarrow \,\,\,\, \bar{z}^{0}X^{z^{0}} + \bar{z}^{1}X^{z^{1}} = if,
\end{split}
\end{equation}
for some real $f$. Now, let $X\in K$, the characteristic foliation of $\omega_{f}$. This means that $(X\lrcorner\omega_{f})|_{TS^{3}} = 0$, ie, $Y\lrcorner(X\lrcorner\omega_{f}) = 0, \,\, \forall Y\in TS^{3}$. Taking into account the previous discussion, we may write both the vectors tangent to the sphere in the form
\begin{equation}\nonumber
\begin{split}
X &= \xi\frac{\partial}{\partial z^{0}} + \bar{\xi}\frac{\partial}{\partial \bar{z}^{0}} + \frac{(if - \bar{z}^{0}\xi)}{\bar{z}^{1}}\frac{\partial}{\partial z^{1}} - \frac{(if + z^{0}\bar{\xi})}{z^{1}}\frac{\partial}{\partial\bar{z}^{1}}\\
Y &= \chi\frac{\partial}{\partial z^{0}} + \bar{\chi}\frac{\partial}{\partial \bar{z}^{0}} + \frac{(ig - \bar{z}^{0}\chi)}{\bar{z}^{1}}\frac{\partial}{\partial z^{1}} - \frac{(ig + z^{0}\bar{\chi})}{z^{1}}\frac{\partial}{\partial\bar{z}^{1}}
\end{split}\,\,\,\, ,
\end{equation}
with $f$ and $g$ real. Hence the condition for $X$ to be in $K$ translates to
\begin{equation}\nonumber
\begin{split}
0 &= Y\lrcorner(X\lrcorner\omega_{f}) = Y\lrcorner2is\left[\xi d\bar{z}^{0} - \bar{\xi}dz^{0} + \frac{(if - \bar{z}^{0}\xi)}{\bar{z}^{1}}d\bar{z}^{1} + \frac{(if + z^{0}\bar{\xi})}{z^{1}}dz^{1}\right] \\
&= \frac{2is}{z^{1}\bar{z}^{1}}\left[\bar{\chi}(\xi - ifz^{0}) - \chi(\bar{\xi} + if\bar{z}^{0}) + ig(z^{0}\bar{\xi} + \bar{z}^{0}\xi)\right], \,\,\, \forall \chi, \, \forall g
\end{split}
\end{equation}
which we solve by $\xi = ifz^{0}$. So the vectors in the characteristic foliation are of the form
\begin{equation}\nonumber
X = ifz^{0}\frac{\partial}{\partial z^{0}} - if\bar{z}^{0}\frac{\partial}{\partial \bar{z}^{0}} + ifz^{1}\frac{\partial}{\partial z^{1}} - if\bar{z}^{1}\frac{\partial}{\partial \bar{z}^{1}}.
\end{equation}
Finally, solving $\frac{d}{dt}z^{\alpha}(t) = X^{\alpha}(z(t))$, we find that the leaves of $K$ are the circles\footnote{This is a well-known fact about the Hopf fibration.}
\begin{equation}\nonumber
e^{i\phi(t)}(z^{0},z^{1}), \,\,\,\,\,\, \phi (t) = \int_{0}^{t}f dt \,\,\,\, \text{real}.
\end{equation}

Therefore, since $S^{3}$ is simply connected, there is a prequantum bundle over the reduction of $(S^{3},d\theta_{f})$ if, and only if, $\oint_{\gamma}\theta_{f}$ is an integer multiple of $2\pi\hbar$ for any closed curve $\gamma$ contained in a leaf of $K$. Hence, taking the parametrization $\gamma = \{e^{it}(z^{0},z^{1}), t\in [0,2\pi]\}$, this translates to
\begin{equation}\nonumber
\begin{split}
\frac{1}{2\pi\hbar}\oint_{\gamma}\theta_{f} &= \frac{is}{2\pi\hbar}\oint_{\gamma}(z^{0}d\bar{z}^{0}+z^{1}d\bar{z}^{1}-\bar{z}^{0}dz^{0}-\bar{z}^{1}dz^{1}) \\
&= \frac{is}{2\pi\hbar}\int_{0}^{2\pi}[e^{it}z^{0}d(e^{-it}\bar{z}^{0}) + e^{it}z^{1}d(e^{-it}\bar{z}^{1}) - e^{-it}\bar{z}^{0}d(e^{it}z^{0}) - e^{-it}\bar{z}^{1}d(e^{it}z^{1})] \\
&= \frac{s}{\hbar}(z^{0}\bar{z}^{0} + z^{1}\bar{z}^{1} + \bar{z}^{0}z^{0} + \bar{z}^{1}z^{1}) = \frac{2s}{\hbar}\in\mathbb{Z},
\end{split}
\end{equation}
so we arrive at the conclusion that only certain spheres can be quantized, the ones with radius an integer multiple of $\hbar/2$. 

Consider now the phase-space of a massive particle of arbitrary spin $s$. As discussed in example \ref{poincareorbits}, this is the reduction of $C_{sm}$ with symplectic structure given by equation (\ref{omegaspin}). Fortunately, it is not necessary to make much effort to discover what the reduction map is, since we saw that it is explicitly given by the quotient
\begin{equation}\nonumber
\pi : C_{sm} = \{(p,q, z)\in T^{*}\mathbb{M}\times\mathbb{S}| p_{a}p^{a} = m^{2}, \sqrt{2}p_{A\bar{A}}z^{A}\bar{z}^{\bar{A}}=\pm m\} \longrightarrow C_{sm}/\sim
\end{equation}
with $(p_{a},q^{b},z^{C})\sim (p_{a}, q^{b}+\lambda p^{b},e^{i\phi}z^{C}), \,\, \forall \lambda, \phi \in \mathbb{R}$. It turns out that the only closed curves which pose some restriction on $M_{sm}$ are the circles in spinor space given by $\{e^{it}z^{C}|t\in [0,2\pi]\}$, which are of the same form as the ones in the case of $S^{3}$ above if we identify the spinors with $\mathbb{C}^{2}$. As was showed in example \ref{poincareorbits}, the restriction of the reduction $C_{sm}\to M_{sm}$ to these directions is the same as the reduction $S^{3}\to S^{2}$, so the same calculation as above implies that $M_{sm}$ is quantizable if, and only if, $s$ is an integer multiple of $\hbar/2$. Thus we recover quantization of spin as a topological obstruction to the construction of a prequantum bundle.

In the massless case, we apply the same criterion to the reduction of $(C_{s0},-id\omega^{A}\wedge d\bar{\pi}_{A}+id\bar{\omega}^{\bar{A}}\wedge d\pi_{\bar{A}}|_{C_{s0}})$. The reduction map was seen to be the quotient
\begin{equation}\label{2s}
\pi : C_{s0}=\{(\omega,\pi)\in\mathbb{S}\times\bar{\mathbb{S}}|\omega^{A}\bar{\pi}_{A}+\bar{\omega}^{\bar{A}}\pi_{\bar{A}}=2s\}\longrightarrow C_{s0}/\sim 
\end{equation}
by the relation $(\omega^{A},\pi_{\bar{A}})\sim (e^{i\phi}\omega^{A},e^{i\phi}\pi_{\bar{A}}), \,\,\, \phi\in\mathbb{R}$. By a similar calculation, one can check whether the symplectic potential
\begin{equation}\nonumber
\theta '' = -i\omega^{A}d\bar{\pi}_{A} + i\bar{\omega}^{\bar{A}}d\pi_{\bar{A}}
\end{equation}
gives an integer multiple of $2\pi\hbar$ when integrated on a path of the form $\gamma = \{(e^{it}\omega^{A},e^{it}\pi_{\bar{A}})|t\in [0,2\pi]\}$:
\begin{equation}\nonumber
\begin{split}
\frac{1}{2\pi\hbar}\oint_{\gamma}\theta '' &= \frac{-i}{2\pi\hbar}\int_{0}^{2\pi}[e^{it}\omega^{A}d(e^{-it}\bar{\pi}_{A}) - e^{-it}\bar{\omega}^{\bar{A}}d(e^{it}\pi_{\bar{A}})] = \frac{-1}{\hbar}(\omega^{a}\bar{\pi}_{A}+\bar{\omega}^{\bar{A}}\pi_{\bar{A}}) \\
&= -\frac{2s}{\hbar}\in\mathbb{Z},
\end{split}
\end{equation}
by equation (\ref{2s}). We conclude that the coadjoint orbits corresponding to massless particles which are quantizable are the ones with helicity $s$ which is an integer multiple of $\hbar/2$.

\pagebreak

\subsection{Quantization}\label{subsecquant}

Let us reflect on what we have so far. Dirac's rules ask us to find a representation of the algebra of classical observables $C^{\infty}(M)$ as operators in a Hilbert space, with the classical Poisson bracket mapping to the commutator under this correspondence. Furthermore, there should be some correspondence between the classical symmetries generated by classical observables and the corresponding quantum symmetries generated by their quantum counterparts. Such a representation is not guaranteed to exist, and is not guaranteed to be unique. As we saw above, prequantization partially solves these questions: it gives a topological criterion to diagnose whether a given classical phase space (and its algebra of observables) can be quantized, and gives an explicit construction of the Hilbert space as the space of sections of the prequantum bundle, with an explicit definition of how the quantized operators act on it.

Unfortunately, that's not enough. Prequantization solves Dirac's axioms at the cost of making the Hilbert space too large. For example, the constructed Hilbert space includes sections which, in local canonical coordinates, depend on \emph{all} coordinates \emph{and} momenta, which is not in accordance with Heisenberg's uncertainty principle. The next step, called \emph{Quantization}, introduces a geometric criterion to restrict the sections to those which depend on half of the canonical coordinates through the concept of a \emph{polarization}. This, on the other hand, has the tradeoff of naturally restricting the observables that can be quantized to the subalgebra of $C^{\infty}(M)$ of observables which preserve the polarization. Such a subtlety is well-known in quantum mechanics: in canonical quantization, for example, it appears as ordering ambiguities which make it impossible to realize Dirac's commutator correspondence for polynomials of arbitrary degree in the canonical variables.

For pedagogical reasons, it is natural to start with real polarizations.

\begin{realpolarization}
A \emph{real polarization}\footnote{$V_P(M)$ denotes vector fields in $M$ which are tangent to the polarization, that is $X|_m\in P_m,\,\, \forall m\in M$.} $V_P(M)$ on a symplectic manifold $(M,\omega)$ is a smooth distribution $m\mapsto P_{m}$ which is
\begin{description}
\item[(i)] Integrable: $X, Y\in V_{P}(M) \Rightarrow [X,Y]\in V_{P}(M)$ ,
\item[(ii)] Lagrangian: $\forall m\in M, \,\,\, P_{m} \text{ is a Lagrangian subspace of } T_{m}M$.
\end{description}
Given a polarization $P$, we denote by $C^{\infty}_{P}(M) = \{f\in C^{\infty}(M)|X(f) = 0, \forall X\in V_{P}(M)\}$ the set of \emph{polarized functions}.
\end{realpolarization}

This can be more concisely stated as follows. $P$ is a foliation and that each one of its leaves is a Lagrangian submanifold of $M$. The prototype is the vertical foliation of a cotangent bundle: let $M = T^{*}Q$ with the canonical coordinates $(p_{a},q^{b})$ and let $P$ be the distribution spanned by the vector fields $\partial /\partial p_{a}$. This is obviously integrable, the leaves being the cotangent spaces $T^{*}_{q}Q$. Each one of the leaves is isotropic since the restriction of the canonical two-form $\omega = dp_{a}\wedge dq^{a}$ to any one of the surfaces of constant $q$ vanishes. Because the dimension of the leaves is half the dimension of $M$, they are actually Lagrangian submanifolds. Hence this is a polarization. Note that if one considers only sections of a prequantum bundle which are constant along the leaves of this polarization, the local representation of the sections will be given by complex wavefunctions on the space of leaves of $P$, which is indeed the configuration space $Q$. This additional constraint is what will be required for quantisation. 

Note, the introduction of a real polarization in a general symplectic manifold $M$ can effectivelly be seen as a splitting of $M$ into position and momentum directions. To show this, we use that any real polarization comes with a natural flat affine connection on its leaves.

\begin{connectionleaf}
Let $P$ be a polarization of a symplectic manifold $(M,\omega)$. Then the map $\nabla : V_{P}(M)\times V_{P}(M)\to V_{P}(M) : (X,Y)\mapsto \nabla_{X}Y$ defined by
\begin{equation}\nonumber
(\nabla_{X}Y)\lrcorner\omega = \mathcal{L}_{X}(Y\lrcorner\omega)
\end{equation}
defines a flat torsionless affine connection on each leaf of $P$. That is, for every $X, Y, Z \in V_{P}(M)$, and any $f\in C^{\infty}(M)$,
\begin{description}
\item[(i)] $\nabla_{X}(fY) = f\nabla_{X}Y + X(f)Y$,
\item[(ii)] $\nabla_{X}Y - \nabla_{Y}X - [X,Y] = 0$,
\item[(iii)] $\nabla_{X}\nabla_{Y}Z - \nabla_{Y}\nabla_{X}Z - \nabla_{[X,Y]}Z = 0$.
\end{description}
\end{connectionleaf}
\begin{proof}
That the definition is unambiguous is due to the fact that $\omega$ is nondegenerate. Also, if $X,Y,Z\in V_{P}(M)$, then
\begin{equation}\nonumber
\begin{split}
\omega(\nabla_{X}Y,Z) &= \frac{1}{2}Z\lrcorner(\nabla_{X}Y\lrcorner\omega) = \frac{1}{2}Z\lrcorner\mathcal{L}_{X}(Y\lrcorner\omega) = \frac{1}{2}\{\mathcal{L}_{X}[Z\lrcorner(Y\lrcorner\omega)] - \mathcal{L}_{X}Z\lrcorner(Y\lrcorner\omega)\} \\
&= X\lrcorner d[\omega(Y,Z)] - \omega(Y,[X,Z]) = 0,
\end{split}
\end{equation}
since both $Z$ and $[X,Z]$ belong to $V_{P^{\perp}}(M)$, as $P = P^{\perp}$ is integrable. Therefore $\nabla_{X}Y\in V_{P}(M)$.
Now, {\bf (i)}-{\bf (iii)} follow easily from the definition:
\begin{description}
\item[(i)] $\nabla_{X}(fY)\lrcorner\omega = \mathcal{L}_{X}(fY\lrcorner\omega) = X(f)Y\lrcorner\omega + f\mathcal{L}_{X}(Y\lrcorner\omega) = (X(f)+f\nabla_{X}Y)\lrcorner\omega$,
\item[(ii)] $(\nabla_{X}Y - \nabla_{Y}X)\lrcorner\omega = X\lrcorner d(Y\lrcorner\omega) - Y\lrcorner d(X\lrcorner\omega) = X\lrcorner\mathcal{L}_{Y}\omega - Y\lrcorner\mathcal{L}_{X}\omega $

$= [X,Y]\lrcorner\omega + 2d[\omega(X,Y)] = [X,Y]\lrcorner\omega$,
\item[(iii)] $(\nabla_{X}\nabla_{Y}Z - \nabla_{Y}\nabla_{X}Z)\lrcorner\omega = (\mathcal{L}_{X}\mathcal{L}_{Y} - \mathcal{L}_{Y}\mathcal{L}_{X})(Z\lrcorner\omega) = \mathcal{L}_{[X,Y]}(Z\lrcorner\omega) = \nabla_{[X,Y]}Z\lrcorner\omega$.
\end{description}

\end{proof}

\begin{polarizationcot}

Let $P$ be a polarization of a symplectic manifold $(M,\omega)$ and $m\in M$. Then there exists a neighbourhood $U$ of $m$ and a symplectomorphism $\rho : U'\to U$, where $U'\subset T^{*}Q$ is an open neighbourhood of the zero section in the cotangent bundle of some manifold $Q$, such that $\rho^{*}P$ is the vertical foliation of $U'$ and $\rho^{-1}(m)$ lies in the zero section in $U'$.
\end{polarizationcot}
\begin{proof}
By Darboux's theorem, one can find local canonical coordinates $(p_{a},q^{b})$ around $m$ such that the surface of constant $q^{b} = q^{b}(m)$ renders a Lagrangian submanifold. Up to a linear canonical transformation on these coordinates, we may assume that the submanifold thus constructed is transverse to $P$ at $m$, which then implies that it is transverse to $P$ in a neighbourhood of $m$. Let us denote this submanifold by $Q$. We further restrict the neighbourhood $U$ of $m$ so that each leaf of the induced polarization $P|_{U}$ is geodesically convex and intersects $Q$ in a unique point.

Let us denote by $\Sigma_{q}$ the leaf of $P|_{U}$ which intersects $Q$ at the point $q$. By the last proposition, there is a flat affine connection on each $\Sigma_{q}$, which gives it the structure of an affine space associated to $P_{q}$. Together with the definition that the point of intersection $q$ shoud be taken as the origin, each $\Sigma_{q}$ gains a vector space structure and hence may be identified with a neighbourhood of the origin in $P_{q}$.

Consider the map $X\mapsto X\lrcorner\omega$. Since $P$ is Lagrangian, $X\lrcorner\omega$ is zero on $P_{q}$. But $\omega$ is non-degenerate, so this should actually give a well-defined one form on the transverse space $T_{q}Q$. Hence we have an isomorphism $P_{q}\to (T_{q}Q)^{*} = T^{*}_{q}Q$. Therefore it is possible to identify $\Sigma_{q}$ with a neighbourhood of zero in $T^{*}_{q}Q$. Letting $q$ vary over $U$ gives the map $\rho : U'\to U$, where $U'$ is some neighbourhood of the zero section in $T^{*}Q$. Note that it follows immediately from the way we defined $\rho $ that $\rho^{*}P|_{U}$ gives the vertical foliation on $U'\subset T^{*}Q$ and that the linear structures coincide. It remains to show that $\rho$ preserves the symplectic structure.

That the symplectic $\omega$ coincide with the canonical two-form $\omega '$ of $T^{*}Q$ at points of $Q$ follows again from the linear form of Darboux's theorem. Now, given a polarized function $f\in C^{\infty}_{P}(U)$, which means that $f$ is constant in each of the leaves of $P|_{U}$, we define $X_{f}\in U$ and $X_{f}' \in U'$ by
\begin{equation}\nonumber
X_{f}\lrcorner\omega + df = 0, \,\,\,\,\,\, X_{f}'\lrcorner\omega ' + d(f\circ\rho) = 0.
\end{equation}
Note that $X_{f}' = \rho_{*}X_{f}$ on $Q$, since the symplectic structures coincide there. Now, because $f$ is constant along $P|_{U}$, $X_{f}$ is parallel to $P|_{U}$. But this implies that
\begin{equation}\nonumber
\nabla_{Y}X_{f}\lrcorner\omega = \mathcal{L}_{Y}(X_{f}\lrcorner\omega) = Y\lrcorner d(-df) = 0, \,\, \forall Y\in V_{P}(U),
\end{equation}
i.e., $X_{f}$ is covariantly constant along $P|_{U}$. Similarly, $X_{f}'$ is covariantly constant along $(\rho^{*}P)|_{U'}$. So actually $X_{f}' = \rho_{*}X_{f}$ everywhere. Furthermore, since $X_{f}$ and $X_{f}'$ are hamiltonian,
\begin{equation}\nonumber
\mathcal{L}_{X_{f}}\omega = 0, \,\,\, \mathcal{L}_{X_{f}'}\omega ' = 0 \,\,\,\, \Rightarrow \,\,\,\, \mathcal{L}_{X_{f}' }(\omega ' - \rho^{*}\omega) = \mathcal{L}_{X_{f}'}\omega ' - \rho^{*}\mathcal{L}_{X_{f}}\omega = 0.
\end{equation}
Since any point in $M$ can be connected to a point in $Q$ by the flow generated by some $f\in C^{\infty}_{P}(M)$, $\omega = \rho^{*}\omega $ everywhere, ie., $\rho$ is a symplectomorphism.

\end{proof}

\begin{polarizationcoord}
Let $P$ be a real polarization of a symplectic manifold $(M,\omega)$. Then in the neighbourhood of any $m\in M$ it is possible to find a canonical coordinate system $(p_{a},q^{b})$ such that $P$ is spanned by the vector fields $\partial/\partial p_{a}$, and a symplectic potential $\theta$ such that $\theta|_{P} = 0$. These coordinates and potential are said to be \emph{adapted} to $P$.
\end{polarizationcoord}

To summarize, even though a general symplectic manifold is always locally symplectomorphic to a cotangent bundle, it is not necessarilly globally equivalent. Still, if it admits a real polarization then the concept of a polarized section allows us to generalise the idea of functions that depend only on either coordinates or momenta. Locally, this looks like the local equivalence with a cotangent bundle, but the importance of these definitions is in that a real polarization may exist even when $M$ is \emph{not} a cotangent bundle.

Up to this point, the polarization had only the role of selecting the relevant physical states. It also accomplishes the task of selecting the right observables to be quantized (in canonical quantization a similar choice has to be made in order to avoid ordering ambiguities). For example, one can restrict to operators which are constant along the leaves of the polarization, which are denoted by $C^{\infty}_P(M)$. More generally, the following definitions are useful.

\begin{polinomial}
Let $P$ be a real polarization of the symplectic manifold $(M,\omega)$. Then the \emph{polynomial observables} of degree $k$ on an open set $U\subset  M$, denoted $S_{P}^{k}(U)$, are defined recursively by: $S_{P}^{0}(U) = C_{P}^{\infty}(U)$ and
\begin{equation}\nonumber
S_{P}^{k}(U) = \{f\in C^{\infty}(U)| \{f,g\}\in S_{P}^{k-1}(U\cap V),\,\, \forall V\subset\circ M,\,\, \forall g\in C_{P}^{\infty}(V)\}.
\end{equation}
In particular, $S_{P}^{0}(M)$ are the generators of the Hamiltonian vector fields tangent to $P$ and $S_{P}^{1}(M)$ are the generators of the Hamiltonian flows that preserve $P$.
\end{polinomial}

Indeed, if we consider local canonical coordinates and take $P$ to be the polarization spanned by the $\partial/\partial p_a$, then $f\in S_{P}^{0}(U) = C_{P}^{\infty}(U)$ is given by a function of the $q^{a}$ only (it is constant along the leaves of $P$). Since in these canonical coordinates $\omega = dp_{a}\wedge dq^{a}$, it follows that $X_{f}$ is a linear combination of the $\partial/\partial p_{a}$, which span $P$. Regarding $S_{P}^{1}(U)$, we formalize the statement that the canonical flow generated by $f$ preserves $P$ by saying that the differential operator $X_{f}:C^{\infty}(U)\to C^{\infty}(U)$ preserves the polarized functions, ie., $X_{f}(C_{P}^{\infty}(U))\subset C_{P}^{\infty}(U)$. This is indeed the case since
\begin{equation}\nonumber
f\in S_{P}^{1}(U)\,\,\,\,\Rightarrow\,\,\,\, X_{f}(g) = \{f,g\}\in S_{P}^{0}(U) = C_{P}^{\infty}(U), \,\, \forall g\in C_{P}^{\infty}(U).
\end{equation}
The local expression of these observables justifies their name:

\begin{polinomialindeed}
For a given open neighborhood $U\subset\circ M$, $f\in S_{P}^{k}(U)$ if, and only if, $f$ is of the form
\begin{equation}\nonumber
f = \sum_{i=0}^{k}f_{i}^{a_{1}a_{2}...a_{i}}p_{a_{1}}p_{a_{2}}...p_{a_{i}},
\end{equation}
where each $f_{i}^{a_{1}a_{2}...a_{i}}$ is independent of the momentum coordinates $p_{a}$.
\end{polinomialindeed}
\begin{proof}
($\Rightarrow$) By induction on $k$:
\begin{description}
\item[(i)] $f\in S_{P}^{0}(U) = C_{P}^{\infty}(U) \,\, \Rightarrow \,\, f = f(q) = f_{0}$ ;
\item[(ii)] Suppose the proposition holds for $k\leq \tilde{k}$;
\item[(iii)] Take $f\in S_{P}^{\tilde{k}+1}(U)\,\, \Rightarrow \,\, \forall g\in C_{P}^{\infty}(U), \,\, \{f,g\} = : h \in S_{P}^{\tilde{k}}(U)$, so that, by the inductive hypotheses,
\begin{equation}\nonumber
\sum_{i=0}^{\tilde{k}}h_{i}^{a_{1}a_{2}...a_{i}}p_{a_{1}}p_{a_{2}}...p_{a_{i}}  = \frac{\partial f}{\partial q^{a}}\frac{\partial g}{\partial p_{a}} - \frac{\partial f}{\partial p_{a}}\frac{\partial g}{\partial q^{a}} = -\frac{\partial f}{\partial p_{a}}\frac{\partial g}{\partial q^{a}}(q),
\end{equation}
and the claim follows integrating both sides.
\end{description}
The converse can be checked by direct calculation.
\end{proof}

A generalization which will serve the same purpose in quantization as the real polarization does but which is sometimes easier to implement is that of a complex polarization.

\begin{cxpolarization}
A \emph{complex polarization} of a symplectic manifold $(M,\omega)$ is a smooth distribution $m\mapsto P_{m}\subset (T_{m}M)_{\mathbb{C}}$ such that it is
\begin{description}
\item[(i)] Integrable: $X, Y\in V_{P}(M) \Rightarrow [X,Y]\in V_{P}(M)$,
\item[(ii)] Lagrangian: $\forall m\in M, \,\,\, P_{m} \text{ is a Lagrangian subspace of } (T_{m}M)_{\mathbb{C}}$,
\item[(iii)] The distribution $D = P\cap \bar{P}\cap TM  \text{ is of constant dimension}\footnote{The bar denotes complex conjugation.} $.
\end{description}
\end{cxpolarization}

Here, the elements of $(T_{m}M)_{\mathbb{C}}$ are sums of the form $X+iY$, where $X,Y\in T_{m}M$, with the obvious rules for addition and multiplication by a complex scalar. The importance of the distribution $D$, the real directions in $P$, should become clear later. Complex polarizations appear naturally in complex geometry, where they arise from the complex structure:

\begin{complexlagrangian}\label{cxlagr}
Let $J:V\to V$ be a complex structure on a symplectic vector space $(V,\omega)$ which is compatible with $\omega$, meaning that it is a linear canonical transformation on $V$. Then $J$ determines a Lagrangian subspace $P_{J}\subset V_{\mathbb{C}}$ such that $P_{J}\cap\bar{P}_{J} = \{0\}$. Conversely, a Lagrangian subspace $P\subset V_{\mathbb{C}}$ satisfying $P\cap\bar{P} = \{0\}$ determines a symplectic structure $J$ on $V$ compatible with $\omega$ such that $P = P_{J}$.
\end{complexlagrangian}

\begin{proof}
Suppose $\exists J$ compatible with $\omega$. Use $J$ to define the inclusion $V_{(J)}\hookrightarrow V_{\mathbb{C}}: X\mapsto \frac{1}{2}(X-iJX)$. Then the image of the map $P_{J} = \{X-iJX|X\in V\}\subset V_{\mathbb{C}}$ is isotropic:
\begin{equation}\nonumber
\omega (X-iJX,Y-iJY) = \omega(X,Y) - \omega(JX,JY) - i[\omega(JX,Y) + \omega(X,JY)] = 0,
\end{equation}
since $\omega (X,JY) = \omega(JX,J^{2}Y) = -\omega(JX,Y)$. Since the dimension of $P_{J}$ is the same as that of $V$, which is half that of $V_{\mathbb{C}}$, it is Lagrangian.

Conversely, let $P\subset V_{\mathbb{C}}$ be the referred Lagrangian subspace, with $P\cap\bar{P} = \{0\}$. Since we also have $\dim(P) = \dim(\bar{P}) = \frac{1}{2}\dim(V_{\mathbb{C}})$, it follows that $V_{\mathbb{C}} = P\oplus\bar{P}$. Hence for any $X\in V\subset V_{\mathbb{C}}$, there is a unique $Z\in P$ such that $X = Z+\bar{Z}$. Use this decomposition to define 
\begin{equation}\nonumber
J:X=Z+\bar{Z}\mapsto iZ-i\bar{Z}.
\end{equation}
This obviously satisfies $J^{2} = -\text{id}$, so that it indeed gives a complex structure on $V$. Finally, note that
\begin{equation}\nonumber
P = \left\{\frac{1}{2}(Z+\bar{Z}) - i\frac{1}{2}(iZ-i\bar{Z})|Z\in P\right\} = \{X - iJX|X\in V\} = P_{J}\subset V_{\mathbb{C}}
\end{equation}
as claimed.
\end{proof}

A Kahler manifold $M$ realizes this at each point. Indeed, one way of defining a Kahler manifold is as a $2n$-dimensional manifold with a symplectic structure $\omega$ and a complex structure $J$ which is compatible with $\omega$ at each point. In this case, the two-form $g(X,Y) = 2\omega(X,JY), \forall X, Y\in TM$ is symmetric and nondegenerate, defining a semi-Riemanian structure on $M$. The complex structure allows the introduction of $n$ holomorphic (with respect to $J$) coordinates $z^{a}$, in which 
\begin{equation}\label{holomomega}
\omega = i\omega_{ab}dz^{a}\wedge d\bar{z}^{b}.
\end{equation}
Since $\omega$ is real, $\omega_{ba} = \bar{\omega}_{ab}$.
Consider the distribution $P$ spanned by the vector fields $\partial/\partial z^{a}$. It is obviously integrable and Lagrangian as follows from the expression (\ref{holomomega}). Also, $D = P\cap\bar{P}\cap TM = \{0\}$ (no `real directions'). So $P$ is a complex polarization, called the holomorphic polarization of $M$ and realizes proposition \ref{cxlagr} in each tangent space. For this reason, one refers to a Lagrangian subspace $P\subset V_{\mathbb{C}}$ of that type (i.e. one such that $P\cap\bar{P} = \{0\}$) as one of a Kahler type. Note also that the conjugate distribution $\bar{P}$, spanned by the vectors $\partial/\partial\bar{z}^{a}$, is also a complex polarization, which is called the antiholomorphic polarization.

An interesting feature of Kahler manifolds is that in the neighbourhood of any point of $M$ one can find a real function $K$ such that
\begin{equation}\label{kahlerscalar}
\omega = i\partial\bar{\partial}K = i\frac{\partial K}{\partial z^{a}\partial\bar{z}^{b}}dz^{a}\wedge d\bar{z}^{b},
\end{equation}
where $\partial := dz^{a}\wedge \partial/\partial z^{a}$ and analogously for $\bar{z}^{a}$. Note that this guarantees that there is a symplectic potential $\theta  = -i\partial K$ with the special property that $\bar{X}\lrcorner \theta = 0,\,\, \forall X\in P$, where $P$ is the holomorphic polarization. We say that such a potential is \emph{adapted} to $P$, as an extension of the real case. Likewise, the potential $\bar{\theta} = i\bar{\partial}K$ is adapted to the antiholomorphic polarization $\bar{P}$. In the more general non-Kahler case, one has the definitions

\begin{strongint}
Let $P$ be a complex polarization of a symplectic manifold $(M, \omega)$ and let $D = P\cap\bar{P}\cap TM$. $P$ is said to be \emph{strongly integrable} if $D^{\perp} = (P + \bar{P})\cap TM$ is integrable, and is \emph{admissible} if there is an adapted complex symplectic potencial in the neighbourhood of each point, that is, a potential such that $\bar{X}\lrcorner\theta = 0,\,\, \forall X\in V_{P}(M)$. The polarized functions $C_{P}^{\infty}(M)$ are defined as the complex smooth functions on $M$ such that $\bar{X}(f) = 0,\,\, \forall X\in V_{P}(M)$.
\end{strongint}

One can show that every strongly integrable polarization is admissible. Note that this is not so straightforward, as not every complex polarization is K\"ahler. In the other extreme, the complexification of a real polarization provides a complex polarization with $P = \bar{P}$. To understand the cases in between, let us look first at more general complex Lagrangian subspaces.

\begin{gencxlagr}\label{gclag}
Let $(V,\omega)$ be a real symplectic vector space and $P\subset V_{\mathbb{C}}$ a Lagrangian subspace. Then $P$ determines a unique complex structure $J'$ on $V' = [(P+\bar{P})\cap V]/(P\cap\bar{P}\cap V)$ which is compatible with $\omega$. Let $V'_{(J')}$ denote the resulting complex vector space. Then
\begin{equation}\nonumber
\langle\cdot , \cdot\rangle _{J'}:V'_{(J')}\times V'_{(J')}\to\mathbb{C}:(X,Y)\mapsto \langle X,Y\rangle _{J'} = 2\omega(X,J'Y) + 2i\omega(X,Y)
\end{equation}
defines a Hermitian inner product on $V'_{(J')}$.
\end{gencxlagr}
\begin{proof}
First recall that, given a coisotropic subspace $F\subset V$ of a symplectic vector space $(V,\omega)$, $F/F^{\perp}$ is a symplectic vector space with the symplectic structure given by the projection of $\omega$ along $F^{\perp}$. This is the linear analogue of the reduction procedure (see the discussion following definition \ref{pres}). Now, $D = P\cap\bar{P}\cap V$ is a real and isotropic suspace of $V$, so that
\begin{equation}\nonumber
D^{\perp}/D = [(P+\bar{P})\cap V]/(P\cap\bar{P}\cap V)
\end{equation}
is a real symplectic vector space. Consider its complexification $V'_{\mathbb{C}} = (P+\bar{P})/(P\cap\bar{P})$ and let $\pi$ be the projection $P+\bar{P}\to V'_{\mathbb{C}}$. Then $P' = \pi(P)$ is a Lagrangian subspace of $V'_{\mathbb{C}}$ and clearly $P'\cap\bar{P}' = \pi(P\cap\bar{P}) = \{0\}$, so that, by proposition \ref{cxlagr}, it determines a complex structure $J'$ on $V'$.

Nondegeneracy and linearity on the second entry of $\langle\cdot , \cdot\rangle_{J'}$ follow from the properties of $J'$ and $\omega$. Finally, it is conjugate linear in the first entry since
\begin{equation}\nonumber
\langle Y,X\rangle_{J'} = 2\omega(J'X, -Y) - 2i\omega(X,Y) = 2\omega(X,J'Y) - 2i\omega(X,Y) = \langle X,Y\rangle_{J'}^{*}.
\end{equation}
\end{proof}

We call the signature $(r,s)$ of $\langle\cdot  , \cdot\rangle_{J'}$ (being $r$ the number of positive eingenvalues and $s$ the number of negative eigenvalues) the \emph{type} of the Lagrangian subspace $P\subset V_{\mathbb{C}}$. In the case of Kahler subspaces, $V' = V$, so that $r+s= n$, where $2n$ is the real dimension of $V$. When $r = n$, we say that $P$ is positive, and when $s = 0$, $P$ is nonnegative. The case of real Lagrangian subspaces, ie., when $P\subset V_{\mathbb{C}}$ is the complexification of a real Lagrangian subspace of $V$, is the one when $r = s= 0$.

A first application of this procedure to the non-linear (manifold) case is the construction of a polarization through symplectic reduction, which we describe in an example below.

\begin{compatiblecoisotropic}
Let $P$ be a polarization of $(M,\omega)$, $C$ a coisotropic submanifold of $M$ and $(M',\omega ')$ the reduction of $(C,\omega|_{C})$. If $(P\cap T_{\mathbb{C}}C)+K_{\mathbb{C}}$ is integrable and $\dim (P'\cap\bar{P}')$ is constant, then $C$ is said to be compatible with $P$. Here, $K$ is the characteristic distribution of $C$ and $P'_{m'}\subset (T_{m'}M')_{\mathbb{C}}$ is the projection into $(T_{m'}M')_{\mathbb{C}} = (T_{m}C)_{\mathbb{C}}/K_{\mathbb{C}}$ of $P_{m}\cap (T_{m}C)_{\mathbb{C}}$, with $m$ in the preimage of $m'\in M'$ through reduction.
\end{compatiblecoisotropic}

This is satisfied, for example, if the submanifold $C$ is given by $f_{i} = 0, \,\, i\in \{1, ..., k\}$, where the $f_{i}$ are real functions inducing Hamiltonian flows which preserve $P$. The reason for this definition is that

\begin{reducpolarize}\label{redpol}
If $C$ is a coisotropic submanifold of $(M,\omega)$ compatible with a polarization $P$, then there is a well-defined polarization $P'$ on the reduction of $(C,\omega|_{C})$, called the reduction of $P$.
\end{reducpolarize}
This can be seen as follows: on each point, $T_{m}C\subset T_{m}M$ is coisotropic, so $T_{m'}M' = T_{m}C/(T_{m}C)^{\perp} = T_{m}C/K_{m}$ is symplectic, where $m' = \pi(m)$ is the image of $m$ through reduction. Then, just as in proposition \ref{gclag}, $P_{m}\subset (T_{m}C)_{\mathbb{C}}$ projects to a Lagrangian subspace of $(T_{m'}M')_{\mathbb{C}}$. Now the hypothesis of compatibility of $C$ guarantees that this Lagrangian subspace is independent on the choice of $m'\in \pi^{-1}(m)$. This is the reason why we wrote $S^{2}$ as the reduction of $S^{3}\subset \mathbb{C}^{2}$ in the first place: we will see that $S^{3}$ is a coisotropic submanifold of $\mathbb{C}^{2}$ compatible with the holomorphic polarization of $\mathbb{C}^{2}$, so that this proposition guarantees the existence of an induced polarization on $S^{2}$.

Finally, we look into more general complex polarizations, which realize the general complex Lagrangian subspaces of propostion \ref{gclag} in each tangent space.

\begin{stronginteg}
Let $P$ be a strongly integrable polarization of a symplectic manifold $(M, \omega)$. Then, in a neighbourhood of each point one can define a coordinate system $(p_{a},q^{b},z^{\alpha})$, with $p_{a}$, $q^{b}$ real and $z^{\alpha}$ complex, in which $P$ is spanned by the vectors $\partial/\partial p_{a}$ and $\partial/\partial z^{\alpha}$ and
\begin{equation}\label{potstr}
\omega = d\left(p_{a}dq^{a}-i\frac{\partial K}{\partial z^{\alpha}}dz^{\alpha} - \frac{i}{2}\frac{\partial K}{\partial q^{a}}dq^{a}\right)
\end{equation}
for some $K(q,z,\bar{z})$.
\end{stronginteg}
\begin{proof}
We present a proof by constructing such coordinates. Let $n = \dim D = \dim P\cap\bar{P}\cap TM$ be the number of real directions in $P$ and take $n$ independent real functions $q^{1}, ..., q^{n}$ which are constant along $E = D^{\perp} = (P+\bar{P})\cap TM$. Since $P$ is integrable, one can find $n' = \frac{1}{2}\dim M - n$ complex functions $z^{1},...,z^{n}$ such that $\bar{P}$ is spanned by the Hamiltonian vector fields $X_{q^{a}}$ and $X_{z^{\alpha}}$.

Restric the analysis to a neighbourhood sufficiently small so that both $D$ and $E$ can be taken to be integrable in it. Note that the functions $z^{\alpha}$ are constant along $D$. Furthermore, each leaf $\Lambda_{q}$ of E is coisotropic (since $D = E^{\perp}$ is isotropic) and compatible with $P$: the characteristic foliation $K$ of $E$ is $E^{\perp} = D$, so that $[P\cap (T\Lambda_{q})_{\mathbb{C}}]+K_{\mathbb{C}} = P$, which is integrable. Hence, by propostion \ref{redpol}, $P$ projects to a polarization $P'$ on the reduction of the leaf $\Lambda_{q}$ of $E$. But the reduction quotients out the real directions in $D$, so $P'$ determines a Kahler structure on the reduction of $\Lambda_{q}$. This means that there exists a function $K_{q}(z,\bar{z})$ such that
\begin{equation}\label{labdaq}
\omega|_{\Lambda_{q}} = i\frac{\partial^{2}K}{\partial z^{\alpha}\partial\bar{z}^{\beta}}dz^{\alpha}\wedge d\bar{z}^{\beta}.
\end{equation}
We then define $K(q,z,\bar{z}) = K_{q}(z,\bar{z})$. Note that the last equation only defines $K(q,z,\bar{z})$ up to the addition of $f(q,z)+\bar{f}(q,\bar{z})$ for $f$ holomorphic in $z$ ($\ast $).

Pick a section $C$ of $D$ and define $n$ real functions $p_{b}$ by: $p_{b}|_{C} = 0$, $X_{q^{a}}(p_{b}) = \delta^{a}_{b}$ ($\ast \ast)$. Since the fields $X_{q^{a}}$ span $D$, this determines the functions $p_{b}$ in a neighbourhood of $C$. Thus the functions $q, p, \text{Im} (z), \text{Re} (z)$ form a local coordinate system in $M$, in which $D$ is spanned by the vector fields $\partial/\partial p_{b}$ and $P$ is spanned by both the vector fields $\partial/\partial p_{b}$ and $\partial/\partial z^{\alpha}$.

Now, on representing the real two-form $\omega$ in these coordinates, one should use the information known about it to narrow down the possible terms in the coordinate basis of $\Omega^{2}(M)$. First, $\{q^{a},p_{b}\} = X_{q^{a}}(p_{b}) = \delta^{a}_{b}$ shows that $p$ and $q$ are conjugate variables. Second, $\omega|_{P} = 0$ and $P$ is spanned by $\{\partial/\partial p_{b}, \partial/\partial z^{\alpha}\}$. Third, $\omega|_{\Lambda_{q}}$ is given by equation (\ref{labdaq}). Therefore, the most general expression for $\omega$ is
\begin{equation}\nonumber
\omega = dp_{q}\wedge dq^{a} + \zeta_{ab}dq^{a}\wedge dq^{b} + \xi_{a\alpha}dq^{a}\wedge dz^{\alpha} + \bar{\xi}_{a\alpha}dq^{a}\wedge d\bar{z}^{\alpha} + i\frac{\partial^{2}K}{\partial z^{\alpha}\partial\bar{z}^{\beta}}dz^{\alpha}\wedge d\bar{z}^{\beta}.
\end{equation}
The linearly independent terms that appear upon imposing $d\omega = 0$ give
\begin{equation}\nonumber
\frac{\partial \zeta_{ab}}{\partial p_{c}}dp_{c}\wedge dq^{a}\wedge dq^{b} = 0, \,\,\,\, \frac{\partial \xi_{a\alpha}}{\partial p_{c}}dp_{c}\wedge dq^{a}\wedge dz^{\alpha} = 0,
\end{equation}
from which follows that $\zeta_{ab}$ and $\xi_{ab}$ don't depend on $p$, and also
\begin{equation}\nonumber
\frac{\partial \zeta_{ab}}{\partial q^{c}}dq^{c}\wedge dq^{a}\wedge dq^{b} = 0,
\end{equation}
which means that the two-form $\zeta_{ab}|_{(z, \bar{z} \text{ fixed) }}dq^{a}\wedge dq^{b}$ is closed. Because the calculations are all local, we can take this to be exact, so that there are functions $\eta_{a}(q,z,\bar{z})$ such that
\begin{equation}\nonumber
\zeta_{ab} = \frac{1}{2}\left(\frac{\partial \eta_{b}}{\partial q^{a}} - \frac{\partial\eta_{a}}{\partial q^{b}}\right).
\end{equation}
Now let us change the section $C$ used in ($\ast\ast $) so that we trade $p_{a}\mapsto p_{a}-\eta_{a}$, which implies
\begin{equation}\nonumber
dp_{a}\wedge dq^{q}+\zeta_{ab}dq^{a}\wedge dq^{b}\mapsto d(p_{a} - \eta_{a})\wedge dq^{a} + \frac{1}{2}\left(\frac{\partial \eta_{b}}{\partial q^{a}} - \frac{\partial\eta_{a}}{\partial q_{b}}\right)dq^{a}\wedge dq^{b} = dp_{a}\wedge dq^{b} + \text{ z-dependent terms},
\end{equation}
so we may actually take $\zeta_{ab} = 0$, meaning that we absorb the corresponding term in the definition of the $p_{b}$. Rewriting $\omega$ and again imposing $d\omega = 0$ gives
\begin{equation}\nonumber
\frac{\partial\xi_{a\alpha}}{\partial z^{\beta}}dz^{\beta}\wedge dz^{\alpha}\wedge dq^{a} = 0, \,\,\,\, \frac{\partial\xi_{a\alpha}}{\partial q^{b}}dq^{b}\wedge dq^{a}\wedge dz^{\alpha} = 0,
\end{equation}
which we again solve trivially, by using a potential $g(q,z,\bar{z})$ such that
\begin{equation}\nonumber
2\xi_{a\alpha} = -\frac{\partial^{2}g}{\partial q^{a}\partial z^{\alpha}},
\end{equation}
which then defines $g$ only up to the addition of an arbitrary $h(z, \bar{z})$ ($\ast \ast\ast$). Finally, an additional term in the equation $d\omega = 0$ gives
\begin{equation}\nonumber
\begin{split}
\left(\frac{\partial\xi_{a\alpha}}{\partial z^{\beta}} - \frac{\partial\bar{\xi}_{a\beta}}{\partial z^{\alpha}} + i\frac{\partial^{3}K}{\partial q^{a}\partial z^{\alpha}\partial\bar{z}^{\beta}}\right) & dq^{a}\wedge dz^{\alpha}\wedge d\bar{z}^{\beta} = 0 \Rightarrow \frac{\partial^{3}(g-\bar{g})}{\partial q^{a}\partial z^{\alpha}\partial\bar{z}^{\beta}} + 2i\frac{\partial^{3}K}{\partial q^{a}\partial z^{\alpha}\partial\bar{z}^{\beta}} = 0 \\
& \Rightarrow \frac{\partial^{3}}{\partial q^{a}\partial z^{\alpha}\partial\bar{z}^{\beta}}\left[ K - \frac{(g-\bar{g})}{2i}\right] = 0,
\end{split}
\end{equation}
so $K - \text{Im}(g)$ is at most something of the form $-f(q,z) - \bar{f}(q,\bar{z}) + h(z,\bar{z})$, with $f$ holomorphic in $z$. One can thus use the freedom in ($\ast$) to redefine $K\mapsto K+f+\bar{f}$ and ($\ast\ast\ast$) to redefine $g\mapsto g+h$. In the end of the day, we may assume that $K = \text{Im}(g)$.

Finally, replace $C$ once more in ($\ast\ast$) in order to take
\begin{equation}\nonumber
p_{a}\mapsto p_{a} - \frac{\partial}{\partial q^{a}}\left(\frac{g+\bar{g}}{4}\right),
\end{equation}
so that, in these coordinates,
\begin{equation}\nonumber
\begin{split}
\omega = d\left[ p_{a} - \frac{\partial}{\partial q^{a}}\left(\frac{g+\bar{g}}{4}\right)\right]\wedge dq^{a} - \frac{1}{2}\frac{\partial^{2}}{\partial q^{a}\partial z^{\alpha}}\left[\left(\frac{g+\bar{g}}{2}\right)+iK\right]dq^{a}\wedge dz^{\alpha} + \\
- \frac{1}{2}\frac{\partial^{2}}{\partial q^{a}\partial\bar{z}^{\beta}}\left[\left(\frac{g+\bar{g}}{2}\right)-iK\right]dq^{a}\wedge d\bar{z}^{\beta}+i\frac{\partial^{2}K}{\partial z^{\alpha}\partial\bar{z}^{\beta}}dz^{\alpha}\wedge d\bar{z}^{\beta} \,\,\,\,\Rightarrow \\
\omega = dp_{a}\wedge dq^{a} - \frac{i}{2}\frac{\partial^{2}K}{\partial q^{a}\partial z^{\alpha}} dq^{a}\wedge dz^{\alpha} + \frac{i}{2}\frac{\partial^{2}K}{\partial q^{a}\partial\bar{z}^{\beta}}dq^{a}\wedge d\bar{z}^{\beta} + i\frac{\partial^{2}K}{\partial z^{\alpha}\partial\bar{z}^{\beta}}dz^{\alpha}\wedge d\bar{z}^{\beta}.
\end{split}
\end{equation}
\end{proof}

Note that the expression between brackets in equation (\ref{potstr}) is a symplectic potential adapted to $P$. Hence every strongly integrable polarization is admissible, as mentioned before.

In quantization, the situation one typically has is: $P$ is a strongly integrable polarization on a symplectic manifold $(M,\omega)$ and $B\to M$ is a prequantum bundle over $M$.
\begin{polsection}
A smooth section $s: M\to B$ is said to be polarized along $P$ if $\nabla_{\bar{X}}s = 0, \,\,\, \forall X\in V_{P}(M)$. We denote these sections by $C^{\infty}_{P}(B)$.
\end{polsection}

Since the curvature of $\nabla$ is proportional to $\omega$, its restriction to $P$ is flat, so it is always possible to find local polarized sections. One can then use the potential $\theta$, defined in a simply connected open neighbourhood $U\subset M$ and adapted to $P$, to define a section $s: U\to B$ by picking initial $m\in U$ and $b_{0}\in B_{m}$ and taking
\begin{equation}\label{sfromtheta}
s(\gamma_{t}m) = b\exp\left(-\frac{i}{\hbar}\int_{m}^{\gamma_{t}m}\theta\right),
\end{equation}
for any piecewise smooth path $\gamma_{t}$, where $b\in B_{\gamma_{t}m}$ is obtained from $b_{0}$ by parallel transport along $\gamma$. Then $Ds = -i\frac{\theta}{\hbar}\otimes s$, so that $\theta = \theta_{s}$. In using this frame to write any other section as $s' = \phi s$, the fact that $\theta_{s}$ is adapted to $P$ implies
\begin{equation}\nonumber
s' = \phi s\in C^{\infty}_{P}(B)\Leftrightarrow 0  = \nabla_{\bar{X}}s' = \left[\bar{X}(\phi) - \frac{i}{\hbar}(\bar{X}\lrcorner\theta_{s})\right]s = \bar{X}(\phi)s , \forall X\in V_{P}(M) \Leftrightarrow \phi \in C^{\infty}_{P}(M),
\end{equation}
so that polarized sections are represented by polarized functions. In the case of a positive Kahler polarization, we have the adapted potential $\theta = -i\partial K$, where $K$ is defined through equation (\ref{kahlerscalar}). In the constructed frame, the polarized sections are given by holomorphic functions of $z$, the coordinates holomorphic with respect to the complex structure determined by $P$. We see that one can use this section to determine the trivialization of $B$ in the neighbourhood of any point, so that the transition functions of $B$ are holomorphic. Hence a Kahler polarization $P$ ends up giving $B$ the structure of a holomorphic line bundle. One can show then that the space $\mathcal{H}_{P}\subset \mathcal{H}$ of square-integrable polarized sections of $B$ is a Hilbert subspace of the prequantum Hilbert space $\mathcal{H}$.

In the more general case of a strongly integrable polarization, equation (\ref{potstr}) says that 
\begin{equation}\label{potentialK}
\theta = p_{a}dq^{a}-i\frac{\partial K}{\partial z^{\alpha}}dz^{\alpha} - \frac{i}{2}\frac{\partial K}{\partial q^{a}}dq^{a}
\end{equation}
is an adapted symplectic potential. In the corresponding frame, polarized sections are represented by elements of $C^{\infty}_{P}(M)$, which are functions $\phi(q,z)$, holomorphic in $z$. The definition of $\mathcal{H}_{P}$, though, is generally not as straightforward as in the positive case, as the existence of square-integrable polarized sections might not even be guaranteed.

Note that, since $\nabla$ is compatible with the Hermitian structure on the fibres,
\begin{equation}\nonumber
d(s,s) = (Ds, s) + (s, Ds) = \frac{i}{\hbar}(\theta - \bar{\theta})(s,s) \Rightarrow \frac{\theta - \bar{\theta}}{2i} = d\left[\frac{\hbar}{2}\ln (s,s)\right].
\end{equation}
But, from equation (\ref{potentialK}), $\text{Im} (\theta) = -\frac{1}{2}dK$, so that, for a section $s' = \phi s$,
\begin{equation}\label{ints}
(s',s') = \bar{\phi}\phi(s,s) = \bar{\phi}\phi \exp\left[\frac{2}{\hbar}\int\left(\frac{\theta - \bar{\theta}}{2i}\right)\right] = \bar{\phi}\phi e^{-K/\hbar},
\end{equation}
by adding the integration constant to $K$. Expression (\ref{ints}) gives the Hermitian structure in terms of the holomorphic functions representing the sections in the local frame.

\begin{quantumsphere}\emph{Quantization of the Sphere}

In the case we have been discussing of elementary systems with rotational symmetry, which are spheres of radii $s = N\frac{\hbar}{2}$, the symplectic manifold is the reduction of 
\begin{equation}\nonumber
(S^{3},\omega_{f} = d\theta_{f}), \,\,\,\,\,\,\,\,\, \theta_{f} = is(z^{0}d\bar{z}^{0}+z^{1}d\bar{z}^{1}-\bar{z}^{0}dz^{0}-\bar{z}^{1}dz^{1}),
\end{equation}
where $z^{1},z^{2}$ are holomorphic coordinates on $\mathbb{C}^{2}$. As we verified before, for these radii proposition \ref{reducequantum} guarantees the existence of a prequantum bundle, also providing a way to construct it. Following the same procedure used in the proof, each smooth section of the prequantum bundle $B\to M$ is a smooth function $\psi :S^{3}\to\mathbb{C}$ such that
\begin{equation}\nonumber
\nabla_{X}(\psi e) = 0, \,\, \forall X\in K \,\,\,\,\Rightarrow \,\,\,\, X\lrcorner d\psi = \frac{i}{\hbar}(X\lrcorner\theta_{f})\psi,\,\, \forall X\in K,
\end{equation}
where $K$ is the characteristic foliation of $(S^{3}, d\theta_{f})$. Integrating this relation along one of the circles which make up the leaves of $K$ gives
\begin{equation}\label{hom}
\psi(e^{it}z^{\alpha},e^{-it}\bar{z}^{\beta}) = \exp\left( \frac{i}{\hbar}\int_{0}^{t}\theta_{f}\right)\psi(z^{\alpha},\bar{z}^{\beta}) = \exp\left(\frac{i}{\hbar}2st\right)\psi(z^{\alpha},\bar{z}^{\beta}) = (e^{it})^{N}\psi(z^{\alpha},\bar{z}^{\beta}),
\end{equation}
by using the explicit form of $\theta_{f}$.

Now, $S^{3}\subset\mathbb{C}^{2}$ is coisotropic, as can be seen from the fact that $(TS^{3})^{\perp}$ is one-dimensional and hence isotropic. Moreover, it is given by $0 = f := z^{0}\bar{z}^{0} + z^{1}\bar{z}^{1} - 1$. The Hamiltonian flow generated by $f$ is found from
\begin{equation}\nonumber
\sum_{j=0}^{1}(z^{j}d\bar{z}^{j} + \bar{z}^{j}dz^{j}) = df = -X_{f}\lrcorner\omega = -\left(X_{f_{j}}\frac{\partial}{\partial z^{j}} + X_{\bar{f_{j}}}\frac{\partial}{\partial \bar{z}^{j}}\right),
\end{equation}
so that $f$ generates the flow
\begin{equation}\nonumber
\frac{dz^{j}(\rho_{t}(z,\bar{z}))}{dt} = X_{f_{j}}(\rho_{t}(z,\bar{z})) = \frac{i}{\hbar}z^{j}(\rho_{t}(z,\bar{z})) \,\,\,\, \Rightarrow \,\, \,\, \rho_{t}z^{j} = e^{it/\hbar}z^{j}.
\end{equation}
Now take an arbitrary $g\in S_{P}^{0}(\mathbb{C}^{2})$, where $P$ is the holomorphic polarization of $\mathbb{C}^{2}$. Then
\begin{equation}\nonumber
\frac{\partial g}{\partial \bar{z}^{j}} = 0, \,\, j\in \{0,1\} \,\,\,\, \Rightarrow \,\,\,\,\frac{\partial}{\partial \bar{z}^{j}}\{f,g\} = \frac{\partial}{\partial \bar{z}^{j}}\left[\frac{d}{dt}g(e^{it/\hbar}z)\Big|_{t=0}\right] = 0,
\end{equation}
so $\{f,g\}\in S_{P}^{0}(M)$, which implies that the flow of $f$ preserves the polarization. Therefore $S^{3}$ is compatible with the holomorphic polarization of $\mathbb{C}^{2}$ so that, by proposition \ref{redpol}, this defines a polarization in the reduction of $S^{3}$. We use this polarization to restric the space of states, so that the only admissible sections are represented by functions on $S^{3}$ which are holomorphic in $(z^{0},z^{1})$. We conclude that the states are functions $\psi(z^{0},z^{1})$, holomorphic in both the coordinates, such that $\psi(e^{it}z^{\alpha}) = (e^{it})^{N}\psi(z^{\alpha})$ for any real $t$. Hence $\psi $ is a homogeneous function of the coordinates $(z^{0},z^{1})$ of degree $N$, which therefore is of the form
\begin{equation}\nonumber
\psi = \psi_{A_{1}A_{2}...A_{N}}z^{A_{1}}z^{A_{2}}...z^{A_{N}},
\end{equation}
where the indices $A_{i}$ are summed over $\{0,1\}$. Therefore a quantum state is given by a symmetric $N$-index spinor $\psi_{A_{1}A_{2}...A_{N}}$, where $s = \frac{\hbar}{2}N$. Lastly, we show that the hermitian structure on the space of states coincides with the inner product of spinors.

Note that, since reduction is, in this case, quotient by the circles $e^{it}(z^{0},z^{1}), \,\, t\in \mathbb{R}$, the holomorphic coordinate $z = z^{0}/z^{1}$ is well defined on the reduction (but for the point at infinity). We claim that $\omega_{f}$ projects to
\begin{equation}\label{omegasphere}
\omega = \frac{i\hbar Ndz\wedge d\bar{z}}{(1+z\bar{z})^{2}}.
\end{equation}
Indeed, if we treat $z^{0},z^{1}$ in $z=z^{0}/z^{1}$ as independent but for the relation $z^{0}\bar{z}^{0}+z^{1}\bar{z}^{1} = 1$ (which implies $z^{0}d\bar{z}^{0}+\bar{z}^{0}dz^{0}+\bar{z}^{1}dz^{1}+z^{1}d\bar{z}^{1} = 0$),
\begin{equation}\nonumber
\begin{split}
\frac{i\hbar Ndz\wedge d\bar{z}}{(1+z\bar{z})^{2}} &= \frac{2is(z^{1}\bar{z}^{1})^{2}}{(z^{1}\bar{z}^{1}+z^{0}\bar{z}^{1})^{2}}\left(\frac{dz^{0}}{z^{1}}-\frac{z^{0}}{(z^{1})^{2}}dz^{1}\right)\wedge\left(\frac{d\bar{z}^{0}}{\bar{z}^{1}}-\frac{\bar{z}^{0}}{(\bar{z}^{1})^{2}}d\bar{z}^{1}\right)\\
& = 2is[z^{1}\bar{z}^{1}dz^{0}\wedge d\bar{z}^{0} - \bar{z}^{0}dz^{0}\wedge(-z^{0}d\bar{z}^{0}-\bar{z}^{0}dz^{0}-\bar{z}^{1}dz^{1}) + \\
& - \bar{z}^{1}dz^{1}\wedge(-\bar{z}^{0}dz^{0}-z^{1}d\bar{z}^{1}-\bar{z}^{1}dz^{1})+z^{0}\bar{z}^{0}dz^{1}\wedge d\bar{z}^{1}]|_{TS^{3}}\\
&= 2is(dz^{0}\wedge d\bar{z}^{0}+dz^{1}\wedge d\bar{z}^{1})|_{TS^{3}} = \omega_{f}|_{TS^{3}}.
\end{split}
\end{equation}
Note that equation (\ref{omegasphere}) also implies that $K = \hbar N \ln (1+z\bar{z})$ is a Kahler scalar, i.e., that $\omega = i\partial\bar{\partial}K$. 

Therefore, over the region where the $z$ coordinate is well-defined,
\begin{equation}\nonumber
\psi = \psi_{A_{1}A_{2}...A_{N}}z^{A_{1}}z^{A_{2}}...z^{A_{N}} = (z^{1})^{N}\sum_{k=0}^{N}\binom{N}{k}\psi_{k}z^{k},
\end{equation}
where $\psi_{k} = \psi_{A_{1}A_{2}...A_{N}}$ with $A_{i} = 1 \, \Leftrightarrow \, 1\leq i \leq k$ (remember that $\psi_{A_{1}A_{2}...A_{N}}$ is totally symmetric). Thus, in the trivialization $s$ specified by the potential $-i\partial K$, a state is of the form $s' = \psi s$, with
\begin{equation}\nonumber
\psi (z) = \sum_{k=0}^{N}\binom{N}{k}\psi_{k}z^{k},
\end{equation}
and, since $K = \hbar N\ln (1+z\bar{z})$ and $\omega/(2\psi\hbar)$ is a natural volume form,
\begin{equation}\label{calcspinor}
\begin{split}
\langle\psi , \psi\rangle &= \int\psi\bar{\psi}e^{-K/\hbar}\omega = \int\psi\bar{\psi}(1+z\bar{z})^{-N}\frac{i\hbar Ndz\wedge d\bar{z}}{(2\pi\hbar)(1+z\bar{z})^{2}} \\
&= \frac{iN}{2\pi}\int\sum_{k=0}^{N}\binom{N}{k}\psi_{k}z^{k}\sum_{l=0}^{N}\binom{N}{l}\bar{\psi}_{l}\bar{z}^{l}\frac{dz\wedge d\bar{z}}{(1+z\bar{z})^{N+2}} \\
&= -2N\sum_{k=0}^{N}\binom{N}{k}^{2}\bar{\psi}_{k}\psi_{k}\int_{0}^{\infty}\frac{r^{2k+1}dr}{(1+r^{2})^{N+2}}\\
&=N\sum_{k=0}^{N}\binom{N}{k}^{2}\bar{\psi}_{k}\psi_{k}\int_{0}^{1}t^{N-k}(1-t)^{k}dt, \,\,\,\,\,\, t = (1+r^{2})^{-1}\\
&=N\sum_{k=0}^{N}\binom{N}{k}^{2}\bar{\psi}_{k}\psi_{k}B(N-k+1,k+1) = N\sum_{k=0}^{N}\binom{N}{k}^{2}\bar{\psi}_{k}\psi_{k}\frac{(N-k)!k!}{(N+1)!} \\
&= \frac{N}{N+1}\sum_{k=0}^{N}\binom{N}{k}\bar{\psi}_{k}\psi_{k} = \frac{N}{N+1}\sum_{A_{1}A_{2}...A_{N}}\bar{\psi}_{A_{1}A_{2}...A_{N}}\psi_{A_{1}A_{2}...A_{N}},
\end{split}
\end{equation}
where $B(x,y)$ is the beta function. Hence the Hermitian structure is given by the inner product of spinors\footnote{On the relation between such spinor spaces and the representation theory of the rotation groups, see for example \cite{chevalley1996algebraic,lounesto2001clifford,penrose1984spinors}.}.
\end{quantumsphere}

\begin{wavequation}\label{wavesss}\emph{The relativistic wave equations}

We apply the same procedure to the phase spaces of the free particles. For a massive particle of spin $s=N\frac{\hbar}{2}$, we have seen that the phase space is the reduction of $(C_{sm},\omega ')$. We can accomodate the sign difference in the definition of $\omega '$ by writing it as $\omega ' = d\theta '$ in terms of the potential
\begin{equation}\label{potwave}
\theta ' = \frac{is\sqrt{2}}{m}\left(p_{A\bar{A}}z^{A}d\bar{z}^{\bar{A}} - p_{A\bar{A}}\bar{z}^{\bar{A}}dz^{A}\right) + \xi q^{a}dp_{a},
\end{equation}
where $\xi = p_{0}/|p_{0}|$ is the sign of $p_{0}$. For these values of spin, proposition \ref{reducequantum} implies that there is a prequantum bundle, and its sections can be constructed as smooth functions $\psi :C_{sm}\to\mathbb{C}$ which are covariantly constant along the leaves of the characteristic foliation in $C_{sm}$. Since reduction is given by the quotient by $(p_{a},q^{b},z^{C})\sim (p_{a}, q^{b}+\lambda p^{b},e^{i\phi}z^{C}), \,\, \forall \lambda, \phi \in \mathbb{R}$, the characteristic foliation is spanned by vector fields which generate the flow $(p_{a}, q^{b}+\lambda p^{b},e^{i\phi}z^{C})$. These are of the form
\begin{equation}\nonumber
p^{b}\frac{\partial}{\partial q^{b}}, \,\,\,\, \text{ and } \,\,\,\, iz^{A}\frac{\partial}{\partial z^{A}} - i\bar{z}^{\bar{A}}\frac{\partial}{\partial \bar{z}^{\bar{A}}}
\end{equation}
(compare the second with the characteristic foliation of $S^{3}$). Using the potential (\ref{potwave}), $\nabla_{X}(\psi e) = 0, \forall X\in K$ gives
\begin{equation}\label{secredpoinc}
\begin{split}
0 &= X\lrcorner d\psi - \frac{i}{\hbar}(X\lrcorner\theta ')\psi \,\,\,\,\Rightarrow \\
0=p^{b}\frac{\partial\psi}{\partial q^{b}},  \,\,\,\, 0= iz^{A}\frac{\partial\psi}{\partial z^{A}} &- i\bar{z}^{\bar{A}}\frac{\partial\psi}{\partial\bar{z}^{\bar{A}}} - \frac{2is\sqrt{2}p_{A\bar{A}}z^{A}\bar{z}^{\bar{A}}}{\hbar m}\psi = i\left[z^{A}\frac{\partial\psi}{\partial z^{A}} - \bar{z}^{\bar{A}}\frac{\partial\psi}{\partial\bar{z}^{\bar{A}}} - n\psi\right].
\end{split}
\end{equation}
Again there is a polarization of $T^{*}\mathbb{M}\times\mathbb{S}$ which is compatible with the coisotropic submanifold $C_{sm}$. On the component in which $\xi =1$ it is $\text{span}\{\partial/\partial q^{a}, \partial/\partial z^{A}\}$ and on the component in which $\xi = -1$ it is $\text{span}\{\partial/\partial q^{a}, \partial/\partial\bar{z}^{A}\}$. Hence it projects to a polarization on the reduction of $C_{sm}$. On the component $M_{sm}^{+}$, polarization further restricts the functions $\psi\in C^{\infty}_{\mathbb{C}}(C_{sm})$ by
\begin{equation}\label{polwave}
\frac{\partial\psi}{\partial q^{a}} = 0 = \frac{\partial\psi}{\partial\bar{z}^{\bar{A}}}.
\end{equation}
Therefore, from equations (\ref{secredpoinc}) and (\ref{polwave}), we see that each state in $\mathcal{H}_{P}$ is given by a $\psi\in C^{\infty}_{\mathbb{C}}(C_{sm})$ independent of $q$, holomorphic in $z$, and such that $z^{A}\partial_{z^{A}}\psi = n\psi$, where $s = n\frac{\hbar}{2}$. Thus its dependence on the $z^{A}$ coordinates should be that of a homogeneous polinomial of degree $n$, so that
\begin{equation}\nonumber
\psi(p,z) = \psi_{A_{1}A_{2}...A_{n}}(p)z^{A_{1}}z^{A_{2}}...z^{A_{n}}
\end{equation}
for some $n$-index spinor field $\psi_{A_{1}...A_{n}}$ on the future-pointing component $H_{m}^{+}$ of the mass shell $\{p_{a}p^{a}=m^{2}\}$. Note that the spinor indices transform correctly, since they are contrated with the spinors $z^{A}$. We denote this half of the Hilbert space $\mathcal{H}_{P}$ by $\mathcal{H}_{sm}^{+}$. Repeating the calculation of equation (\ref{calcspinor}), the inner product in $\mathcal{H}_{sm}^{+}$ becomes, in terms of the spinor components,
\begin{equation}\nonumber
\langle\psi e,\psi e\rangle = \int_{H_{m}^{+}}(\psi , \psi)d\tau, \,\,\,\, \text{ where }\,\,\,\, (\psi,\psi) = p^{A_{1}\bar{A}_{1}}p^{A_{2}\bar{A}_{2}}...p^{A_{n}\bar{A}_{n}}\bar{\psi}_{\bar{A}_{1}\bar{A}_{2}...\bar{A}_{n}}\psi_{A_{1}A_{2}...A_{n}},
\end{equation}
and $d\tau $ is the natural volume element in $H_{m}^{+}\subset\mathbb{M}$ invariant under the Poincar\'e group,
\begin{equation}\nonumber
d\tau = \frac{1}{\hbar^{2}}\frac{dp_{1}\wedge dp_{2}\wedge dp_{3}}{|p_{0}|}.
\end{equation}

Consider now the massive wave equation
\begin{equation}\label{massivewave}
(\Box + \mu^{2})\phi_{A_{1}A_{2}...A_{n}} = 0
\end{equation}
for $\mu = m/\hbar$. A plane wave of the form $e^{-ip_{a}x^{a}/\hbar}$, where $p_{a}p^{a} = m^{2}$, is obviously a solution. In fact, all `well-behaved' solutions are linear combinations of these, a fact summarized in the Fourier transform
\begin{equation}\nonumber
\phi_{A_{1}A_{2}...A_{n}}(x) = \left(\frac{1}{2\pi}\right)^{\frac{3}{2}}\int_{H_{m}}\psi_{A_{1}A_{2}...A_{n}}(p)e^{-ip_{a}x^{a}/\hbar}d\tau,
\end{equation}
for some spinor-valued function $\psi_{A_{1}A_{2}...A_{n}}(p)$ on $H_{m}=\{p_{a}p^{a}=m^{2}\}$. So we can use the Fourier transform to associate each element of $\mathcal{H}_{sm}^{+}$ to an element of the `space of solutions of (\ref{massivewave}) with well-defined Fourier transform'. Note that this sends $\mathcal{H}_{sm}^{+}$ to positive frequency solutions: the ones for which the Fourier transform vanishes on $H_{m}^{-}$. That this association is invariant under the Lorentz group is clear from the spinor indices and the form of the Fourier transform. Under a translation $x\mapsto x+y$,
\begin{equation}\nonumber
\begin{split}
\left(\frac{1}{2\pi}\right)^{\frac{3}{2}}\int_{H_{m}^{+}}\tilde{\psi}_{A_{1}A_{2}...A_{n}}(p)e^{-ip_{a}x^{a}/\hbar}d\tau &= \tilde{\phi}_{A_{1}...A_{n}}(x) = \phi_{A_{1}...A_{n}}(x-y) \\
&= \left(\frac{1}{2\pi}\right)^{\frac{3}{2}}\int_{H_{m}^{+}}\psi_{A_{1}...A_{n}}e^{-ip_{a}(x-y)^{a}/\hbar}d\tau,
\end{split}
\end{equation}
so we must have $\psi_{A_{1}...A_{n}}\mapsto e^{ip_{a}y^{a}/\hbar}\psi_{A_{1}...A_{n}}$. Indeed, by equation (\ref{potwave}), the symplectic potential transforms by $\theta '\mapsto \theta ' + d(y^{a}p_{a})$, which implies that the trivialization of the prequantum bundle determined by it (equation (\ref{sfromtheta})) transforms as
\begin{equation}\nonumber
s(\gamma_{t}m)\mapsto  b\exp\left(-\frac{i}{\hbar}\int_{m}^{\gamma_{t}m}\theta ' +d(y^{a}p_{a})\right) = e^{-ip_{a}y^{a}/\hbar}s(\gamma_{t}m),
\end{equation}
and, therefore, for some section $(\psi s)$ one must have $\psi_{A_{1}...A_{n}}\mapsto e^{ip_{a}y^{a}/\hbar}\psi_{A_{1}...A_{n}}$.

On the other component, $C_{sm}^{-}$, the polarization condition reads
\begin{equation}\nonumber
\frac{\partial\psi}{\partial q^{a}} = 0 = \frac{\partial\psi}{\partial\bar{z}^{\bar{A}}},
\end{equation}
so that, together with (\ref{secredpoinc}), it says that the states in $\mathcal{H}_{sm}^{-}$ are elements of $C^{\infty}_{\mathbb{C}}(C_{sm}^{-})$ which do not depend on $q$, are antiholomorphic in $z^{A}$, being furthermore homogeneous of degree $n$ in the $\bar{z}^{\bar{A}}$. We write
\begin{equation}\label{cxonj}
\psi(p,\bar{z}) = \bar{\psi}_{\bar{A}_{1}\bar{A}_{2}...\bar{A}_{n}}(p)\bar{z}^{\bar{A}_{1}}\bar{z}^{\bar{A}_{2}}...\bar{z}^{\bar{A}_{n}},
\end{equation}
and this time $\psi_{A_{1}A_{2}...A_{n}}$ is spinor-valued function on the other component $H_{m}^{-}$ of $H_{m}$. Again the inner product in $\mathcal{H}$ corresponds to the inner product of spinors and the Fourier transform gives a well-defined correspondence between $\mathcal{H}_{sm}^{-}$ and the negative frequency (the ones whose Fourier transform vanish on $H_{m}^{+}$) solutions of the wave equation (\ref{massivewave}).

Note that the complex structure 
\begin{equation}\nonumber
\mathcal{H}_{P}=\mathcal{H}_{sm}^{+}\oplus\mathcal{H}_{sm}^{-}\ni (\psi_{A_{1}...A_{n}}|_{H_{m}^{+}},\psi_{A_{1}...A_{n}}|_{H_{m}^{-}})\mapsto (i\psi_{A_{1}...A_{n}}|_{H_{m}^{+}},i\psi_{A_{1}...A_{n}}|_{H_{m}^{-}})
\end{equation}
is not mapped to the complex structure $\phi_{A_{1}...A_{n}}\mapsto i\phi_{A_{1}...A_{n}}$ on the space of solutions of (\ref{massivewave}). Rather, it is mapped antilinearly to the complex structure $J$ which multiplies the positive frequency part of $\phi_{A_{1}...A_{n}}$ by $-i$ and the negative frequency part by $i$, because of the complex conjugate in equation (\ref{cxonj}). Therefore, one can identify $\mathcal{H}_{P}$ with $\bar{V}_{(J)}$, the dual (as a complex vector space) of the space of solutions $V$ of the linear equation (\ref{massivewave}) with complex structure $J$. 
\end{wavequation}

The case of a massless particle of helicity $s=N\frac{\hbar}{2}$ is very similar. The phase space the reduction of $(C_{s0},d\theta')$, where
\begin{equation}\label{potwave2}
\theta ' = -i\omega^{A}d\bar{\pi}_{A}+i\bar{\omega}^{\bar{A}}d\pi_{\bar{A}}.
\end{equation}
Sections of the prequantum bundle are smooth functions $\psi :C_{sm}\to\mathbb{C}$ which are covariantly constant along the leaves of the characteristic foliation in $C_{s0}$. Vectors generating this foliation are of the form
\begin{equation}\nonumber
 i\omega^{A}\frac{\partial}{\partial \omega^{A}} - i\bar{\omega}^{\bar{A}}\frac{\partial}{\partial\bar{\omega}^{\bar{A}}}, \,\,\,\, \text{ and } \,\,\,\, i\pi^{\bar{A}}\frac{\partial}{\partial \pi^{\bar{A}}} - i\bar{\pi}^{A}\frac{\partial}{\partial \bar{\pi}^{A}}.
\end{equation}
The space comes with a polarization spanned by the projections of $\partial/\partial\omega^{A}$ and $\partial/\partial\bar{\omega}^{\bar{A}}$. Therefore, using the potential (\ref{potwave2}), and this polarization, we see that the elements of $\mathcal{H}_{P}$ are given by complex functions on $C_{s0}$ of the form $\psi(\pi_{\bar{A}},\bar{\pi}_{A})$ and such that
\begin{equation}\nonumber
\pi_{\bar{A}}\frac{\partial\psi}{\partial\pi_{\bar{A}}}-\bar{\pi}_{A}\frac{\partial\psi}{\partial\bar{\pi}_{A}} = -\frac{2n}{\hbar}\psi,
\end{equation}
remembering that $\omega^{A}\bar{\pi}_{A}+\bar{\omega}^{\bar{A}}\pi_{\bar{A}} = 2s$ on $C_{s0}$. For $n>0$, these are mapped to the positive frequency solutions of the massless wave equation, that is, equation (\ref{massivewave}) for $m^{2}=0$ and for $n<0$ they are mapped to negative frequency solutions. The correspondence is given by the Fourier transform
\begin{equation}\nonumber
\phi_{\bar{A}_{1}\bar{A}_{2}...\bar{A}_{n}}(x) = \left(\frac{1}{2\pi}\right)^{\frac{3}{2}}\int_{H_{0}}\psi(p)\pi_{\bar{A}_{1}}\pi_{\bar{A}_{2}}...\pi_{\bar{A}_{n}}e^{-ip_{a}x^{a}/\hbar}d\tau,
\end{equation}
where $H_{0}$ is the light-cone $\{p_{a}p^{a}=0\}$ and $p_{A\bar{A}} = \bar{\pi}_{A}\pi_{\bar{A}}$. Just like in the previous example, $\mathcal{H}_{P}=\mathcal{H}_{s0}\oplus\mathcal{H}_{-s0}$ is identified with $\bar{V}_{(J)}$, where $V$ is the space of well-behaved solutions of the massless wave equation and $J$ is the same complex structure.

\pagebreak

\section{Free Fields}\label{fields}
We can now examine the next step, passing from relativistic wave equations to quantum fields, in terms of geometric quantization. This provides us the motivation to study in detail the quantization of a vector space, which will subsequently lead to Fock space quantization.

\subsection{The Space of Solutions}

Up to this point, we have been thinking of symplectic manifolds as arising from the phase spaces of physical systems. There is a similar symplectic geometry of the Lagrangian formalism, which is more useful when speaking of quantum fields.

A (classical) field will mean here a smooth section of a vector bundle $F\to Q$ over spacetime $Q$ (assumed to have a semi-Riemannian structure) which vanishes sufficiently rapidly at infinity so that all the integrals we will write converge. We assume that the collection of all these fields forms a manifold\footnote{Clearly, this manifold will in general be infinite-dimensional, which introduces a number of complications in defining the various quantities we will be using. We will ignore these entirely and focus on the main ideas and applications instead. We refer to \cite{chernoff2006properties} for some of the technical details.} $\mathcal{F}$. Then Hamilton's principle is implemented by the action, which is a function $S_{D}:\mathcal{F}\to\mathbb{R}$ for a given compact oriented $D\subset Q$, together with boundary conditions on the boundary $\partial D$. One can think of boundary conditions in terms of a foliation of $\mathcal{F}$: for a given hypersurface $\sigma\subset Q$, let $P_{\sigma}$ be the foliation of $\mathcal{F}$ such that each leaf is composed of the fields which have the same boundary data on $\sigma$.

Let us denote by $\sigma_{\alpha},\sigma_{\beta}$ two arbitrary Cauchy surfaces in $Q$ which bound an oriented region $D_{\alpha\beta}$. Then the space of solutions of Hamilton's variational principle, $\mathcal{M}\subset\mathcal{F}$ is defined by
\begin{equation}\nonumber
\mathcal{M}=\{\phi\in\mathcal{F}|X\lrcorner dS_{D_{\alpha\beta}}=0,\,\, \forall X\in (P_{\alpha}\cap P_{\beta})_{\phi},\,\, \forall \sigma_{\alpha},\sigma_{\beta}\},
\end{equation}
where $P_{i}=P_{\sigma_{i}}$. To understand this definition note that, if $X\in (P_{\alpha}\cap P_{\beta})_{\phi}\subset T_{\phi}\mathcal{F}$, then $X$ can be seen as an infinitesimal perturbation in the field $\phi$ which is tangent to the the leaves through $\phi $ of both the foliations determined by $\sigma_{\alpha}$ and $\sigma_{\beta}$, so that it is compatible with the boundary conditions on both Cauchy surfaces. Hence $\mathcal{M}$ is the space of fields at which $S$ is stationary with respect to variations compatible with the boundary conditions. In the case where $S$ is the integral of a Lagrangian density, this is obviously equivalent to the Euler-Lagrange equations.

\begin{eulerlagrange}
If there is a first-order Lagrangian density $L=L(\phi(x),\nabla\phi(x),x)$ such that
\begin{equation}\nonumber
S_{D}=\int_{D}L\epsilon 
\end{equation}
and the boundary conditions are that the values of $\phi$ on $\partial D$ should be kept fixed when the action is varied, then the manifold of solutions is
\begin{equation}\nonumber
\mathcal{M}=\left\{\phi\in\mathcal{F}\Bigg|\frac{\partial L}{\partial\phi^{\alpha}}-\nabla_{a}\left(\frac{\partial L}{\partial(\nabla_{a}\phi^{\alpha})}\right)=0\right\},
\end{equation}
where $\nabla$ is the Levi-Civita connection and $\epsilon$ is a volume form on $Q$.
\end{eulerlagrange}
\begin{proof}
In this case we may take $X$ to be an arbitrary field on $Q$ which is supported on a compact subset strictly contained in $D$ and the condition for $\phi\in\mathcal{M}$ is then rewritten as
\begin{equation}\nonumber
\begin{split}
0&=\frac{d}{dt}\left\{\int_{D}L[(\phi+tX)(x),(\nabla\phi+t\nabla X)(x), x]\epsilon\right\}_{t=0}=\int_{D}\left(\frac{\partial L}{\partial\phi^{\alpha}}X^{\alpha}+\frac{\partial L}{\partial(\nabla_{a}\phi^{\alpha})}\nabla_{a}X^{\alpha}\right)\epsilon \\
&= \int_{D}\left(\frac{\partial L}{\partial \phi^{\alpha}}-\nabla_{a}\frac{\partial L}{\partial (\nabla_{a}\phi^{\alpha})}\right)X^{\alpha}\epsilon,
\end{split}
\end{equation}
where we have integrated by parts and dropped the boundary term since $X$ vanishes on $\partial D$. Now, since $X$ is  arbitrary, we are left with the condition
\begin{equation}\nonumber
\left[\frac{\partial L}{\partial \phi^{\alpha}}-\nabla_{a}\frac{\partial L}{\partial (\nabla_{a}\phi^{\alpha})}\right](\phi(x),\nabla\phi(x),x)=0, \,\,\, \forall x\in Q.
\end{equation}
\end{proof}
In this case (the action is the integral of some Lagrangian) one can also find an equation which characterizes the tangent vectors. Let $\phi+tX$ be a curve in $\mathcal{M}$ generated by $X\in T_{\phi}\mathcal{M}$. Then
\begin{equation}\nonumber
0 = \frac{d}{dt}\left\{\left[\frac{\partial L}{\partial\phi^{\alpha}}-\nabla_{a}\frac{\partial L}{\partial(\nabla_{a}\phi^{\alpha})}\right](\phi+tX,\nabla\phi+t\nabla X,x)\right\},
\end{equation}
so that $X$ should be a solution of the linearized equation of motion around the point $\phi$,
\begin{equation}\nonumber
\frac{\partial^{2}L}{\partial\phi^{\beta}\phi^{\alpha}}X^{\beta}+\frac{\partial^{2}L}{\partial(\nabla_{b}\phi^{\beta})\partial\phi^{\alpha}}\nabla_{b}X^{\beta}=\nabla_{a}\left[\frac{\partial^{2}L}{\partial\phi^{\beta}\partial(\nabla_{a}\phi^{\alpha})}X^{\beta}+\frac{\partial^{2}L}{\partial(\nabla_{b}\phi^{\beta})\partial(\nabla_{a}\phi^{\alpha})}\nabla_{b}X^{\beta}\right],
\end{equation}
where it is understood that all the coefficients are evaluated at $\phi$.


There is a standard way in which the action principle introduces a symplectic structure on $\mathcal{M}$. First let $\sigma_{\alpha}$ and $\sigma_{\beta}$ be two disjoint Cauchy surfaces in $Q$. Then since $dS$ vanishes on directions tangent to both $P_{\alpha}$ and $P_{\beta}$, one can decompose
\begin{equation}\label{stheta}
dS_{D_{\alpha\beta}}=\theta_{\alpha}-\theta_{\beta},
\end{equation}
where $X\lrcorner\theta_{i}=0,\,\,\forall X\in P_{i}$. Hence the restriction of, say, $\theta_{\alpha}$ to $\mathcal{M}$ gives a one-form and its exterior derivative is a closed two-form $\omega$ on $\mathcal{M}$. Note that $\omega$ does not depend on the choice of $\sigma_{\alpha}$. For example, we might just as well take $\theta_{\beta}$, as the difference between the two one-forms is exact. If we use the Euler-Lagrange equations to define $\mathcal{M}$, then $\theta_{\alpha}$ can be defined as
\begin{equation}\label{whoistheta}
X\lrcorner\theta_{\alpha}=\int_{\sigma_{\alpha}}X^{\gamma}\frac{\partial L}{\partial(\nabla_{c}\phi^{\gamma})}n_{c}d\sigma,
\end{equation}
where $n^{c}$ is the unit vector normal to $\sigma_{\alpha}$. Indeed, this is consistent with equation (\ref{stheta}):
\begin{equation}\nonumber
\begin{split}
X\lrcorner\theta_{\alpha}-X\lrcorner\theta_{\beta}&=\int_{\sigma_{\alpha}}X^{\gamma}\frac{\partial L}{\partial(\nabla_{c}\phi^{\gamma})}n_{c}d\sigma -\int_{\sigma_{\beta}}X^{\gamma}\frac{\partial L}{\partial(\nabla_{c}\phi^{\gamma})}n_{c}d\sigma = \int_{\partial D_{\alpha\beta}}X^{\gamma}\frac{\partial L}{\partial(\nabla_{c}\phi^{\gamma})}n_{c}d\sigma \\
&= \int_{D_{\alpha\beta}}\nabla_{c}\left(\frac{\partial L}{\partial(\nabla_{c}\phi^{\gamma})}X^{\gamma}\right)\epsilon + \int_{D_{\alpha\beta}}\left[\frac{\partial L}{\partial \phi^{\gamma}}-\nabla_{c}\left(\frac{\partial L}{\partial(\nabla_{c}\phi^{\gamma})}\right)\right]X^{\gamma}\epsilon \\
&= X\lrcorner d\left(\int_{D_{\alpha\beta}}L(\phi,\nabla\phi,x)\epsilon\right) = X\lrcorner dS_{D_{\alpha\beta}},
\end{split}
\end{equation}
where we used Stokes' theorem and then added a term which is zero by the Euler-Lagrange equations. Using this form of the symplectic potentials $\theta_\alpha$, one finds

\begin{omegafield}
The closed two-form $\omega=d\theta_{\alpha}$ is given by
\begin{equation}\nonumber
\omega(X,Y)=\frac{1}{2}\int_{\sigma_{\alpha}}\left[\frac{\partial^{2}L}{\partial\phi^{\beta}\partial(\nabla_{c}\phi^{\gamma})}(X^{\beta}Y^{\gamma}-Y^{\beta}X^{\gamma})+\frac{\partial^{2}L}{\partial(\nabla_{b}\phi^{\beta})\partial(\nabla_{c}\phi^{\gamma})}(Y^{\gamma}\nabla_{b}X^{\beta}-X^{\gamma}\nabla_{b}Y^{\beta})\right]n_{c}d\sigma,
\end{equation}
for $X,Y\in T_{\phi}\mathcal{F}$.
\end{omegafield}
\begin{proof}
Formally,
\begin{equation}\nonumber
\begin{split}
\omega(X,Y) &=\frac{1}{2}Y\lrcorner(X\lrcorner d\theta_{\alpha})=\frac{1}{2}Y\lrcorner(\mathcal{L}_{X}\theta - d(X\lrcorner\theta))=\frac{1}{2}\{X(Y\lrcorner\theta)-Y(X\lrcorner\theta)-[X,Y]\lrcorner\theta\} \\
&= \frac{1}{2}\left\{ X\left(\int_{\sigma_{\alpha}}Y^{\gamma}\frac{\partial L}{\partial(\nabla_{c}\phi^{\gamma})}n_{c}d\sigma\right) - Y\left(\int_{\sigma_{\alpha}}X^{\gamma}\frac{\partial L}{\partial(\nabla_{c}\phi^{\gamma})}n_{c}d\sigma\right) - \int_{\sigma_{\alpha}}[X,Y]^{\gamma}\frac{\partial L}{\partial(\nabla_{c}\phi^{\gamma})}n_{c}d\sigma \right\} \\
& = \frac{1}{2}\int_{\sigma_{\alpha}}\left[\frac{\partial^{2}L}{\partial\phi^{\beta}\partial(\nabla_{c}\phi^{\gamma})}(X^{\beta}Y^{\gamma}-Y^{\beta}X^{\gamma})+\frac{\partial^{2}L}{\partial(\nabla_{b}\phi^{\beta})\partial(\nabla_{c}\phi^{\gamma})}(Y^{\gamma}\nabla_{b}X^{\beta}-X^{\gamma}\nabla_{b}Y^{\beta})\right]n_{c}d\sigma.
\end{split}
\end{equation}
\end{proof}

Although closure is obvious since $\omega$ is defined as $d\theta_{\alpha}$, nondegeneracy is not guaranteed. In particular examples, it can be shown by using the properties of the spaces of solutions of PDE's of certain types. In our (hyperbolic) examples, it will be nondegenerate, hence giving a symplectic strucuture\footnote{A famous example of the geometric quantization of the space of solutions is the case of Chern-Simons theory \cite{witten1989quantum}. We also refer to the notes \cite{nairnotes} for a discussion on this example.} on $\mathcal{M}$.

Because it takes a nice geometric meaning in this formalism, let us look now at Noether's theorem. Let $V\in V(Q)$ be a vector field in $Q$ and $\rho: Q\times\mathbb{R}\to Q$ be its flow. Then choose a lift $V'\in E$ of $V$ to the vector bundle $E$ which projects to $Q$ under $E\to Q$, and denote by $\rho '$ its flow. This allows us to define a flow $\varrho$ in the sections of $E$ by
\begin{equation}\nonumber
(\varrho_{t}\phi)(x)=\rho_{t}'[\phi(x)], \,\,\, \forall x\in Q.
\end{equation}
The vector $V\in V(Q)$ is said to be a symmetry of the variational problem if there is a lift of $V$ to $V(E)$ such that the induced flow in $\mathcal{F}$ preserves the `variational data', by which we mean
\begin{equation}\label{ssym}
S_{\rho_{t}(D)}(\varrho_{t}\phi) = S_{D}(\phi), \,\,\,\, P_{\rho_{t}(\sigma)} = \varrho_{t*}P_{\sigma}, \,\,\, \forall t.
\end{equation}
From the perspective of the space of motions, this will imply that $\varrho$ gives a canonical flow in $(\mathcal{M},\omega)$. Indeed, it implies that both variational problems have the same solutions, $\varrho_{t}(\mathcal{M})=\mathcal{M}$, and also that $\varrho_{t}^{*}\theta_{\rho_{t}(\sigma)}=\theta_{\sigma}$ (compare with expression (\ref{whoistheta})). But then $\varrho_{t}^{*}\omega = \varrho_{t}^{*}d\theta_{\rho_{t}(\sigma)} = d\theta_{\sigma}=\omega$, so the flow is canonical. In fact, we can then use the symplectic structure in $\mathcal{M}$ to find the Hamiltonian function generating this flow, which is the usual constant of motion following from Noether's theorem.

To see this let $X\in V(\mathcal{F})$ be the vector field generating $\varrho :\mathcal{F}\times\mathbb{R}\to\mathcal{F}$ and $\theta_{t}=\theta_{\rho_{t}(\sigma)}$ for some fixed $\sigma$. Then,
\begin{equation}\nonumber
0=\lim_{t'\to 0}\frac{\varrho_{t'}^{*}\theta_{t'+t}-\theta_{t}}{t'} = \mathcal{L}_{X}\theta_{t} = X\lrcorner d\theta_{t}+d(X\lrcorner\theta_{t})+\partial_{t}\theta_{t} = X\lrcorner\omega+d(X\lrcorner\theta_{t})+\partial_{t}\theta_{t},
\end{equation}
where, if we define the functions $\tilde{\theta}^{i}(x,t)$ to be, for each $t$, the components of $\theta_{t}$, then $\partial_{t}\theta_{t}$ is the one-form whose components are $\partial_{t}\tilde{\theta}^{i}$. To calculate this, we use (\ref{stheta}),
\begin{equation}\nonumber
\partial_{t}\theta_{t} = \lim_{t'\to t}\left(\frac{\theta_{t'}-\theta_{t}}{t'-t}\right) = \lim_{t'\to t}\left(\frac{dS_{D_{tt'}}}{t'-t}\right) = \frac{d}{dt'}\left(\int_{D_{tt'}}L\epsilon\right)_{t'=t} = \int_{\sigma_{t}}LV^{c}n_{c}d\sigma,
\end{equation}
where $D_{tt'}$ is bounded by $\sigma_{t}=\rho_{t}(\sigma)$ and $\sigma_{t'}=\rho_{t'}(\sigma)$ (remember that $\rho$ is generated by $V$). From the last two equations (evaluating at $t=0$), $X\lrcorner\omega+dh=0$, where
\begin{equation}\label{hdef}
h = X\lrcorner\theta_{0}+\int_{\sigma}LV^{c}n_{c}d\sigma = \int_{\sigma}\left(X^{\gamma}\frac{\partial L}{\partial(\nabla_{c}\phi^{\gamma})}+LV^{c}\right)n_{c}d\sigma,
\end{equation}
seen as a function on $\mathcal{M}$. An important example of such a symmetry is when $V\in V(Q)$ is a Killing vector and $V'\in V(E)$ is defined by Lie dragging, so that the flow $\varrho$ in $\mathcal{F}$ is given infinitesimally by
\begin{equation}\nonumber
X^{\gamma} = \frac{d}{dt}\left(\varrho_{t}\phi^{\gamma}\right)_{t=0} = -\mathcal{L}_{V}\phi^{\gamma}.
\end{equation}
Then, because $\rho_{t}$ preserves the metric and the connection, $\varrho_{t}(\nabla\phi^{\gamma}) = \nabla(\varrho_{t}\phi^{\gamma})$ and, if the flow preserves the Lagrangian, i.e., if $L(\phi,\nabla\phi,x)=L(\varrho_{t}\phi,\nabla(\varrho_{t}\phi),\rho_{t}x)$, then it satisfies equation (\ref{ssym}), thus giving a symmetry of the variational problem. Moreover, a simple calculation shows that, in this case,
\begin{equation}\nonumber
\begin{split}
\int_{\partial D}X^{\gamma}\frac{\partial L}{\partial(\nabla_{c}\phi^{\gamma})}n_{c}d\sigma &= \int_{D}\nabla_{c}\left(X^{\gamma}\frac{\partial L}{\partial(\nabla_{c}\phi^{\gamma})}\right)\epsilon + \int_{D}\left[\frac{\partial L}{\partial\phi^{\gamma}}-\nabla_{c}\left(\frac{\partial L}{\partial(\nabla_{c}\phi^{\gamma})}\right)\right]\epsilon \\
&=\int_{D}\left(\frac{\partial L}{\partial\phi^{\gamma}}X^{\gamma}+\frac{\partial L}{\partial(\nabla_{c}\phi^{\gamma})}\nabla_{c}X^{\gamma}\right)\epsilon = \frac{d}{dt}\left[\int_{D}L(\varrho_{t}\phi,\varrho_{t}\nabla\phi,x)\epsilon\right]\\
&= \frac{d}{dt}\left(\int_{\rho_{t}^{-1}(D)}[L(\varrho_{t}\phi,\varrho_{t}\nabla\phi,x)\circ\rho_{t}]\rho_{t}^{*}\epsilon\right) = \frac{d}{dt}\left(\int_{\rho_{t}^{-1}(D)}L(\varrho_{t}\phi,\nabla(\varrho_{t}\phi),\rho_{t}x)\epsilon\right)\\
&= \frac{d}{dt}\left(\int_{\rho_{t}^{-1}(D)}L(\phi,\nabla\phi,x)\epsilon\right) = -\int_{\partial D}LV^{c}n_{c}d\sigma \,\,\, \Rightarrow \\
\Rightarrow \,\,\,\,\,\, 0 = \int_{\partial D} & \left(X^{\gamma}\frac{\partial L}{\partial(\nabla_{c}\phi^{\gamma})}+V^{c}L\right)n_{c}d\sigma = \int_{D}\nabla_{c}\left(X^{\gamma}\frac{\partial L}{\partial(\nabla_{c}\phi^{\gamma})}+V^{c}L\right)\epsilon,
\end{split}
\end{equation}
for arbitrary compact $D$. Therefore, the last bracketed expression is the divergenceless current implied by Noether's theorem. Furthermore, the last line says that the same integral by which we defined $h$, when calculated over an arbitrary closed surface $\partial D$ (we assume that $H^{d-1}(Q)=\{0\}$), vanishes. Hence $h$ is actually independent of the Cauchy surface $\sigma$ in expression (\ref{hdef}).

We see that the existence of a symmetry on $\mathcal{M}$ implies the existence of the conserved quantity $h$, Noether's conserved charge. Moreover, the presence of a symplectic structure on $\mathcal{M}$ allows us to identify the conserved quantity as the corresponding Hamiltonian generating the symmetry.

\begin{classical}\emph{Classical mechanics}

A simple example is just classical mechanics. Here, spacetime $Q$ is the time axis $\mathbb{R}$ with metric $dt^{2}$ and the fields are the coordinates $q^{a}(t)$. $\mathcal{M}$ is the space of solutions $q^{a}(t)$ of the Euler-Lagrange equations
\begin{equation}\nonumber
\frac{\partial L}{\partial q^{a}}-\frac{d}{dt}\left(\frac{\partial L}{\partial \dot{q}^{a}}\right)=0.
\end{equation}
And the symplectic structure is given by
\begin{equation}\nonumber
\omega(X,Y) = \frac{1}{2}\left[\frac{\partial^{2}L}{\partial q^{a}\partial\dot{q}^{b}}(X^{a}Y^{b}-Y^{a}X^{b})+\frac{\partial^{2}L}{\partial\dot{q}^{a}\partial\dot{q}^{b}}(Y^{b}\dot{X}^{a}-X^{b}\dot{Y}^{a})\right],
\end{equation}
where $X$ and $Y$ are solutions of the linearized form of the Euler-Lagrange equations. Note that this does not have the integration sign because a Cauchy surface is simply a point $t'\in\mathbb{R}$.

If the Lagrangian is time independent then $-\partial/\partial t$ is a Killing vector of the metric and a symmetry according to the above criteria. In this case, the flow on the space of solutions is $R_{s}:q(t)\mapsto q(t+s)$ and the corresponding conserved quantity is
\begin{equation}\nonumber
h=\dot{q}^{a}\frac{\partial L}{\partial\dot{q}^{a}} - L,
\end{equation}
which we recognize as the energy Hamiltonian. Hence one recovers the standard connection between symmetry under time translations and conservation of energy.
\end{classical}

\begin{fieldsymplectic}\emph{Fields of particles}

One can also consider the spaces of solutions of the relativistic wave equations derived in the previous chapter, which we found to be the quantum Hilbert spaces corresponding to free relativistic particles. Applying the theory above we see that these are themselves symplectic manifolds as well. The symplectic structures arise from the fact that the wave equations can be seen as Euler-Lagrange equations for certain Lagrangians. For example, let $\phi$ be a complex function on spacetime $Q$ and consider the Lagrangian
\begin{equation}\nonumber
L=\frac{1}{2}(\nabla_{a}\phi\nabla^{a}\bar{\phi} - \mu^{2}\phi\bar{\phi}).
\end{equation}
The corresponding equation of motion is the wave equation for a scalar particle,
\begin{equation}\nonumber
(\Box + \mu^{2})\phi = 0.
\end{equation}
Furthermore, taking advantage of the linear structure of the space of solutions $V$ to identify $T_{\phi}\mathcal{M}$ at any $\phi$ with $\mathcal{M}$ itself, one can write the two-form $\omega$ explicitly as
\begin{equation}\nonumber
\omega(\phi,\phi ') = \frac{1}{4}\int_{\sigma}(\phi '\nabla_{a}\bar{\phi}+\bar{\phi}'\nabla_{a}\phi - \phi\nabla_{a}\bar{\phi}' - \bar{\phi}\nabla_{a}\phi ')n^{a}d\sigma .
\end{equation}

Likewise, the solution spaces of all the wave equations we saw are infinite-dimensional symplectic vector spaces. The quantization of such vector spaces is called \emph{second quantization} and leads to \emph{quantum fields} as we now discuss.

\end{fieldsymplectic}

\pagebreak

\subsection{Fock Space}

In this section we consider the geometric quantization of a symplectic vector space, which in the case of the space of solutions of the free particle wave equations will lead to the Fock space picture of a quantum field. Let $(V,\omega)$ be a $2n$-dimensional symplectic vector space with a symplectic frame $(p_{a},q^{b})$ and consider the one-form $\theta_{0}$ invariantly defined as
\begin{equation}\nonumber
(X\lrcorner\theta_{0})(Y)=-\omega(X,Y), \,\,\, \forall X,Y\in V,
\end{equation}
implicitly using the fact that $V$ is a vector space to identify $T_{X}V=V, \,\, \forall X\in V$. In the frame chosen it is given by $\theta_{0}=\frac{1}{2}(p_{a}dq^{a}-q^{a}dp_{a})$.

Now, the trivial topology of $(V,\omega)$ guarantees the existence of a unique prequantum bundle $B\to V$ by proposition \ref{critpreq}. Let $s:V\to B$ be the section specified by $\theta_{0}$, ie., such that $Ds=-\frac{i}{\hbar}\theta_{0}\otimes s$. Then, one can write any other section as $s'=\psi(p,q)s$ for some $\psi:V\to\mathbb{C}$. Each element $X\in V$, seen as a constant vector field in the manifold $V$, generates a hamiltonian flow: let
\begin{equation}\nonumber
W=u_{a}\frac{\partial}{\partial p_{a}}+v^{b}\frac{\partial}{\partial q^{b}}\,\, \in V(M=V), \,\,\,\, u_{a}, v^{b}\in\mathbb{R}.
\end{equation}
Then
\begin{equation}\nonumber
W\lrcorner\omega = \left( u_{a}\frac{\partial}{\partial p_{a}}+v^{b}\frac{\partial}{\partial q^{b}}\right)\lrcorner dp_{a}\wedge dq^{a} = -d(v^{a}p_{a}-u_{b}q^{b}) =:-df(p,q),
\end{equation}
so that $W=X_{f}$ for $f=v^{a}p_{a}-u_{b}q^{b}$. This definition can be made coordinate independent by $f:V\to\mathbb{R}:X\mapsto 2\omega(X,W)$. Note, however, that this $f$ does not define a moment for the translation action of the Abelian group $V$ on the manifold $V$. Indeed, let $W, Z\in V$ and define $f,g\in C^{\infty}(V)$ by $W=X_{f},Z=X_{g}$ as above, and write $f=v^{a}p_{a}-u_{b}q^{b}, \,\, g=d^{a}p_{a}-c_{b}q^{b}$. Then
\begin{equation}\nonumber
\{f,g\} = -v^{a}c_{a}-u_{b}d^{b},
\end{equation}
which is in general not zero while, since $X_{f}=W, X_{g}=Z$ are then seen as constant vector fields, $[X_{f},X_{g}]=[W,Z]=\mathcal{L}_{W}Z=0\neq X_{\{f,g\}}$.

Now, using the general prescription (\ref{expquant}), the observable $f$ generating $W$ is quantized to
\begin{equation}\nonumber
\begin{split}
\hat{f}(\psi s) &= -i\hbar\left[X_{f}(\psi)-\frac{i}{\hbar}(X_{f}\lrcorner\theta_{0})\psi\right]s+f\psi s = [-i\hbar X(\psi)-(X\lrcorner\theta_{0})\psi + f\psi ]s \\
&= \left[-i\hbar W(\psi)-\left(u_{a}\frac{\partial}{\partial p_{a}}+v^{b}\frac{\partial}{\partial q^{b}}\right)\lrcorner\frac{1}{2}(p_{a}dq^{a}-q^{b}dp_{b})\psi + (v^{a}p_{a}-u_{b}q^{b})\psi\right]s \\
&= \left[-i\hbar W(\psi)+\frac{1}{2}f\psi\right] s.
\end{split}
\end{equation}

Again, just as $W\mapsto f$ did not define a momentum map, $\hat{f}$ does not give rise to a unitary representation of the Abelian group $V$. Indeed, if $f$ is such that $X_{f}=W$, then translation by $W$ is given by $\rho^{f}_{1}:V\to V$, where $\rho^{f}:V\times\mathbb{R}\to V$ is the Hamiltonian flow defined by $f$. Recall that this flow can be lifted to a unitary flow $\hat{\rho}^{f}:C^{\infty}(B)\times\mathbb{R}\to C^{\infty}(B)$ on the sections of the prequantum bundle $B\to M=V$, and this flow is generated by $\hat{f}$ according to equation (\ref{absf}). Thus one might try to make the translation by $W$ act on the space of sections by $\hat{W}:=\hat{\rho}^{f}_{1}$. Using the explicit form of the lifted flow, given in (\ref{flowsection}),
\begin{equation}\nonumber
\begin{split}
(\hat{W}\psi)(X) &= (\hat{\rho}^{f}_{1}\psi)(X) = \psi(\rho^{f}_{1}X)\exp\left(-\frac{i}{\hbar}\int_{0}^{1}[X_{f}\lrcorner\theta_{0}-f]dt'\right) \\
&= \psi(X+W)\exp\left(-\frac{i}{2\hbar}\int_{0}^{1}[u_{a}q^{a}(t')-v^{b}p_{b}(t')]dt'\right) \\
&= \psi(X+W)\exp\left(-\frac{i}{2\hbar}\int_{0}^{1}[u_{a}(q^{a}+v^{a}t')-v^{b}(p_{b}+u_{b}t')]dt'\right) \\
&= \psi(X+W)e^{-i(u_{a}q^{a}-v^{b}p_{b})/2\hbar} \\
&= \psi(X+W)e^{-i\omega(W,X)/\hbar},\,\,\,\, \forall X\in V,
\end{split}
\end{equation}
where we have simplified the notation denoting the section $\psi s$ by just $\psi$. These operators actually compose as the elements of the Heisenberg group $V\ltimes S^{1}$ with the circle group acting by multiplication. Indeed, if $(W,w),(Z,z)\in V\times S^{1}$, then one has the consistent representation
\begin{equation}\nonumber
\begin{split}
[(W,w)(Z,z)\psi](X) &= [(W,w)z\hat{Z}\psi](X) = (W,w)z\psi(X+Z)e^{i\omega(Z,X)} = wz\hat{W}[\psi(X+Z)e^{i\omega(Z,X)}] \\
&= wz\psi(X+W+Z)e^{i\omega(Z,X+W)}e^{i\omega(W,X)} \\
&= wz e^{i\omega(Z,W)}\psi[X+(W+Z)]e^{i\omega[(W+Z),X]} \\
&= [(W+Z,wz e^{i\omega(Z,W)})\psi](X) = \{[(W,w)\circ_{H}(Z,z)]\psi\}(X), \,\,\,\, \forall X\in V.
\end{split}
\end{equation}
The appearance of the Heisenberg group is a well-known fact in both quantum theory and geometric representation theory \cite{kirillov2004lectures}.

Moving on to quantization, it is especially convenient to use a positive K\"ahler polarization, which is equivalent to a complex structure on $TV$. Again, by identifying $T_{X}V=V$, one only needs to specify a positive complex structure on $V$. These can be simply described.
\begin{cxg}
Let $(V,\omega)$ be a $2n$-dimensional symplectic vector space. Then $J$ is a positive symplectic structure compatible with $\omega$ if, and only if, there is a symplectic frame ${X^{a},Y_{b}}$ and a real symmetric positive definite matrix $(g_{ab})$ such that
\begin{equation}\nonumber
JX^{a}=g^{ab}Y_{b}, \,\,\,\,\,\,\, JY_{a}=-g_{ab}X^{b},
\end{equation}
where $g^{ab}g_{bc}=\delta^{a}_{c}$.
\end{cxg}
\begin{proof}
($\Rightarrow$) That $J$ is linear, canonical, and $J^{2}=- 1$ follow directly from expressing vectors of $V$ in the referred symplectic frame. Proposition \ref{gclag} implies that the signature of $g$ gives the type of $J$, which is then positive.

($\Leftarrow$) Suppose $J$ is a positive compatible complex structure on $(V,\omega)$. Then we can find a basis $\{Y_{1},...,Y_{n}\}$ of $V_{(J)}$ diagonalizing the bilinear form $\langle\cdot , \cdot\rangle_{J}$ defined in proposition \ref{gclag}, so that
\begin{equation}\nonumber
\langle Y_{a}, Y_{b}\rangle_{J} = \delta_{ab}.
\end{equation}
Now, define $X^{a} = -g^{ab}JY_{b}$, where $g^{ab}g_{bc}=\delta^{a}_{b}$. Then $\{X^{a},Y_{b}\}$ is the referred symplectic frame and $(g_{ab})$ the positive definite matrix.
\end{proof}

Therefore, take such a complex structure $J$, so that $V_{(J)}$ is a flat Kahler manifold and choose coordinates $(p_{a},q^{b})$ such that
\begin{equation}\nonumber
J\frac{\partial}{\partial p_{a}} = g^{ab}\frac{\partial}{\partial q^{b}}, \,\,\,\,\,\,\,\,\, J\frac{\partial}{\partial q^{a}} = -g_{ab}\frac{\partial}{\partial p_{b}},
\end{equation}
with $J$ ($\Rightarrow g$) positive. Note that the coordinates $z^{a}=g^{ab}p_{b}+iq^{a}$ are then holomorphic since
\begin{equation}\nonumber
J\frac{\partial}{\partial z^{a}} = J\left[\frac{1}{2}\left(g_{ba}\frac{\partial}{\partial p_{b}} - i\frac{\partial}{\partial q^{a}}\right)\right] = i\frac{\partial}{\partial z^{a}},
\end{equation}
and likewise $J\partial/\partial\bar{z}^{a} = -i\partial/\partial\bar{z}^{a}$. In these coordinates, $K=\frac{1}{2}g_{ab}z^{a}\bar{z}^{b}$ is a Kahler scalar, since
\begin{equation}\nonumber
i\partial\bar{\partial}K = i\partial\left(\frac{1}{2}g_{ab}z^{a}d\bar{z}^{b}\right) = \frac{i}{2}g_{ab}dz^{a}d\bar{z}^{b} = \frac{i}{2}g_{ab}d(g^{ac}p_{c}+iq^{a})\wedge d(g^{bd}p_{d}-iq^{b}) = dp_{a}\wedge dq^{a} = \omega.
\end{equation}
From this equation it is also obvious that $\theta = -\frac{i}{2}g_{ab}\bar{z}^{a}dz^{b}$ is a symplectic potential, and it is adapted to the holomorphic polarization $P = \text{span}\{\partial/\partial z^{a}\}$. We can use this potential to represent sections of the prequantum bundle $B\to V$ by $s' = \phi \tilde{s}$, where $D\tilde{s} = -\frac{i}{\hbar}\theta\otimes \tilde{s}$. By definition, one such section is polarized along the holomorphic polarization $P$ if, and only if,
\begin{equation}\nonumber
0 = \nabla_{\bar{X}}(\phi\tilde{s}) = \left[\bar{X}(\phi)-\frac{i}{\hbar}(\bar{X}\lrcorner\theta)\phi\right]\tilde{s} = \bar{X}(\phi)\tilde{s}, \,\, \forall X\in V_{P}(M=V),
\end{equation}
that is, if, and only if, $\phi$ is an entire holomorphic function of the coordinates $z^{a}$. Note that we now have at hand two interesting symplectic potantials, $\theta_{0}$, which determines the frame $s:V\to B$, and $\theta$, which is adapted to $P$ and determines the frame $\tilde{s}:V\to B$. It will be usefull to know how to translate from one trivialization to the other, so let $s' = \psi s = \phi\tilde{s}$. The potentials are related by
\begin{equation}\nonumber
\begin{split}
\theta +\frac{i}{2}dK &= -\frac{i}{2}g_{ab}\bar{z}^{b}dz^{a}+\frac{i}{4}g_{ab}\bar{z}^{b}dz^{a}+\frac{i}{4}g_{ab}z^{a}d\bar{z}^{b} \\
&= \frac{i}{4}g_{ab}[-ig^{ac}p_{c}dq^{b}+ig^{bd}q^{a}dp_{d}-ig^{bd}p_{d}dq^{a}+ig^{ac}q^{b}dp_{c}] \\
&= \frac{1}{2}(p_{a}dq^{a}-q^{b}dp_{b}) = \theta_{0},
\end{split}
\end{equation}
where we used $g_{ab}=g_{ba}$. Now, each potential defines the corresponding trivialization by equation (\ref{sfromtheta}), so that
\begin{equation}\nonumber
s(\gamma_{t}m) = b\exp\left(-\frac{i}{\hbar}\int_{m}^{\gamma_{t}m}\theta_{0}\right) = b\exp\left(-\frac{i}{\hbar}\int_{m}^{\gamma_{t}m}\theta+\frac{i}{2}dK\right) = \tilde{s}(\gamma_{t}m)e^{K/2\hbar}.
\end{equation}
For the notation, see the discussion preceding equation (\ref{sfromtheta}). So, if $s' = \psi s = \phi\tilde{s} = \phi e^{-K/2\hbar}s$, then $\psi = \phi e^{-K/2\hbar}$. In particular, we conclude that the sections of $B$ which are polarized along the holomorphic polarization $P$ have the local expressions
\begin{equation}\nonumber
s' = \psi(z,\bar{z}) s= \phi(z)e^{-z_{a}\bar{z}^{a}/4\hbar}s
\end{equation}
in the frame specified by $\theta_{0}$, where $\phi$ is holomorphic on the coordinates $z^{a}$ and $z_{a}\bar{z}^{a}:=g^{ab}z_{a}\bar{z}_{b}$.

The Hermitian structure on $C^{\infty}_{P}(V)$ is given by equation (\ref{ints}), so
\begin{equation}\nonumber
\langle\psi s,\psi ' s\rangle = \langle\phi\tilde{s},\phi '\tilde{s}\rangle = \int_{V}(\phi\tilde{s},\phi '\tilde{s})\epsilon = \int_{V}\bar{\phi}\phi ' e^{-K/\hbar}\epsilon = \int_{V}\bar{\phi}\phi ' e^{-g_{ab}z^{a}\bar{z}^{b}/2\hbar}\epsilon = \int_{V}\bar{\psi}\psi '\epsilon,
\end{equation}
where we have written $(g_{ab})$ explicitly to illustrate why it was assumed from the beginning that $J$ was a positive complex structure: this is what guarantees that $g$ is positive definite and hence that the integral converges for a wide class of wavefunctions.

Since $P$ is Kahler, $f\in C^{\infty}(V)$ preserves $P$ if, and only if, it is real and linear in the holomorphic coordinates $z^{a}$. So the most general such observable can be written
\begin{equation}\label{classicalf}
f = \frac{i}{2}\bar{w}_{a}z^{a}-\frac{i}{2}w_{a}\bar{z}^{a}+\frac{1}{2}U_{ab}z^{a}\bar{z}^{b}+c,
\end{equation}
with $c\in\mathbb{R}$ and $\bar{U}_{ab}=U_{ba}$. We find the expression for the quantum analogue of this general observable in the $\theta_{0}$ `gauge':
\begin{equation}\nonumber
X_{f}\lrcorner\left(\frac{i}{2}g_{ab}dz^{a}\wedge d\bar{z}^{b}\right) = -\frac{i}{2}\bar{w}_{a}dz^{a}+\frac{i}{2}w_{a}d\bar{z}^{a}-\frac{1}{2}U_{ab}z^{a}d\bar{z}^{b}-\frac{1}{2}U_{ab}\bar{z}^{b}dz^{a}
\end{equation}
\begin{equation}\nonumber
\Rightarrow
\begin{cases}
X_{f} = \left(w^{a}+i\tensor{U}{_b^a}z^{b}\right)\frac{\partial}{\partial z^{a}} - \left(\bar{w}^{a}+i\tensor{U}{^a_b}\bar{z}^{b}\right)\frac{\partial}{\partial \bar{z}^{a}} \\
X_{f}\lrcorner\theta_{0} = -\frac{i}{4}(w^{a}\bar{z}_{a}+\bar{w}^{a}z_{a})+\frac{1}{2}U_{ab}z^{a}\bar{z}^{b}
\end{cases}
\end{equation}
\begin{equation}\label{quantumf}
\Rightarrow
\hat{f}(\phi e^{-z_{a}\bar{z}^{a}/4\hbar}s) = \left(-i\hbar w^{a}\frac{\partial\phi}{\partial z^{a}} + \frac{i}{2}\bar{w}_{a}z^{a}\phi + \hbar\tensor{U}{_a^b}z^{a}\frac{\partial\phi}{\partial z^{b}} + c\phi\right)e^{-z_{a}\bar{z}^{a}/4\hbar} s.
\end{equation}

Remember that, for $W\in V$, the observable $f$ generating $W$ is real and linear, so that $\hat{W}(\mathcal{H}_{P})=\mathcal{H}_{P}$, where $\mathcal{H}_{P}$ is the space of polarized sections. In the local representation of the elements of $\mathcal{H}_{P}$ it is given by
\begin{equation}\label{wpsi}
\begin{split}
\hat{W}[\phi(z)e^{-z_{a}\bar{z}^{a}/4\hbar}] &= \phi(z+w)e^{-(z_{a}+w_{a})(\bar{z}^{a}+\bar{w}^{a})/4\hbar}e^{\frac{i}{\hbar}\left[\frac{i}{4}g_{ab}(w^{a}\bar{z}^{b}-z^{a}\bar{w}^{b})\right]} \\
&= \phi(z+w)e^{-\frac{1}{4\hbar}(2\bar{w}_{a}z^{a}+w_{a}\bar{w}^{a}+z_{a}\bar{z}^{a})},
\end{split}
\end{equation}
where $w^{a}$ are the holomorphic coordinates of $W\in V$.

We define the vacuum state $\psi_{0}s$ to be the one represented by $\phi_{0}(z)=1$ and the coherent state based at the point $W\in V$ as the translation of the vacuum by $W$, that is,
\begin{equation}\label{coherentexplicit}
\begin{split}
\psi_{W}s &= -\hat{W}(\phi_{0}(z)e^{-z_{a}\bar{z}^{a}/4\hbar}s) = \phi_{0}(z-w) e^{-\frac{1}{4\hbar}(-2\bar{w}_{a}z^{a}+\bar{w}_{a}w^{a}+z_{a}\bar{z}^{a})} \\
&= [1 e^{-\frac{1}{4\hbar}(-2\bar{w}_{a}z^{a}+w_{a}\bar{w}^{a})}]e^{-z_{a}\bar{z}^{a}/4\hbar}\,\,\,
 \Rightarrow \,\,\, \phi_{W}(z)=e^{(2\bar{w}_{a}z^{a}-w_{a}\bar{w}^{a})/4\hbar}.
\end{split}
\end{equation}
These are localized states which span $\mathcal{H}_{P}$. To see this, let $\tilde{\psi}s = \tilde{\phi}e^{-z_{a}\bar{z}^{a}/4\hbar}s := \hat{W}(\psi s)$. Then we have
\begin{equation}\label{coherent}
\begin{split}
\langle\psi_{W}s, \psi s\rangle &= \langle (-\hat{W})\psi_{0}s,\psi s\rangle = \langle\psi_{0} s,\hat{W}\psi s\rangle = \int_{V_{(J)}}\bar{\phi}_{0}\tilde{\phi}e^{-K/\hbar}\epsilon \\
&= \int_{\mathbb{C}^{n-1}}\left[\int_{\mathbb{C}}\tilde{\phi}(z_{1},z_{2},...)e^{-z_{1}\bar{z}_{1}/2\hbar}\frac{d^{2}z_{1}}{(2\pi\hbar)}\right]e^{-\sum_{i=2}^{n}z_{i}\bar{z}_{i}/2\hbar}\frac{d^{2(n-1)}z}{(2\pi\hbar)^{n-1}} \\
&= \int_{\mathbb{C}^{n-1}}\left[\int_{0}^{\infty}\int_{0}^{2\pi}\tilde{\phi}\frac{d\alpha}{2\pi}e^{-r^{2}/2\hbar}\frac{r}{\hbar}dr\right]e^{-\sum_{i=2}^{n}z_{i}\bar{z}_{i}/2\hbar}\frac{d^{2(n-1)}z}{(2\pi\hbar)^{n-1}} \\
&= \int_{\mathbb{C}^{n-1}}\tilde{\phi}(0,z_{2},...)e^{-\sum_{i=2}^{n}z_{i}\bar{z}_{i}/2\hbar}\frac{d^{2(n-1)}z}{(2\pi\hbar)^{n-1}} \\
&= \tilde{\phi}(0) = \phi(w)e^{-w_{a}\bar{w}^{a}/4\hbar} = \psi(W),
\end{split}
\end{equation}
where the identification of $V_{(J)}$ with $\mathbb{C}^{n}$ amounts symply to the choice of a frame in which $g_{ab}=\delta_{ab}$. We used polar coordinates in $\mathbb{C}$ and Cauchy's theorem $n$ times, then used equation (\ref{wpsi}) in the last step. It follows that any $\psi s\in \mathcal{H}_{P}$ can be expressed in terms of the coherent states as
\begin{equation}\label{projincoherent}
\psi s = \left(\int_{V}\frac{\psi(W)}{\langle\psi ,\psi\rangle}\psi_{W}dW\right)s.
\end{equation}

Because the functions $\phi(z)$ are entire holomorphic, their Laurent series are of the form
\begin{equation}\label{phiseries}
\phi (z) = \sum_{k=0}^{\infty}\phi_{a_{1}a_{2}...a_{k}}z^{a_{1}}z^{a_{2}}...z^{a_{k}},
\end{equation}
for some constants $\phi_{a_{1}...a_{k}}\in\mathbb{C}$ (sum in the $a_{i}$ is implicit). One can then perform a calculation similar to (\ref{coherent}) and use the referred property of the coherent states to write the inner product in $\mathcal{H}_{P}$ in terms of the $\phi_{a_{1}...a_{k}}$.
\begin{equation}\label{innerphi}
\begin{split}
\langle\psi s,\psi s\rangle &= \int_{V_{(J)}}\bar{\psi}(W)\psi(W)\epsilon_{W} = \int_{V_{(J)}}\langle\psi s, \psi_{W}s\rangle\langle\psi_{W}s,\psi s\rangle\epsilon_{W} \\
&= \int_{V_{(J)}}\left[\int_{V_{(J)}}\bar{\phi}(z)\phi_{W}(z)e^{-K(z)/\hbar}\epsilon_{z}\right]\left[\int_{V_{(J)}}\bar{\phi}(y)\phi(y)e^{-K(y)/\hbar}\epsilon_{y}\right]\epsilon_{W}\\
&=\int_{V_{(J)}}\left[\int_{V_{(J)}}\left(\sum_{i=0}^{\infty}\bar{\phi}_{a_{1}...a_{i}}\bar{z}^{a_{1}}...\bar{z}^{a_{i}}\right)\left(e^{\bar{w}_{a}z^{a}/2\hbar}e^{-w_{a}\bar{w}^{a}/4\hbar}\right)e^{-z_{a}\bar{z}^{a}/2\hbar}\epsilon_{z}\right]\times \\
& \times \left[\int_{V_{(J)}}\left(e^{w_{a}\bar{y}^{a}/2\hbar}e^{-w_{a}\bar{w}^{a}/4\hbar}\right)\left(\sum_{j=0}^{\infty}\phi_{b_{1}...b_{j}}y^{b_{1}}...y^{b_{j}}\right)e^{-y_{a}\bar{y}^{a}/2\hbar}\epsilon_{y}\right]\epsilon_{W} \\
&= \int_{V_{(J)}}\left(\sum_{i,j=0}^{\infty}\bar{\phi}_{a_{1}...a_{i}}\bar{w}^{a_{1}}...\bar{w}^{a_{i}}\phi_{b_{1}...b_{j}}w^{b_{1}}...w^{b_{j}}\right)e^{-w_{a}\bar{w}^{a}/2\hbar}\epsilon_{W}\\
&=\sum_{i=0}^{\infty}(2\hbar)^{i}i!\bar{\phi}_{a_{1}...a_{i}}\phi^{a_{1}...a_{i}}.
\end{split}
\end{equation}
We used the explicit form of the coherent states (\ref{coherentexplicit}), expanded some of the exponentials, and used Cauchy's theorem.

Thus we can picture $\mathcal{H}_{P}$ in another way: let $\mathsf{H}_{1}=(V_{(J)})^{*}$, the (complex) dual vector space. Equation (\ref{phiseries}) implies that the $\phi_{a}$'s transform as the components of an element of $\mathsf{H}_{1}$. Likewise, for each $i$, $\phi_{a_{1}...a_{i}}$ can be seen as the components of an element of the symmetrization of $\otimes^{i}\mathsf{H}_{1}$, which we denote by $\mathsf{H}_{i}$. Therefore the decomposition (\ref{phiseries}) provides an identification of $\mathcal{H}_{P}$ with
\begin{equation}\nonumber
\mathsf{F}:=\oplus_{i=0}^{\infty}\mathsf{H}_{i},
\end{equation}
where $\mathsf{H}_{0}:=\mathbb{C}$. Finally, after a rescalling of the components in $\mathsf{H}_{i}$ the inner product of $\mathcal{H}_{P}$ is mapped to the natural inner product in $\mathsf{F}$, as shown in equation (\ref{innerphi}). We conclude that the quantum states are elements of the \emph{Fock Space} $\mathsf{F}$.

\begin{sho}\label{shorr}
The simplest example of Fock space quantization is the simple harmonic oscillator. In this case we take $V=\mathbb{R}^{2}$ with coordinates $(p,q)$, $\omega = dp\wedge dq$, and symplectic structure given by 
\begin{equation}\nonumber
\begin{cases}
J(\partial/\partial p) = \partial/\partial q\\
J(\partial/\partial q) = -\partial/\partial p.
\end{cases}
\end{equation}
Then $z=p+iq$ is a holomorphic coordinate, $P=\text{span}\{\partial/\partial z\}$ and, in the trivialization $s$ specified by $\theta_{0}=\frac{1}{2}(pdq - qdp)$, an element of $\mathcal{H}_{P}$ is of the form $\phi(z)e^{-z\bar{z}/4\hbar}$, $\phi$ entire. The Hamiltonian is
\begin{equation}\nonumber
H = \frac{1}{2}(p^{2}+q^{2}) = \frac{1}{2}z\bar{z},
\end{equation}
which is in the form (\ref{classicalf}), with only the quadratic part $(U_{ab}) = (1)$. Hence the application of (\ref{quantumf}) gives
\begin{equation}\nonumber
\hat{H}[\phi(z)e^{-z\bar{z}/4\hbar}s] = \hbar z\frac{\partial\phi}{\partial z}(z)e^{-z\bar{z}/4\hbar}s,
\end{equation}
so, in the trivialization $\tilde{s}$ chosen by $\theta = -\frac{i}{2}\bar{z}dz$, it acts as $\hat{H}:\phi\tilde{s}\mapsto\hbar z\frac{\partial\phi}{\partial z}\tilde{s}$. We consider also the functions $z, \bar{z}$, which do not preserve $P$ and hence are not quantized by the rule (\ref{quantumf}). Nevertheless, we can repeat the derivation in the $\tilde{s}$ frame
\begin{equation}\nonumber
\begin{cases}
X_{z}\lrcorner\left(\frac{i}{2}dz\wedge d\bar{z}\right) = -dz\\
X_{\bar{z}}\lrcorner\left(\frac{i}{2}dz\wedge d\bar{z}\right) = -d\bar{z}
\end{cases}
\Rightarrow
\begin{cases}
X_{z} = -2i\frac{\partial}{\partial\bar{z}}\\
X_{\bar{z}} = 2i\frac{\partial}{\partial z}
\end{cases}
\Rightarrow
\begin{cases}
\hat{z}[\phi(z)\tilde{s}] = z\phi(z)\tilde{s}\\
\hat{\bar{z}}[\phi(z)\tilde{s}] = 2\hbar\frac{\partial\phi}{\partial z}(z)\tilde{s}.
\end{cases}
\end{equation}
\end{sho}

The indentification with Fock space $\mathsf{F}$ is given by the expression of $\phi(z)$ as a polinomial in $z$ so that $\mathsf{H}_{i}$ corresponds to the monomials of degree $i$. Therefore each space $\mathsf{H}_{n}$ is an eigenspace of the hamiltonian and the eigenvalue is given by
\begin{equation}\nonumber
\hat{H}|_{\mathsf{H}_{n}} = n\hbar 1,
\end{equation}
which is shifted from the corrected value by $\frac{1}{2}\hbar$! This shift will be corrected in the next section. Also, we recognize $\hat{z}$ and $\hat{\bar{z}}$ as the raising and lowering operators of the SHO. In particular, $\bar{z}(\mathsf{H}_{n}) = \mathsf{H}_{n+1}$ and $\hat{\bar{z}}(\mathsf{H}_{n}) = \mathsf{H}_{n-1}$.

\begin{fockfield}\emph{Quantum fields}

In the case of the spaces of solutions of field equations, the above manipulations are formal, but one use the analogy with the finite-dimensional case to borrow the well-defined Fock space as its quantization. This depends on the introduction of a positive K\"ahler polarization, which again we fix using a given positive compatible complex structure on $V$. Interestingly, the obvious complex structure $\phi_{A_{1}...A_{n}}\mapsto i\phi_{A_{1}...A_{n}}$ is not positive, while the complex structure $J$ that multiplies the positive frequency part by $-i$ and the negative frequency part by $i$ is. Then, using this complex structure $J$, quantization leads to
\begin{equation}\nonumber
\mathsf{F}=\oplus_{i=0}^{\infty}\mathsf{H}_{i},
\end{equation}
where $\mathsf{H}_{i}$ is the symmetrization of the i-th tensor power of $(V_{(J)})^{*}$. But we saw in example \ref{wavesss} that this is exactly the Hilbert space of the one-particle wavefunctions. Hence each $\mathsf{H}_{i}$ is the Hilbert space correspoding to $i$ identical quantum relativistic particles of the corresponding type. The operators analogous to the raising and lowering operators of the SHO are the creation and anihilation operators. Thus one recovers the particle interpretation of a quantum field.

\end{fockfield}

\pagebreak

\end{preqsu2}

\section{The metaplectic correction}\label{metaplectic}

As we saw above, prequantization gives a general and concrete geometric construction of wave functions and quantum operators which satisfies Dirac's quantization rules.  It accomplished that, however, at the cost of introducing nonphysical states in the Hilbert space, a problem which is resolved by the introduction of a polarization in the second step. Quantization gives a geometrical way of choosing the right states and, consequently, also selecting the correct subalgebra of the Poisson algebra to be quantized. An interesting question to ask is whether this whole process is possible, and if yes, whether it is unique. The conditions for the existance and uniqueness of the prequantum bundle were mentioned to be given by the topology of the symplectic structure in $(M,\omega)$. The situation with the polarization is much more subtle. In particular, making the constructed Hilbert space and quantum operators independent of the choice of polarization leads to a whole new (and final) step of geometric quantization: the metaplectic correction,  which we now address. As we will see in the examples, this step is far from being a mere mathematical technicality but has physical consequences: for example, correcting the spectrum of the simple harmonic oscillator, explaining the transformation law of the dilaton field under T-duality symmetry of the partition function of a bosonic string, and even motivating mirror symmetry \cite{tyurin2000special,gukov2011quantization}. We also note that, although not addressed here, the question on the existence of a polarization in a given symplectic manifold is also very interesting \cite{karasev1984pseudodifferential}.

\subsection{Metaplectic representation}

Recall the quantization of a symplectic vector space $(V,\omega)$ performed in the previous section. For each positive complex structure $J$, geometric quantization constructed an associated Fock Space $\mathsf{F}_{J}$, which is an irreducible representation of the Heisenberg group $V\ltimes S^{1}$. Stone von-Neuman's theorem says that this representation is unique up to a unitary transformation, so that all of the representations $\mathsf{F}_{J}$ should be unitarily related. This is indeed the case.

\begin{projectj}
Let $(V,\omega)$ be a symplectic vector space, $\mathsf{F}_{\alpha}$ be the Fock Space constructed from $(V,\omega)$ and a positive compatible complex structure $J_{\alpha}$ on $V$ by geometric quantization, and $\pi_{\alpha\beta}:\mathsf{F}_{\beta}\to\mathsf{F}_{\alpha}$ be the restriction to $\mathsf{F}_{\beta}$ of the orthogonal projection $\mathcal{H}\to\mathsf{F}_{\alpha}$ in the prequantum Hilbert space $\mathcal{H}$. Then 
\begin{description}
\item[(i)] $\hat{X}\circ\pi_{\alpha\beta} = \pi_{\alpha\beta}\circ\hat{X}, \,\, \forall X\in V$,
\item[(ii)] The rescaled projection $\Delta_{\alpha\beta}\pi_{\alpha\beta}:\mathsf{F}_{\beta}\to\mathsf{F}_{\alpha} \text{ is unitary}$,
\end{description}
where
\begin{equation}\nonumber
 \Delta_{\alpha\beta} = \sqrt[4]{\det\frac{1}{2}(J_{\alpha}+J_{\beta})}.
\end{equation}
\end{projectj}

\begin{proof}
The first assertion follows simply from the fact that, as we have seen, each $X\in V$ generates a flow which preserves both $J_{\alpha(\beta)}$, so that $\hat{X}$ acts on both subspaces $\mathsf{F}_{\alpha(\beta)}$. Since it is a symmetric operator on the whole of $\mathcal{H}$, $\langle \hat{X}s, s'\rangle = \langle s, \hat{X}s'\rangle $ for arbitrary $s\in\mathsf{F}_{\alpha}$ and $s'\in\mathsf{F}_{\beta}$, so that $\hat{X}\circ\pi_{\alpha\beta} = \pi_{\alpha\beta}\circ\hat{X}$.

For the second statement, remember that every element of $\mathsf{F}_{\alpha}$ is a section of $B\to V$ of the form $\phi e^{-K_{\alpha}/2\hbar}s$, where $\phi$ is holomorphic with respect to $J_{\alpha}$, $K_{\alpha}$ is the corresponding Kahler scalar, and $Ds = -\frac{i}{\hbar}\theta_{0}\otimes s$. We will not introduce holomorphic coordinates with respect to one of the complex structures since we have an interest in working with all of them interchangeably. Thus the holomorphicity of $\phi$ with respect to $J_{\alpha}$ should mean to us that $(J_{\alpha}X-iX)\lrcorner d\phi=0, \,\forall X\in V$. Indeed,
\begin{equation}\nonumber
J_{\alpha}(J_{\alpha}X-iX) = -X-iJ_{\alpha}X = -i(J_{\alpha}X-iX),
\end{equation}
so these vectors are all antiholomorphic with respect to $J_{\alpha}$. It is not difficult to see that all antiholomorphic vector fields are of this form. Likewise, we should write $K_{\alpha}(X) = \omega(X,J_{\alpha}X)$. As a check, note that this gives the standard formula in coordinates holomorphic with respect to $J_{\alpha}$:
\begin{equation}\nonumber
\begin{split}
\omega(X,JX) &= \frac{1}{2}JX\lrcorner \left[(z^{a}\partial_{z^{a}}+\bar{z}^{b}\partial_{\bar{z}^{b}})\lrcorner\frac{i}{2}dz^{c}\wedge d\bar{z}^{c}\right] = \frac{i}{4}(iz^{a}\partial_{z^{a}}-i\bar{z}^{b}\partial_{\bar{z}^{b}})\lrcorner (z^{c}d\bar{z}_{c} - \bar{z}^{d}dz_{d}) \\
&= \frac{1}{2}g_{ab}z^{a}\bar{z}^{b} = K(X).
\end{split}
\end{equation}
Furthermore, the vacuum state in an arbitrary $\mathsf{F}_{\alpha}$ is $\psi_{0,\alpha}s = e^{-K_{\alpha}/2\hbar}s$ and the coherent states $\psi_{X,\alpha}s$ are
\begin{equation}\nonumber
\begin{split}
\psi_{X,\alpha}(Z) &= (-\hat{X})e^{-\omega(Z,J_{\alpha}Z)/2\hbar} = e^{-\omega(Z-X,J_{\alpha}(Z-X))/2\hbar}e^{i\omega(-X,Z)/\hbar} \\
&= \exp\left\{\frac{1}{2\hbar}[2\omega(X,(J_{\alpha}+i)Z)-\omega(X,J_{\alpha}X)-\omega(Z,J_{\alpha}Z)]\right\}.
\end{split}
\end{equation}
We first look at how these states project, as they are related to the projection map itself: remember that equation (\ref{projincoherent}) implies that the projection of a general (normalised) $\psi s\in \mathcal{H}$ in $\mathsf{F}_{\alpha}$ is
\begin{equation}\nonumber
\left(\int_{V}\langle(\psi_{W,\alpha}s),(\psi s)\rangle\psi_{W,\alpha}dW\right)s.
\end{equation}

A lengthy but otherwise straightforward calculation gives

\begin{projectcoherent}
The following formulas hold:
\begin{equation}\nonumber
\begin{cases}
\pi_{\alpha\beta}(\psi_{0,\beta}s) = \Delta_{\alpha\beta}^{-2}\Phi_{\alpha\beta}e^{-K_{\alpha}/2\hbar}s \\
\langle\pi_{\alpha\beta}(\psi_{W,\beta}s),\pi_{\alpha\beta}(\psi_{0,\beta}s)\rangle = \Delta_{\alpha\beta}^{-2}e^{-K_{\beta}(W)/2\hbar}, \,\, \forall W\in V
\end{cases},
\end{equation}
where
\begin{equation}\nonumber
\Phi_{\alpha\beta}(X) = \exp\left[\frac{1}{2\hbar}\omega(X,J_{\alpha}L_{\alpha\beta}X-iL_{\alpha\beta}X)\right], \,\,\,\, L_{\alpha\beta}=(J_{\alpha}+J_{\beta})^{-1}(J_{\alpha}-J_{\beta}).
\end{equation}
\end{projectcoherent}

Then one can use this to consider the projection of two arbitrary coherent states:
\begin{equation}\nonumber
\begin{split}
\langle\pi_{\alpha\beta}(\psi_{X,\beta}s),\pi_{\alpha\beta}(\psi_{Y,\beta}s)\rangle &= \langle\pi_{\alpha\beta}(\psi_{X,\beta}s),(-\hat{Y})\pi_{\alpha\beta}(\psi_{0,\beta}s)\rangle \\
&= \langle\pi_{\alpha\beta}\hat{Y}(-\hat{X})(\psi_{0,\beta}s),\pi_{\alpha\beta}(\psi_{0,\beta}s)\rangle \\
&= \langle\pi_{\alpha\beta}(\psi_{X-Y,\beta}s),\pi_{\alpha\beta}(\psi_{0,\beta}s)\rangle e^{-i\omega(X,Y)/\hbar} \\
&= \Delta_{\alpha\beta}^{-2}\exp\left[-\frac{1}{2\hbar}(2i\omega(X,Y)+K_{\beta}(X-Y))\right] = \Delta_{\alpha\beta}^{-2}\psi_{Y,\beta}(X) \\
&= \Delta_{\alpha\beta}^{-2}\langle(\psi_{X,\beta}s),(\psi_{Y,\beta}s)\rangle, \,\,\, \forall X,Y\in V.
\end{split}
\end{equation}
But $\mathsf{F}_{\beta} = \text{span}\{\psi_{X,\beta}s|X\in V\}$, so this extends linearly to $\langle\pi_{\alpha\beta}(\psi_{\beta}s),\pi_{\alpha\beta}(\psi_{\beta}s)\rangle = \Delta_{\alpha\beta}^{-2}\langle(\psi_{\beta}s),(\psi_{\beta}s)\rangle$ for any state $\psi_{\beta}s\in\mathsf{F}_{\beta}$. Finally, this projects to all of $\mathsf{F}_{\alpha}$: suppose $\psi_{\alpha}s\in(\pi_{\alpha\beta}\mathsf{F}_{\beta})^{\perp}\subset\mathsf{F}_{\alpha}$. Then $\langle(\psi_{\beta}s),(\psi_{\alpha}s)\rangle=0$ for any $\psi_{\beta}s\in\mathsf{F}_{\beta}$. But this implies that $\pi_{\beta\alpha}(\psi_{\alpha}s)=0$, where $\pi_{\beta\alpha}:\mathsf{F}_{\alpha}\to\mathsf{F}_{\beta}$ is also given by the orthogonal projection. In turn, this gives
\begin{equation}\nonumber
0 = \langle\pi_{\beta\alpha}(\psi_{\alpha}s),\pi_{\beta\alpha}(\psi_{\alpha}s)\rangle = \Delta_{\beta\alpha}^{-2}\langle(\psi_{\alpha}s),(\psi_{\alpha}s)\rangle \,\, \Rightarrow \,\, \psi = 0.
\end{equation}
We conclude that $(\pi_{\alpha\beta}\mathsf{F}_{\beta})^{\perp}=0\Rightarrow\mathsf{F}_{\alpha} = \pi_{\alpha\beta}(\mathsf{F}_{\beta})$. Hence the rescaled map $\Delta_{\alpha\beta}\pi_{\alpha\beta}$ is unitary.

\end{proof}

The rescaled projection provides then a unitary intertwiner between the two representations.
Since we shall use these projection maps to construct the metaplectic representation, it is important to know how they compose.

\begin{composepi}\label{compi}
Let $J_{1}$, $J_{2}$ and $J_{3}$ be positive complex structures on $(V,\omega)$ and $\mathsf{F}_{1}$, $\mathsf{F}_{2}$, $\mathsf{F}_{3} \subset \mathcal{H}$ the corresponding Fock Spaces. Then
\begin{equation}\nonumber
\Delta_{12}\Delta_{23}\pi_{32}\circ\pi_{21} = \tau_{123}\Delta_{13}\pi_{31},
\end{equation}
where $\pi_{\alpha\beta}$ is the projection $\mathsf{F}_{\beta}\to\mathsf{F}_{\alpha}$,
\begin{equation}\nonumber
\Delta_{\alpha\beta} = \sqrt[4]{\det\frac{1}{2}(J_{\alpha}+J_{\beta})},
\end{equation}
and $\tau_{\alpha\beta\gamma}\in S^{1}$. Moreover, $\tau_{112}=1$, and $\tau$ is symmetric under even permutations of its subscripts, and goes to its inverse under odd permutations.
\end{composepi}
\begin{proof}

First we prove that if $U:\mathsf{F}_{\beta}\to\mathsf{F}_{\alpha}$ is unitary and $\hat{X}\circ U = U\circ\hat{X}$ for every $X\in V$, then $U$ is given by $\tau\Delta_{\beta\alpha}\pi_{\alpha\beta} :\mathsf{F}_{\beta}\to\mathsf{F}_{\alpha}$ for some $\tau\in S^{1}$.
To see this, let $U' = \Delta_{\beta\alpha} U^{-1}\circ\pi_{\alpha\beta}$. Then $U':\mathsf{F}_{\beta}\to\mathsf{F}_{\beta}$ is unitary and $\hat{X}\circ U = U\circ\hat{X}, \, \forall X\in V$. Now, the ground state $\psi_{0}$ in $\mathsf{F}_{\beta}$, seen as a function on $V$, is determined up to a constant factor by the relation
\begin{equation}\nonumber
(\hat{W}\psi_{0})(Z) = \psi_{0}(z)e^{-(2\bar{w}_{a}z^{a}+\bar{w}_{a}w^{a})/4\hbar}, \,\, \forall W\in V,
\end{equation}
where $z$ and $w$ are the complex coordinates of $Z,W\in V_{(J)}$, respectivelly. It follows that the constant multiples of $\psi_{0}$ are distinguished from the other elements of $\mathsf{F}_{\beta}$ by the property that $e^{\bar{w}_{a}w^{a}/4\hbar}\hat{W}\psi_{0}$ is an antiholomorphic function of the coordinates $w$. But $e^{\bar{w}_{a}w^{a}/4\hbar}\hat{W}U'\psi_{0} = U' e^{\bar{w}_{a}w^{a}/4\hbar}\hat{W}\psi_{0}$ also depends antiholomorphically on $w$, and therefore $U'\psi_{0}=\tau\psi_{0}$ for some constant $\tau\in\mathbb{C}$. Therefore $U'\psi_{W} = U'(-\hat{W})\psi_{0} = \tau\psi_{W}$, so that $U'$ acts as $\tau\text{id}$ on all the coherent states $\psi_{W}$. Since $\mathsf{F}_{\beta} = \text{span}\{\psi_{W}|W\in V\}$, $U' = \tau\text{id}$. Finally, the map is unitary, so $\tau\in S^{1}$. Hence $\Delta_{\beta\alpha} U^{-1}\circ\pi_{\alpha\beta} = \tau\text{id}\Rightarrow U=\tau^{-1}\Delta_{\beta\alpha}\pi_{\alpha\beta}$.

With this in mind, note that both $\Delta_{12}\Delta_{23}\pi_{32}\circ\pi_{21}:\mathsf{F}_{1}\to\mathsf{F}_{3}$ and $\Delta_{13}\pi_{31}:\mathsf{F}_{1}\to\mathsf{F}_{3}$ are unitary and commute with all $\hat{X}$, so there should be some $\tau_{123}\in S^{1}$ such that $\Delta_{12}\Delta_{23}\pi_{32}\circ\pi_{21} = \tau_{123}\Delta_{13}\pi_{31}$. To evaluate a particular value, let $\psi_{1}$ be the ground state in $\mathsf{F}_{1}$. Then\footnote{For brevity, we ommited one of the computational lemmas necessary here. We refer to \cite{woodhouse1997geometric} for the full proof.}
\begin{equation}\nonumber
\begin{split}
\langle\psi_{1},\pi_{12}\pi_{21}\psi_{1}\rangle &= \langle\psi_{1},\pi_{21}\psi_{1}\rangle = \int_{V_{(J)}}\bar{\psi}_{1}(W)(\pi_{21}\psi_{1})(W)\epsilon \\
&= \int_{V(J)}e^{-w'\bar{w}'/4\hbar}\Delta_{12}^{-2}e^{[(\lambda^{2}-1)\bar{w}'^{2} - (\lambda^{2}+1)w'\bar{w}']/4\hbar(\lambda^{2}+1)}\epsilon \\
&= \Delta_{12}^{-2}\int_{V_{(J)}}f(\bar{w})e^{-w'\bar{w}'/2\hbar}\epsilon = \Delta_{12}^{-2}f(0) = \Delta_{12}^{-2},
\end{split}
\end{equation}
that is, $\pi_{12}\pi_{21} = \Delta_{12}^{-2}\text{id} \,\Rightarrow\, \tau_{121}\text{id}=\tau_{121}\Delta_{11}\pi_{11} = \Delta_{12}^{2}\pi_{12}\pi_{21}=\text{id}$, so $\tau_{121}=1$. This then implies that $(\Delta_{12}\pi_{21})\circ (\Delta_{21}\pi_{12}) = \text{id}$, so $\Delta_{21}\pi_{12} = (\Delta_{12}\pi_{21})^{-1}$. This can be used to show the ciclicity properties of $\tau$:
\begin{equation}\nonumber
\begin{split}
\tau_{123}\text{id} &= \tau_{123}\Delta_{23}\pi_{32}(\Delta_{23}\pi_{32})^{-1} = (\tau_{123}\Delta_{13}\pi_{31})(\Delta_{23}\Delta_{31}\pi_{13}\pi_{32})(\Delta_{23}\pi_{32})^{-1} \\
&= (\Delta_{12}\Delta_{23}\pi_{32}\pi_{21})(\tau_{231}\Delta_{21}\pi_{12})(\Delta_{23}\pi_{32})^{-1} = (\tau_{231}\Delta_{23}\pi_{32})(\Delta_{23}\pi_{32})^{-1} \\
&= \tau_{231}\text{id},
\end{split}
\end{equation}
and, similarly, the behaviour of $\tau$ under odd permutations of its indices is seen from
\begin{equation}\nonumber
\tau_{213}\text{id} = (\Delta_{21}\Delta_{32}\pi_{12}\pi_{23})(\tau_{213}\Delta_{23}\pi_{32})(\Delta_{21}\pi_{12})^{-1} = (\tau_{123}\Delta_{13}\pi_{31})^{-1}(\Delta_{13}\pi_{31}) = \tau_{123}^{-1}\text{id}.
\end{equation}
\end{proof}

Now we are in a position to present the starting point for the construction of the Metaplectic representation. For short, let us denote the set of all the positive compatible complex structures on a symplectic vector space $(V,\omega)$ by $L^{+}V$.

\begin{metasymplectic}
Let $J_{0}\in L^{+}V$ and $\mathsf{F}_{0}$ the corresponding subspace of $\mathcal{H}$. Then $SP(V,\omega)\ni\rho\mapsto\tilde{\rho}$, where $\tilde{\rho} = \pi\circ\hat{\rho}$ with $\pi:\mathcal{H}\to\mathsf{F}_{0}$ the orthogonal projection, is a projective representation of $SP(V,\omega)$.
\end{metasymplectic}
\begin{proof}

Let us represent the states in $\mathsf{F}_{0}$ in the trivialization $s$ associated with the symplectic potential $\theta_{0} = \frac{1}{2}(p_{a}dq^{a}-q^{a}dp_{a})$ in canonical coordinates $(p_{a},q^{b})$. When presenting the construction of Fock space, we showed that this is invariant under $SP(V,\omega)$. Therefore, taking and arbitrary one-parameter family of symplectomorphisms $\rho_{t}\in SP(V,\omega)$ generated by some hamiltonian vector field $X$,
\begin{equation}\nonumber
0 = \mathcal{L}_{X}\theta_{0} = X\lrcorner\omega + d(X\lrcorner\theta_{0}),
\end{equation}
so that $X$ has hamiltonian $f=X\lrcorner\theta_{0}$. The action of $\hat{\rho}_{t}$ on the elements of $\mathcal{H}$ was understood in equation (\ref{flowsection}), from which it follows that, in this frame,
\begin{equation}\nonumber
\begin{split}
[\hat{\rho}_{t}(\psi s)](m) &=\psi(\rho_{t}m)\exp\left(-\frac{i}{\hbar}\int_{0}^{t}[(X_{f}\lrcorner\theta_{0} - f)(\rho_{t'}m)]dt'\right)s(m)\\
&=[(\psi\circ\rho_{t})s](m),
\end{split}
\end{equation}
so we can alternativelly look at the action of $SP(V,\omega)$ on $\mathsf{F}_{0}\subset\mathcal{H}$ as a change in coordinates on the base symplectic manifold. Furthermore, the change in coordinates leads to a change in the complex structure by $J\mapsto \rho_{t}^{-1}J\rho_{t}$. We represent this schematically in the commuting diagram
\begin{equation}\nonumber
\begin{tikzcd}
L^{+}V\ni J_{\beta} \arrow[d] \arrow[r] & J_{\alpha} = \rho^{-1}J\rho \arrow[d] & \\
\mathsf{F}_{\beta} \arrow[r, "\hat{\rho}"] & \mathsf{F}_{\alpha} = \{\psi\circ\rho |\psi\in\mathsf{F}_{\beta}\}
\end{tikzcd},
\end{equation}
where the vertical lines represent the quantization procedure. To find the composition law for the action $\tilde{\rho}$, on $\mathsf{F}_{0}$, consider
\begin{equation}\nonumber
\begin{tikzcd}
\mathsf{F}_{0} \arrow[r, "\hat{\rho}_{1}"] \arrow[d, "\tilde{\rho}_{1}"] & \mathsf{F}_{1} \arrow[dl, "\pi_{01}"] \arrow[r, "\hat{\rho}_{2}"] & \mathsf{F}_{3} \arrow[dl, "\pi_{23}"] \\
\mathsf{F}_{0} \arrow[d, "\tilde{\rho}_{2}"] \arrow[r, "\hat{\rho}_{2}"] & \mathsf{F}_{2} \arrow[dl, "\pi_{02}"] & \\
\mathsf{F}_{0} &&
\end{tikzcd}.
\end{equation}
Here, $\mathsf{F}_{\alpha}$ is constructed from $J_{\alpha} = \rho_{\alpha}^{-1}J_{0}\rho_{\alpha}$ and $\rho_{3} = \rho_{2}\circ\rho_{1}$. The diagram commutes: the two triangular subdiagrams by the definition of $\tilde{\rho}_{\alpha} = \pi_{0\alpha}\circ\hat{\rho}_{\alpha}$ and the rombus shaped because $\hat{\rho}_{2}$ is a unitary transformation on $\mathcal{H}$ and so it commutes with the orthogonal projection. This implies, by proposition \ref{compi}, that
\begin{equation}\nonumber
\tilde{\rho}_{2}\circ\tilde{\rho}_{1} = \pi_{02}\circ\pi_{23}\circ\hat{\rho}_{2}\circ\hat{\rho}_{1} = \chi_{320}\pi_{03}\circ\hat{\rho}_{3} = \chi_{320}\tilde{\rho}_{3},
\end{equation}
where
\begin{equation}\nonumber
\chi_{320} = \frac{\tau_{320}\Delta_{03}}{\Delta_{32}\Delta_{20}},
\end{equation}
so $\rho\mapsto\tilde{\rho}$ gives a representation of $SP(V,\omega)$ on $\mathsf{F}_{0}$ up to a constant factor.
\end{proof}

The goal of this subsection is to make this into a representation. From the previous proposition we see that one can approach this problem by asking how the Fock Space $\mathsf{F}_{\alpha}$ changes as one moves in the space of complex structures $L^{+}V$. The first step is to understand better the geometry of $L^{+}V$.

\begin{lagrangiangrassmanian}
Given a symplectic vector space $(V,\omega)$, the space $L^{+}V$ of positive complex structures of $V$ compatible with $\omega$, called the \emph{Positive Lagrangian Grassmanian}, is a K\"ahler manifold with K\"ahler scalar 
\begin{equation}\nonumber
K = -\ln\det (y),
\end{equation}
where, if $\{X^{a},Y_{b}\}$ is a symplectic frame in $V$ and $P_{J} = \text{span}\{X^{a}-z^{ab}Y_{b}\}\subset V_{\mathbb{C}}$, for some symmetric matrix $z$, is the Lagrangian subspace determined by $J$, then $y = \text{Im} (z)$.
\end{lagrangiangrassmanian}

\begin{proof}

Recalling the definitions from subsection \ref{subsecquant}, an element $J\in L^{+}V$ determines a positive K\"ahler Lagrangian subspace $P_{J}\subset V_{\mathbb{C}}$. Hence the map $J\mapsto P_{J}$ identifies $L^{+}V$ with a complex submanifold of $\text{Gr}(n,V_{\mathbb{C}})$, the Grassmanian of $n$-dimensional complex subspaces of $V_{\mathbb{C}}$ (thus the name of $L^{+}V$).

We shall use two different parametrizations of $L^{+}V$ interchangeably. The first is given by the matrix $(z^{ab})$ such that $P_{J} = \text{span}\{X^{a}-z^{ab}Y_{b}\}$ for a fixed symplectic frame on $(V,\omega)$. These are holomorphic coordinates because every $P_{J}$ is Kahler, and also symmetric because every $P_{J}$ is Lagrangian: for any $A, B\in P_{J}$,
\begin{equation}\nonumber
0=\omega(A,B) = \sum_{ab}A_{a}B_{b}\omega(X^{a}-z^{ac}Y_{c},X^{b}-z^{bd}Y_{d}) = \sum_{ab}A_{a}B_{b}(z^{ac}\delta^{b}_{c}-z^{bd}\delta^{a}_{d}).
\end{equation}

Alternatively, we may parametrize $L^{+}V$ by the entries of the matrix $(J_{ab})$ representing $J$ in the same symplectic frame. The coordinates $(z^{ab}) = x+iy$ and $(J_{ab})$ of a point $J\in L^{+}V$ are related by
\begin{equation}\label{nrelate}
(J_{ab}) = NJ_{0}N^{-1}, \,\,\,\, \text{ where }\,\,\,\, N = \left( \begin{array}{cc}
1 & 0 \\
-x & y
\end{array} \right) \,\,\,\, \text{ and } \,\,\,\, J_{0} = \left(\begin{array}{cc}
0 & -1 \\
1 & 0
\end{array} \right).
\end{equation}

We define a vector $T\in T_{J}(L^{+}V)$ to be a linear map $T:V\to V$ such that $J+tT$ is a positive compatible complex structure in $V$ to first order in the parameter $t$. This implies 
\begin{equation}\label{antit}
-1 = (J+tT)^{2} = J^{2}+t(JT+TJ)+O(t^{2}) \Rightarrow TJ+JT=0
\end{equation}
and
\begin{equation}\nonumber
\begin{split}
\omega(Y,X) = \omega[(J+tT)Y,(J+tT)X] = \omega(Y,X) + t[\omega(TY,JX)+\omega(JY,TX)] + O(t^{2})\\
\Rightarrow \omega(JTY,X) = \omega(J^{2}TY,JX) = -\omega(TY,JX) = \omega(JY,TX) = \omega(JTX,Y) \\
\Rightarrow \omega(JTY,X) = \omega(JTX,Y), \,\, \forall X,Y\in V.
\end{split}
\end{equation}
Conversely, we can think of $T$ as given by the vector tangent to the curve $z+tw$ in the $z$ coordinates and represent it by some symmetric $w = u+iv$. Then equation (\ref{nrelate}) implies that the two expressions for $T\in T_{J}(L^{+}V)$ are related by
\begin{equation}\nonumber
T = (WJ_{0} - JW)N^{-1}, \,\,\,\, \text{ where } \,\,\,\, W = \left( \begin{array}{cc}
0 & 0 \\
-u & v
\end{array}\right).
\end{equation}

There is a symplectic structure given by
\begin{equation}\nonumber
\Omega :T_{J}(L^{+}V)\times T_{J}(L^{+}V)\to\mathbb{R}:(T,T')\mapsto \frac{1}{8}\text{tr}(TJT').
\end{equation}
That this is antisymmetric follows from (\ref{antit}). Note also that it is invariant under the action $J\mapsto\rho^{-1}J\rho$ of $SP(V,\omega)$ due to the cyclic property of the trace. Closure will follow from the expression of $\omega$ in terms of a Kahler scalar $K$. To evaluate this, note that the complex structure at $T_{J}(L^{+}V)$ with respect to which the $z$ coordinates are holomorphic is given by $J$ itself. More precisely,
\begin{equation}\nonumber
TJ = [(WJ_{0}-JW)N^{-1}](NJ_{0}N^{-1}) = [(WJ_{0})J_{0} - J(WJ_{0})]N^{-1}.
\end{equation}
But
\begin{equation}\nonumber
WJ_{0} = \left(\begin{array}{cc}
0 & 0 \\
-u & v
\end{array}\right) \left(\begin{array}{cc}
0 & -1 \\
1 & 0
\end{array}\right) = \left(\begin{array}{cc}
0 & 0 \\
v & u
\end{array}\right),
\end{equation}
which corresponds to the vector $iw = -v + iu$. Hence the natural metric $g(\cdot , \cdot) =2\omega(\cdot, J\cdot)$ is given by
\begin{equation}\nonumber
\begin{split}
2\Omega(T,TJ) &= \frac{1}{4}\text{tr}(TJTJ) = \frac{1}{4}\text{tr}(T^{2}) \\
&= \frac{1}{4}\text{tr}[(WJ_{0}N^{-1})^{2} - (WJ_{0}N^{-1}JWN^{-1}) - (JWN^{-1}WJ_{0}N^{-1}) - (JWN^{-1}JWN^{-1})] \\
&= \frac{1}{2}\text{tr}[(WJ_{0}N^{-1})^{2}+(WN^{-1})^{2}] = \frac{1}{2}\text{tr}[(uy^{-1})^{2}+(vy^{-1})^{2}],
\end{split}
\end{equation}
and hence in the $z$ coordinates this metric is given by
\begin{equation}\nonumber
\frac{1}{2}\frac{\partial^{2}K}{\partial z^{ab}\partial\bar{z}^{cd}}w^{ab}\bar{w}'^{cd} = -\frac{1}{8}\frac{\partial^{2}(\ln\det(y))}{\partial y^{ab}\partial y^{cd}}(u^{ab}u^{cd}+v^{ab}v^{cd}) = \frac{1}{2}\text{tr}[(uy^{-1})^{2}+(vy^{-1})^{2}]
\end{equation}
for $K = -\ln\det (y)$. We conclude that $\Omega = i\partial\bar{\partial}K$, which is then obviously closed.
\end{proof}

A geometric way of proceeding is to consider the hermitian vector bundle $F\to L^{+}V:(J,\mathsf{F}_{J})\mapsto J$. Because $\mathsf{F}_{J}\subset\mathcal{H}, \,\,\forall J\in L^{+}V$, it can be embedded in the trivial bundle $L^{+}V\times\mathcal{H}$, from which it inherits a natural connection, parallel transport from $J_{\alpha}$ to $J_{\beta}$ being given, to first order in $J_{\alpha}-J_{\beta}$ by the projection $\mathsf{F}_{\alpha}\to\mathsf{F}_{\beta}$. This is seen to have a nonzero curvature, and the modification of the definition of the Fock spaces which makes this bundle flat renders a representation of (the double cover of) $SP(V,\omega)$ on the space of covariantly constant sections.

The main strategy is to calculate the curvature from the cocycle $\chi$, which is expressed in terms of $\tau$ and $\Delta$. To fill in the technical steps, let $\psi_{\alpha}$ be the ground state in $\mathsf{F}_{\alpha}$. Then
\begin{equation}\nonumber
\langle\psi_{\alpha},\psi_{\beta}\rangle = \langle\psi_{\alpha},\pi_{\alpha\beta}\psi_{\beta}\rangle = (\Delta_{\alpha\beta})^{-2}\Phi_{\alpha\beta}(0)e^{-K(0)/2\hbar} = (\Delta_{\alpha\beta})^{-2}.
\end{equation}
On the other hand,
\begin{equation}\nonumber
\begin{split}
\langle\psi_{\alpha},\psi_{\beta}\rangle &= \langle\psi_{\alpha},\pi_{\alpha\beta}\psi_{\beta}\rangle = \frac{1}{\chi_{\alpha 0\beta}}\langle\psi_{\alpha},\pi_{\alpha 0}\pi_{0\beta}\psi_{\beta}\rangle = \frac{\Delta_{0\alpha}\Delta_{0\beta}}{\Delta_{\alpha\beta}\tau_{0\alpha\beta}}\langle\pi_{0\alpha}\psi_{\alpha},\pi_{0\beta}\psi_{\beta}\rangle \\
&= \frac{\Delta_{0\alpha}\Delta_{0\beta}}{\Delta_{\alpha\beta}\tau_{0\alpha\beta}}\left(\frac{1}{\Delta_{0\alpha}\Delta_{0\beta}}\right)^{2}\int_{V}\exp\left[\frac{1}{2\hbar}\omega(X,J_{\alpha}L_{\alpha}X - iL_{\alpha}X)\right]^{*}\times\\
&\,\,\,\,\,\,\,\,\,\,\,\,\,\,\,\,\,\,\,\,\,\,\,\,\,\,\,\,\,\,\,\,\,\,\,\,\,\,\,\,\,\,\,\,\,\,\,\,\,\,\,\,\,\,\,\,\,\,\,\,\,\,\,\,\,\,\,\,\,\,\,\,\,\,\,\,\,\,\,\,\times\exp\left[\frac{1}{2\hbar}\omega(X,J_{\beta}L_{\beta}X - iL_{\beta}X)\right]e^{-K/\hbar}\epsilon,
\end{split}
\end{equation}
therefore, from the two equations,
\begin{equation}\nonumber
(\Delta_{\alpha\beta})^{-2} = \frac{1}{\Delta_{\alpha\beta}\Delta_{0\alpha}\Delta_{0\beta}\tau_{0\alpha\beta}}\int_{V}e^{-Q_{\alpha\beta}/\hbar}\epsilon ,
\end{equation}
where $Q_{\alpha\beta}(X) = \omega(X,A_{\alpha\beta}X)$, with
\begin{equation}\nonumber
A_{\alpha\beta} = J_{0} - \frac{1}{2}J_{0}(L_{\alpha}(1+iJ_{0})+L_{\beta}(1-iJ_{0})).
\end{equation}
Since $Q$ is a quadratic form with positive real part, the integral converges and we get
\begin{equation}\nonumber
\tau_{0\alpha\beta} = \left(\frac{\Delta_{\alpha\beta}}{\Delta_{0\alpha}\Delta_{0\beta}}\right)\left(\frac{1}{\det (A_{\alpha\beta})}\right)^{\frac{1}{2}}.
\end{equation}

To calculate the determinant of $(A_{\alpha\beta})$, consider the frame $Z^{a},\bar{Z}^{a}$, where $Z^{a} = X^{a} - iY_{a}$. Then $J_{0}Z^{a} = iZ^{a}$ and $J_{0}\bar{Z}^{a} = -i\bar{Z}^{a}$ so that the expression for $J_{0}$ in this frame is
\begin{equation}\nonumber
J_{0} = \left( \begin{array}{cc}
i & 0 \\
0 & -i
\end{array}\right).
\end{equation}
Likewise, the matrix $N$ becomes
\begin{equation}\nonumber
\begin{cases}
NX^{a} = X^{a}-xY_{a}\\
NY_{a} = yY_{a}
\end{cases}
\Rightarrow
\begin{cases}
NZ^{a} = \frac{1}{2}(1-iz)Z^{a}+\frac{1}{2}(1+iz)\bar{Z}^{a}\\
N\bar{Z}^{a} = \frac{1}{2}(1-i\bar{z})Z^{a} + \frac{1}{2}(1+i\bar{z})\bar{Z}^{a}
\end{cases}
\Rightarrow
N(z) = \frac{1}{2}\left(\begin{array}{cc}
1-iz & 1-i\bar{z}\\
1+iz & 1+i\bar{z}
\end{array}\right),
\end{equation}
so that, using equation (\ref{nrelate}),
\begin{equation}\nonumber
\begin{split}
L_{\alpha} = L(z_{\alpha}) = \left(\begin{array}{cc}
0 & i+\bar{z}_{\alpha} \\
i-z_{\alpha} & 0
\end{array}\right)\left(\begin{array}{cc}
i+z_{\alpha} & 0\\
0 & i-\bar{z}_{\alpha}
\end{array}\right)^{-1} \Rightarrow \\
\Rightarrow A_{\alpha\beta} =\left(\begin{array}{cc}
-1+iz_{\beta}&-i-\bar{z}_{\alpha}\\
-i+z_{\beta}&1+i\bar{z}_{\alpha}
\end{array}\right)\left(\begin{array}{cc}
i+z_{\beta}&0\\
0&i-\bar{z}_{\alpha}
\end{array}\right)^{-1} \Rightarrow \\
\Rightarrow \det(A_{\alpha\beta}) = \frac{2i(z_{\beta}-\bar{z}_{\alpha})}{\det(i+z_{\beta})\det(i-\bar{z}_{\alpha})} = \frac{\zeta_{\alpha\beta}}{\zeta_{0\beta}\zeta_{\alpha 0}}, \,\,\,\, \text{ where } \,\,\,\, \zeta_{\alpha\beta} = \det\frac{i}{2}(\bar{z}_{\alpha}-z_{\beta}).
\end{split}
\end{equation}
We additionally use the properties of $\tau$ and $\Delta$. For example, we know from proposition \ref{compi} that $\tau_{0\alpha\alpha}=1$, so the previous equations give
\begin{equation}\nonumber
1=(\tau_{0\alpha\alpha})^{4} = \left(\frac{\Delta_{\alpha\alpha}}{\Delta_{0\alpha}\Delta_{0\alpha}}\right)^{4}\left(\frac{\zeta_{0\alpha}\zeta_{\alpha 0}}{\zeta_{\alpha\beta}}\right)^{2} \Rightarrow (\Delta_{0\alpha})^{4} = (\Delta_{\alpha 0})^{4} = \frac{\zeta_{0\alpha}\zeta_{\alpha 0}}{\zeta_{\alpha\alpha}\zeta_{00}}.
\end{equation}
This, in turn, implies
\begin{equation}\nonumber
1 = |\tau_{0\alpha\beta}|^{4} = \Big|\frac{\Delta_{\alpha\beta}}{\Delta_{\alpha 0}\Delta_{0 \beta}}\Big|^{4}\Big|\frac{\zeta_{0\beta}\zeta_{\alpha 0}}{\zeta_{\alpha\beta}}\Big|^{2} \Rightarrow (\Delta_{\alpha\beta})^{4} = \frac{\zeta_{\alpha\beta}\zeta_{\beta\alpha}}{\zeta_{\alpha\alpha}\zeta_{\beta\beta}},
\end{equation}
since $\zeta_{\alpha\beta} = \zeta_{\beta\alpha}^{*}$. Now, we are finally able to calculate
\begin{equation}\nonumber
(\tau_{0\alpha\beta})^{4} = \left(\frac{\Delta_{\alpha\beta}}{\Delta_{0\alpha}\Delta_{0\beta}}\right)^{4}\left(\frac{\zeta_{0\beta}\zeta_{\alpha 0}}{\zeta_{\alpha\beta}}\right)^{2} = \frac{\zeta_{\alpha 0}\zeta_{\beta\alpha}\zeta_{0\beta}}{\zeta_{0\alpha}\zeta_{\alpha\beta}\zeta_{\beta 0}}.
\end{equation}
Since $\tau$ is a cocycle, ie., $\tau_{123}\tau_{301} = \tau_{230}\tau_{012}$, this allows to calculate a general $\tau_{\alpha\beta\gamma}$,
\begin{equation}\nonumber
(\tau_{\alpha\beta\gamma})^{4} = \frac{\zeta_{\beta\alpha}\zeta_{\gamma\beta}\zeta_{\alpha\gamma}}{\zeta_{\alpha\beta}\zeta_{\beta\gamma}\zeta_{\gamma\alpha}},
\end{equation}
and, finally,
\begin{equation}\nonumber
\chi_{123} = \frac{\tau_{123}\Delta_{31}}{\Delta_{12}\Delta_{23}} = \left( \frac{\zeta_{22} \zeta_{13}}{\zeta_{12}\zeta_{23}}\right)^{\frac{1}{2}},
\end{equation}
where we choose the square-root by: $\chi_{111} = 1$, $\chi(z_{1},z_{2},z_{3}) = \chi_{123}$ is continuous in $z_{1,2,3}$. This is well-defined because $L^{+}V$ is connected and simply-connected. Therefore the curvature $\Gamma$ of the bundle $(J,\mathsf{F}_{J})\mapsto J$ is
\begin{equation}\nonumber
\Gamma (z_{1}) = id_{3}\wedge d_{2}\ln\chi_{123}|_{z_{1}=z_{2}=z_{3}} = id_{3}\wedge d_{2}\ln(\zeta_{23})^{-\frac{1}{2}}|_{z_{1}=z_{2}=z_{3}} = \frac{1}{2}\left[-i\partial\bar{\partial}\ln\det\frac{i}{2}(\bar{z}-z)\right]\Bigg|_{z=z_{1}},
\end{equation}
so the curvature is $\frac{1}{2}\Omega$, where $\Omega$ is the symplectic structure in $L^{+}V$.

We wish to progress by tensoring $F\to L^{+}V$ with a bundle with connection with curvature $-\frac{1}{2}\Omega$. Additionally, this should be a line-bundle, since one should not effectively change the Fock spaces $\mathsf{F}_{J}$, but only the relative phases between them. A prequantum bundle is then seen to be naturally of use in the construction. Consider
\begin{equation}\nonumber
\pi :L=\{(J,\xi),J\in L^{+}V, \xi\in L_{J}\}\to L^{+}V,
\end{equation}
where $L_{J}:=\{\xi\in\wedge^{n}V_{\mathbb{C}}|X\wedge\xi=0,\forall X\in P_{J}\}$. These are one-dimensional since $\dim(P_{J})=n\Rightarrow\dim(\wedge^{n}P_{J})=\binom{n}{n}=1$. Using the holomorphic coordinates $z^{ab}$ in $P_{J}=\text{span}\{X^{a}-z^{ab}Y_{b}\}$ one can give it the structure of a holomorphic line bundle.

This bundle has a hermitian structure: in a given symplectic frame $(X^{a},Y_{b})$ in $(V,\omega)$, it is expressed as
\begin{equation}\nonumber
(\xi,\xi ')\Xi = i^{n}\bar{\xi}\wedge\xi ' , \,\,\,\, \text{ where } \,\,\,\, \Xi = \left(\overset{n}{\underset{i=1}{\wedge}}X^{i}\right)\wedge\left(\overset{n}{\underset{j=1}{\wedge}}Y_{j}\right).
\end{equation}
This hermitian structure then defines a compatible connetion: let $Z^{a}=X^{a}-z^{ab}Y_{b}$ be a basis for $P_{J}$, where $z^{ab}$ are the holomorphic coordinates of $J\in L^{+}V$, then this gives a holomorphic ($\Rightarrow$ polarized) section of $L$, and
\begin{equation}\label{calculxi}
\begin{split}
(\xi,\xi)\Xi &= i^{n}\bar{\xi}\wedge\xi = i^{n}\left[\overset{n}{\underset{k=1}{\wedge}}(X^{k}-\bar{z}^{ka}Y_{a})\right]\wedge\left[\overset{n}{\underset{l=1}{\wedge}}(X^{l}-z^{lb}Y_{b})\right] \\
&= i^{n}\overset{n}{\underset{k=1}{\wedge}}\left[(X^{k}-\bar{z}^{ka}Y_{a})\wedge (X^{k}-z^{kb}Y_{b})\right](-1)^{\frac{n(n-1)}{2}} \\
&= i^{n}\overset{n}{\underset{k=1}{\wedge}}\left[X^{k}\wedge(\bar{z}^{ka}-z^{ka})Y_{a}\right](-1)^{\frac{n(n-1)}{2}} = i^{n}\left(\overset{n}{\underset{k=1}{\wedge}}X^{k}\right)\wedge\left[\overset{n}{\underset{l=1}{\wedge}}(\bar{z}^{la}-z^{la})Y_{a}\right] \\
&= \left(\overset{n}{\underset{k=1}{\wedge}}X^{k}\right)\wedge i^{n}\det(\bar{z}-z)\left(\overset{n}{\underset{l=1}{\wedge}}Y_{l}\right) \\
&= \det(2y)\Xi,
\end{split}
\end{equation}
$z=x+iy$. Then the Kahler scalar can be calculated to be $K=-\ln(\xi,\xi) = -\ln\det(2y)\Rightarrow i\partial\bar{\partial}K = \Omega$. So, up to factors of $\hbar$, this is a prequantum bundle. 

We are interested in a bundle with the opposite curvature and so let us consider now the dual of this bundle, the \emph{canonical bundle} $K\to L^{+}V$, whose fibre over $J$ is $K_{J}=K_{P_{J}}$, where
\begin{equation}\nonumber
K_{P} = \{\mu\in\wedge^{n}V_{\mathbb{C}}^{*}|X\lrcorner\mu = 0, \forall X\in \bar{P}\},
\end{equation}
for any Lagrangian subspace $P\subset V_{\mathbb{C}}$. We calculate its curvature from the cocycle of the embedding of $K$ in the trivial bundle $L^{+}V\times\wedge^{n}V_{\mathbb{C}}^{*}$.

Similarly, the $L^{+}V\times\wedge^{n}V_{\mathbb{C}}^{*}$ has an indefinite inner product given by $i^{n}(\mu,\mu ') \epsilon = \bar{\mu}\wedge\mu '$, where $\epsilon = \omega^{n}/(2\pi\hbar)^{n}$. Because the Lagrangian subspaces $P_{J}$ are positive, the restriction of this to $K$ is positive definite. Take the holomorphic section
\begin{equation}\nonumber
\mu(J) = \left(\frac{1}{4\pi\hbar}\right)^{\frac{n}{2}}(\bar{Z}^{1}\lrcorner\omega)\wedge ...\wedge (\bar{Z}^{n}\lrcorner\omega),
\end{equation}
where the $Z^{a}$ span $P_{J}$. A calculation totally analogous to (\ref{calculxi}) gives $(\mu_{1},\mu_{2}) = \det\frac{i}{2}(\bar{z}_{2}-z_{1}) = \zeta_{21}$, where $\mu_{i}=\mu(J_{i})$ and $z_{i}$ are the holomorphic coordinates of $J_{i}\in L^{+}V$. Since the connection obtained from the embedding is such that parallel transport from $J_{\alpha}$ to $J_{\beta}$ is given, to first order in $J_{\alpha}-J_{\beta}$, by the orthogonal projection $K_{J_{\alpha}}\to K_{J_{\beta}}$ with respect to this inner product, the cocycle of the embedding is found to be
\begin{equation}\nonumber
\frac{(\mu_{2},\mu_{1})(\mu_{3},\mu_{2})}{(\mu_{3},\mu_{1})(\mu_{2},\mu_{2})} = \frac{\zeta_{12}\zeta_{23}}{\zeta_{13}\zeta_{22}} = (\chi_{123})^{-2},
\end{equation}
so that the curvature is
\begin{equation}\nonumber
id_{3}\wedge d_{2}\ln\chi_{123}^{-2}|_{z_{1}=z_{2}=z_{3}} = -2\left(\frac{1}{2}\Omega\right) = -\Omega.
\end{equation}

\begin{halform}
The half-form bundle $\delta\to L^{+}V$ is the line bundle $\sqrt{K}$. It has a connection and compatible Hermitian structure inherited from $K$. Likewise, we define the pairing,
\begin{equation}\nonumber
(\nu_{\alpha},\nu_{\beta}):=\sqrt{(\nu_{\alpha}^{2},\nu_{\beta}^{2})}, \,\,\,\, \forall \nu_{\alpha}\in\delta_{J_{\alpha}}, \nu_{\beta}\in\delta_{J_{\beta}},
\end{equation}
by using the pairing in $K$. Here, the sign of the square-root is fixed by continuity together with $(\nu,\nu)\geq 0$. A half-form on $P_{J}$ is an element of $\delta_{J}$.
\end{halform}

The definition of the half-form pairing implies that the corresponding cocycle is $(\chi_{123})^{-1}$, so that the curvature of the half-form bundle is $-\frac{1}{2}\Omega$, as we needed. The construction is heavily dependent on the fact that $K$ (and therefore $\delta$) are topologically trivial and that $L^{+}V$ is simply connected.

Therefore we should substitute the bundle $F\to L^{+}V$ by $\tilde{F}= F\otimes\delta$. The cocycle of the pairing 
\begin{equation}\nonumber
\langle s_{1}\otimes\nu_{1},s_{2}\otimes\nu_{2}\rangle = \langle s_{1},s_{2}\rangle (\nu_{1},\nu_{2}),
\end{equation}
where $s_{\alpha}\otimes\nu_{\alpha}\in\tilde{F}_{\alpha}$, is $(\chi_{123})(\chi_{123}^{-2})^{1/2}=1$, so the resulting bundle, with connection such that parallel transport is given by the projection $\tilde{\pi}$ defined by $\langle\tilde{s},\tilde{\pi}\tilde{s}'\rangle = \langle\tilde{s},\tilde{s}\rangle $, is flat. The action of the $\hat{X}, \,\, X\in V$ does not change because these preseve the complex structure $J$ on $V$, so that they act trivially on $K$, and hence on $\delta$. On the other hand, $SP(V,\omega)$ acts on $K$ by $\mu\mapsto\rho^{*}\mu$. This action preserves the pairing in $K$, 
\begin{equation}\nonumber
i^{n}(\rho^{*}\mu_{1},\rho^{*}\mu_{2})\epsilon = \rho^{*}\bar{\mu}_{1}\wedge\rho^{*}\mu_{2} = \rho^{*}(\bar{\mu}_{1}\wedge\mu_{2}) = i^{n}(\mu_{1},\mu_{2})\rho^{*}\epsilon = i^{n}(\mu_{1},\mu_{2})\epsilon,
\end{equation}
since $\epsilon \propto\omega^{n}$, which is preserved by the definition of $SP(V,\omega)$. This is consistent with the action of $SP(V,\omega)$ on the base $L^{+}V$, since the elements of $\rho^{*}K_{J}$ annihilate $\rho^{*}\bar{P}_{J}=\bar{P}_{\rho(J)}$, that is $\rho^{*}K_{J} = K_{\rho(J)}$. However, it is not possible to transfer this action to $\delta$, as we now motivate: consider the group $U_{J}=\{\rho\in SP(V,\omega)|\rho(J)=J\}$. It can be shown to be isomorphic to $U(n)$, and the relation is that, if $\rho\in U_{J}$ corresponds to $u_{\rho}\in U(n)$, then $\rho|_{K_{J}} = \det u_{\rho}\text{id}$. The natural way to proceed would be to define $\rho|_{\delta_{J}} = \sqrt{\det u_{\rho}}\text{id}$. The problem is that one cannot define $\sqrt{\det u}$ on all of $U(n)$ continuously. 

Conversely, the double cover $MP(V,\omega)$ of $SP(V,\omega)$ has a well-defined action on $\delta$. In fact, we can define this group through the way it acts on $\delta$: let $0\neq\mu\in C^{\infty}(K)$. For any $\rho\in SP(V,\omega)$, $\rho^{*}\mu=\lambda_{\rho}\mu$ for some $\lambda_{\rho}:L^{+}V\to\mathbb{C}$. Define $MP(V,\omega)=\{(\rho,\sqrt{\lambda_{\rho}})\}$ with composition rule $(\rho_{1},\sqrt{\lambda_{\rho_{1}}})\circ (\rho_{2},\sqrt{\lambda_{\rho_{2}}})=(\rho_{1}\circ\rho_{2},\sqrt{\lambda_{\rho_{1}}\lambda_{\rho_{2}}})$, where $\sqrt{\lambda_{\rho}}$ is one of the two square roots of $\lambda_{\rho}$. Then, for $\varrho = (\rho,\sqrt{\lambda_{\rho}})\in MP(V,\omega)$ and $\nu\in\delta$, define $\varrho^{*}\nu = \sqrt{\lambda_{\rho}}\nu$.

Finally, because the cocycle of the pairing is trivial, the projections $\tilde{\pi}:\tilde{F}_{\alpha}\to\tilde{F}_{\beta}$ are unitary and $\tilde{\pi}_{32}\circ\tilde{\pi}_{21}=\tilde{\pi}_{31}$, so that $\rho\mapsto\tilde{\pi}\circ(\hat{\rho}\otimes\rho^{*})$ gives a representation of $MP(V,\omega)$ on each $\mathsf{F}_{\alpha}$. All in all, we have

\begin{metadone}
The bundle $\tilde{F}=F\otimes\delta\to L^{+}V$, with 
\begin{equation}\nonumber
\langle s_{1}\otimes\nu_{1},s_{2}\otimes\nu_{2}\rangle = \langle s_{1},s_{2}\rangle (\nu_{1},\nu_{2}), \,\,\, \forall s_{\alpha}\otimes\nu_{\alpha}\in\tilde{F}_{\alpha}=\mathsf{F}_{\alpha}\otimes\delta_{J_{\alpha}},
\end{equation}
has a flat connection such that parallel transport from $J_{1}$ to $J_{2}$ is given by the orthogonal projection $\tilde{\pi}:\tilde{F}_{1}\to\tilde{\mathsf{F}}_{2}$. Each $\tilde{\mathsf{F}}_{\alpha}$ carries a representation $X\mapsto\tilde{X}, \,\, X\in V$ of the Heisenberg group, defined by $\tilde{X}(s\otimes\nu)=(\hat{X}s)\otimes\nu$, and the metaplectic representation of $MP(V,\omega)$ $\rho\mapsto\tilde{\rho}$, defined by $\tilde{\rho}(s\otimes\nu)=\tilde{\pi}(\hat{\rho}(s)\otimes\rho^{*}(\nu))$, with $\tilde{\pi}:\tilde{\mathsf{F}}_{\rho(J)}\to\tilde{\mathsf{F}}_{J}$ the orthogonal projection.
\end{metadone}

\begin{correctsho}\emph{Corrected SHO}

One does not have to go far to find an example of a physics application. In fact, it is the metaplectic correction which fixes the spectrum of the simple harmonic oscillator, discussed earlier in example \ref{shorr}. We take $V=\{(p,q)\in\mathbb{R}^{2}\}$ and $\omega=dp\wedge dq$. Then $SP(V,\omega)=SP(1,\mathbb{R})$. But it is well known that $SP(1,\mathbb{R})=SL(2,\mathbb{R})$ (in two dimensions, symplectic means `area-preserving'), so
\begin{equation}\nonumber
\text{Lie}(SP(V,\omega))=sl(2,\mathbb{R})=\text{span}\left\{A_{1}=\left(
\begin{array}{cc}
0&-1\\
1&0
\end{array}\right), A_{2}=\left(
\begin{array}{cc}
-1&0\\
0&1
\end{array}\right), A_{3}=\left(
\begin{array}{cc}
0&1\\
1&0
\end{array}
\right)\right\}.
\end{equation}
The vector fields generating the corresponding flows are then
\begin{equation}\nonumber
\begin{cases}
e^{tA_{1}}\left(\begin{array}{cc}
p\\
q
\end{array}\right) = \left(\begin{array}{cc}
p\\
q
\end{array}\right)+t\left(\begin{array}{cc}
-q\\
p
\end{array}\right)+O(t^{2})\\
e^{tA_{2}}\left(\begin{array}{cc}
p\\
q
\end{array}\right) = \left(\begin{array}{cc}
p\\
q
\end{array}\right)+t\left(\begin{array}{cc}
-p\\
q
\end{array}\right)+O(t^{2})\\e^{tA_{1}}\left(\begin{array}{cc}
p\\
q
\end{array}\right) = \left(\begin{array}{cc}
p\\
q
\end{array}\right)+t\left(\begin{array}{cc}
q\\
p
\end{array}\right)+O(t^{2})
\end{cases}\Rightarrow\begin{cases}
X_{1}=-q\frac{\partial}{\partial p}+p\frac{\partial}{\partial q}\\
X_{2}=-p\frac{\partial}{\partial p}+q\frac{\partial}{\partial q}\\
X_{3}=q\frac{\partial}{\partial p}+p\frac{\partial}{\partial q}
\end{cases}.
\end{equation}
And, by Hamilton's equation $df=-X\lrcorner\omega$, these are generated by the functions 
\begin{equation}\nonumber
f_{1}=\frac{1}{2}(p^{2}+q^{2}), \,\,\,\, f_{2}=pq, \,\,\,\, f_{3}=\frac{1}{2}(p^{2}-q^{2}).
\end{equation}
As we can see, choosing this basis of $\text{Lie}(SP(V,\omega))$ one immediately recognizes the Hamiltonian of the SHO as one of the generators of the action of the symplectic group on $(V,\omega)$.

As explained in example \ref{shorr}, the coordinate $z=p+iq$ is holomorphic with respect to $J_{0}$ and $\mathsf{F}_{0}=\{\phi_{z}e^{-z\bar{z}/4\hbar}s,\,\,\, \phi \text{ holomorphic}\}$, where $Ds=-\frac{i}{\hbar}\theta_{0}\otimes s$ for $\theta_{0}=\frac{1}{2}(pdq-qdp)$. The projective representation of $SP(V,\omega)$ on $\mathsf{F}_{0}$ is generated by $\tilde{f}_{i}=\pi\circ\hat{f}_{i}$. Since we saw that $f_{1}$ preserves the polarization $J_{0}$,
\begin{equation}\nonumber
\tilde{f}_{1} = \pi\circ\hat{f}_{1}=\text{id}\circ\hat{f}_{1}:\phi\mapsto\hbar z\frac{\partial\phi}{\partial z}.
\end{equation}

And since the half-form bundle has one-dimensional fibres, after choosing a holomorphic section $\nu$, which we take as 
\begin{equation}\nonumber
\nu=\sqrt{\mu}, \,\,\,\, \mu = \frac{1}{\sqrt{4\pi\hbar}}dz,
\end{equation}
the other polarized sections of $\delta$ are fixed to be of the form $\varphi(z)\nu$, with $\varphi(z)$ holomorphic. Therefore $\tilde{\mathsf{F}}_{0}=\{\phi(z)e^{-z\bar{z}/4\hbar} s\otimes\nu\}$. This space carries the metaplectic representation of $MP(V,\omega)$, and we wish to see how the generator $f_{1}$ acts in it. Note that
\begin{equation}\nonumber
\mathcal{L}_{X_{1}}dz = d(X_{1}\lrcorner dz) = d\left[\left(p\frac{\partial}{\partial q}-q\frac{\partial}{\partial p}\right)\lrcorner (dp+idq)\right]=idz
\end{equation}
and $(\mu,\mu)=1$ ($\Rightarrow (\nu,\nu)=1$), since
\begin{equation}\nonumber
i(\mu,\mu)\frac{dp\wedge dq}{2\pi\hbar}=\bar{\mu}\wedge\mu=\frac{1}{4\pi\hbar}d\bar{z}\wedge dz=\frac{idp\wedge dq}{2\pi\hbar}.
\end{equation}
Therefore $(\mu,\mathcal{L}_{X_{1}}\mu)=(\mu,i\mu)=i$ and thus
\begin{equation}\nonumber
(\nu,\mathcal{L}_{X_{1}}\nu)=\frac{1}{2}2(\nu,\nu)(\nu,\mathcal{L}_{X_{1}}\nu)=\frac{1}{2}(\nu^{2},\mathcal{L}_{X_{1}}\nu^{2}) = \frac{1}{2}(\mu,\mathcal{L}_{X_{1}}\mu)=\frac{i}{2},
\end{equation}
which implies $\mathcal{L}_{X_{1}}\nu=\frac{i}{2}\nu$. We use this to evaluate the correction to $\tilde{f}_{1}$ proposed by the metaplectic prescription. Analogously to equation (\ref{absf}) $\tilde{f}$ should generate the action $\tilde{\rho}(s\otimes\nu)=\hat{\rho}(s)\otimes\rho^{*}\nu$ by
\begin{equation}\nonumber
\frac{d\tilde{\rho}_{t}}{dt} = \frac{i}{\hbar}\tilde{\rho}_{t}\tilde{f},
\end{equation}
which gives
\begin{equation}\nonumber
\begin{split}
\tilde{f}_{1}(\psi s\otimes\nu)&=-i\hbar\frac{d}{dt}\tilde{\rho}_{1,t}(\psi s\otimes\nu)\Big|_{t=0}=-i\hbar\frac{d}{dt}\hat{\rho}_{1,t}(\psi s)\otimes\rho^{*}_{1,t}(\nu)\Big|_{t=0}\\
&=-i\hbar\frac{d}{dt}\left\{\left[(\psi s)+t\frac{i}{\hbar}\hat{f}_{1}(\psi s)+O(t^{2})\right]\otimes [\nu+t\mathcal{L}_{X_{1}}\nu +O(t^{2})]\right\}\Big|_{t=0} \\
&= -i\hbar\left[\frac{i}{\hbar}\hat{f}_{1}(\psi s)\otimes\nu+(\psi s)\otimes\mathcal{L}_{X_{1}}\nu\right] \\
&= -i\hbar\left[\frac{i}{\hbar}\left(\hbar z\frac{\partial\phi}{\partial z}e^{-z\bar{z}/4\hbar}\right)s\otimes\nu+\phi e^{-z\bar{z}/4\hbar}s\otimes\frac{i}{2}\nu\right] \\
&= \left[\hbar\left(z\frac{\partial}{\partial z}+\frac{1}{2}\right)\phi\right]e^{-z\bar{z}/4\hbar}s\otimes\nu,
\end{split}
\end{equation}
which agrees with the physical expectation. In one of the eigenspaces $\mathsf{H}_{n}$,
\begin{equation}\nonumber
\tilde{f}_{1}|_{\mathsf{H}_{n}} = \hbar \left(n+\frac{1}{2}\right)\text{id}.
\end{equation}
\end{correctsho}

\pagebreak

\subsection{Half-form quantization}

Many generalizations are needed to extend the metaplectic correction to the non-linear case. We comment on some of the results\cite{woodhouse1997geometric,hall2013quantum,tuynman2016metaplectic,lyris2021importance,wu2015projective}. First, the positive Lagrangian Grassmanian generalizes to the the non-negative Lagrangian Grassmanian $LV$ of a symplectic vector space $(V,\omega)$, made up of all non-negative Lagrangian subspaces of $V_{\mathbb{C}}$. We should try to repeat the construction of the last subsection on each tangent space $T_{m}M$ in a way that can be extended to all of $(M,\omega)$ in the presence of a polarization. This is called a metaplectic structure on a manifold.

\begin{metastruc}
Let $(M,\omega)$ be a $2n$-dimensional symplectic manifold and let $LM:=\{(m,P)|m\in M, P\subset (T_{m}M)_{\mathbb{C}} \text{ is a non-negative Lagrangian subspace}\}$. $LM$ has the structure of a bundle over $M$ with projection $\pi :(m,P)\mapsto m$ and fibre $L_{m}M=L(T_{m}M)$. The canonical bundle $K\to LM$ over it is the line bundle whose fibre over $(m,P)$ is $K_{P}$. A \emph{metaplectic structure} on $M$ is a square-root of $K$. That is, a line bundle $\delta\to LM$ such that $\delta^{2}=K$.
\end{metastruc}

This definition of a metaplectic structure and the choice of polarization are related by the following.

\begin{metapol}
A square-root $\delta_{P_{0}}$ of the canonical bundle $K_{P_{0}}$ of a non-negative polarization $P_{0}$ determines a metaplectic structure. Conversely, a metaplectic structure determines a square-root $\delta_{P}$ of $K_{P}$ in a natural way for any other non-negative polarization $P$.
\end{metapol}

To understand this, remember that the construction of a metaplectic structure is equivalent to taking square-roots of the transition functions of the canonical bundle in a way that they still satisfy the cocycle conditions ($\Rightarrow$ still define a line bundle). Since one can show that each $L_{m}M=L(T_{m}M)$ is contractible, one can take the transition functions of $K$ to be constant on each fibre $L_{m}M$ of $LM$ and thus $\delta$ is defined by a square-root of $K|_{\Sigma}$, where $\Sigma=\sigma(M)$ is the graph of a smooth section $\sigma$ of $LM$.

A section $\sigma:M\to LM$ is a complex distribution on $M$ made up of nonnegative Lagrangian subspaces, so $\sigma^{*}K$ is the canonical bundle of this distribution. If $\sigma$ is also integrable, then it is a polarization. Conversely, any nonnegative polarization $P$ is a section $\sigma:M\to LM:m\mapsto(m,P_{m})$, so one may take $\delta_{P}=\sigma^{*}(\delta|_{\sigma(M)})$.

Neither the existence nor uniqueness of a metaplectic structure are guaranteed in symplectic manifolds which admit a prequantum bundle and a polarization, so the existence of a metaplectic structure imposes an additional constraint for a given symplectic manifold to be quantizable. More technically, it should happen that the classical phase space is not only a symplectic manifold, but also a metaplectic manifold\cite{asymptotics}.

\begin{stringdual}\emph{The dilaton shift}

As an interesting example we discuss how the transformation of the dilaton field in string theory under a T-duality relates to the inclusion of half-forms in the quantization procedure\cite{alfonsi2021double}. Let us briefly summarize the canonical interpretation of t-duality \cite{alvarez1995introduction,belov2007t,berman2015duality}. Consider the non-linear sigma model constructed on the space of maps $X:\Sigma\to M$, where $\Sigma = S^{1}\times\mathbb{R}$ is the string world-sheet and $M$ is the target space-time, which comes with a metric $G$, a closed $3$-form $H$ and a scalar field $\Phi$ (the dilaton). The dynamics is given by the action
\begin{equation}\nonumber
S=\frac{1}{4\pi\alpha '}\int_{\Sigma}d^{2}\xi [\sqrt{h}h^{\mu\nu}G_{ij}\partial_{\mu}X^{i}\partial_{\nu}X^{j}+i\epsilon^{\mu\nu}B_{ij}\partial_{\mu}X^{i}\partial_{\nu}X^{j}+\alpha '\sqrt{h}R^{(2)}\Phi(X)],
\end{equation}
where $h$ and $R^{(2)}$ are the metric and scalar curvature on $\Sigma$, respectively, and $B$ is the gauge potential of $H$ (locally, $H=dB$, although $B$ might not be globally defined). T-duality refers to the fact that, if the target space-time is a torus fibration $\mathbb{T} \hookrightarrow M \to\tilde{M}$, so that it has abelian isometries generated by translations along the torus directions $\vartheta^{i}$, there is a different (T-dual) background such that the above procedure gives the same quantum field theory on the space of maps $X:\Sigma\to M$. Specifically, if we consider T-duality with respect to the isometry generated by $\partial/\partial\vartheta$, where $(X^{I})=(\vartheta,X^{\alpha})$ is a coordinate system adapted to the $S^{1}\subset\mathbb{T}^{n}$ action, the background fields $(G,B,\Phi)$, $(\tilde{G},\tilde{B},\tilde{\Phi})$ should be related by the Buscher rules
\begin{equation}\label{buscher}
\begin{split}
\tilde{G}_{00}&=\frac{1}{G_{00}}\\
\tilde{G}_{0\alpha} &= \frac{B_{0\alpha}}{G_{00}}, \,\,\, \tilde{B}_{0\alpha} = \frac{G_{0\alpha}}{G_{00}}\\
\tilde{G}_{\alpha\beta} &= G_{\alpha\beta}-\frac{G_{0\alpha}G_{0\beta}-B_{0\alpha}B_{0\beta}}{G_{00}}\\
\tilde{B}_{\alpha\beta} &=B_{\alpha\beta}-\frac{G_{0\alpha}B_{0\beta}-G_{0\beta}G_{0\alpha}}{G_{00}}.
\end{split}
\end{equation}
Additionally, if one requires the dual model to also have conformal invariance, it is necessary that
\begin{equation}\nonumber
\tilde{\Phi} = \Phi - \frac{1}{2}\ln G_{00}.
\end{equation}

The most standard procedure to derive the relations (\ref{buscher}) is that of gauging the isometry group by introducing an auxiliary gauge field $A$, which is forced to be flat by introducing a Lagrange multiplier term $\lambda dA$. Then integrating out $\lambda$ gives back the original model, while integrating first $A$ gives the dual theory depending on $\lambda$, understood as the dual variable (\cite{buscher1987symmetry}). The dilaton shift from the perspective of this proceedure is a one loop effect that we see in the path integral formalism. For us, however, the canonical approach is more interesting: the background fields on the target define a Lagrangian on the space of embeddings of $\Sigma$ on $M$, which then gives it a symplectic structure (more precisely, on the tangent bundle of the loop space of $M$). Then the background fields are related by the Buscher rules if the induced symplectic manifolds are related by a specific type of canonical transformation.

Instead of giving the transformation now, we wish to approach it from the perspective of double field theory, as this will show one more piece of symplectic geometry which appears. The physical idea is to construct a version of the theory which has the duality as a manifest symmetry, by including both the original and dual coordinates. One can then recover not only the original and dual theories but also other equivalent backgrounds, which are related to each other by the infinite-order discrete group $O(n,n,\mathbb{Z})$ \cite{berman2015duality}.

Geometrically, let $\mathbb{T}^{n}\hookrightarrow M\to\tilde{M}$ be the torus fibration of target space $M$. After a Legendre transform, we get the phase space of the string model $T^{*}\mathcal{L}M$, where $\mathcal{L}M$ is the space of loops $X:S^{1}\hookrightarrow M$. As any cotangent bundle, this comes with the symplectic structure which at the point $X(\sigma)$ has the form
\begin{equation}\nonumber
\omega = \oint_{S^{1}}d\sigma (\delta P_{I}(\sigma)\wedge\delta X^{I}(\sigma)),
\end{equation}
where we think of the momentum $P = P_{I}(\sigma)\delta X^{I}(\sigma)$ as a section of the pullback of $T^{*}M$ to $S^{1}$ by $X(\sigma)$ and $\delta$ is the differential on $\mathcal{L}M$. The correct symplectic manifold for the sigma model with $H$-field, however, has the symplectic structure `twisted' by $H$,
\begin{equation}\nonumber
\omega_{M} = \omega + \oint_{S^{1}}d\sigma (\partial_{\sigma}X)\lrcorner H :=\omega + \frac{1}{2}\oint_{S^{1}}d\sigma \partial_{\sigma}X^{I}(\sigma)H_{IJK}(X(\sigma))\delta X^{J}(\sigma)\wedge\delta X^{K}(\sigma),
\end{equation}
where the components are defined by $H=H_{IJK}dX^{I}\wedge dX^{J}\wedge dX^{K}$ on $M$. Since locally $H=dB$, prequantization should construct a prequantum bundle over $T^{*}(\mathcal{L}M)$ with connection which can be expressed by $Ds = -\frac{i}{\hbar}\theta_{M}\otimes s$ for some local section $s$ and local symplectic potential
\begin{equation}\nonumber
\theta_{M} = \delta z + \oint_{S^{1}}d\sigma [p-(\partial_{\sigma}X)\lrcorner B],
\end{equation}
where $z$ is a local coordinate on the fibre.

In our case, $M$ is a principal torus bundle $\pi :M\to\tilde{M}$ with connection $\Theta$. This is given by a globally defined smooth one-form on $M$ with values in $\text{Lie}(\mathbb{T}^{n})=\mathbb{R}^{n}$ such that 
\begin{equation}\nonumber
\frac{\partial}{\partial\vartheta^{i}}\lrcorner\Theta = \text{id}\in \text{Lie}(\mathbb{T}^{n})^{*}\otimes \text{Lie}(\mathbb{T}^{n}) \,\,\,\, \text{ and } \,\,\,\, \mathcal{L}_{\partial/\partial\vartheta^{i}}\Theta = 0
\end{equation}
for the generators $\frac{\partial}{\partial\vartheta^{i}}\in C^{\infty}(TM\otimes Lie(\mathbb{T}^{n})^{*})$ of the torus action on the total space $M$. These two criteria imply that $d\Theta = \pi^{*}F$ for some two-form with integral periods on $\tilde{M}$ with values on $\text{Lie}(\mathbb{T})^{n}$. Furthermore, the space of maps $X:\Sigma\to M$ also has the structure of a bundle over the space of maps to the base $X:\Sigma\to\tilde{M}$: one takes the projection map to be composition with $\pi :M\to\tilde{M}$, 
\begin{equation}\nonumber
Map(\Sigma,M)\ni X\mapsto \pi\circ X\in \text{Map}(\Sigma,\tilde{M}),
\end{equation}
and then the fibre over $\pi\circ X$ is the space of smooth sections of the pullback by $\pi\circ X$ of the torus bundle $M\to\tilde{M}$, $C^{\infty}[(\pi\circ X)^{*}M]$. We use these to rewrite the symplectic potential in the case in which the sigma model is constructed over such a torus bundle
\begin{equation}\nonumber
\theta_{M}=\delta z + \oint_{S^{1}}d\sigma [p_{\alpha}\delta X^{\alpha}+\langle p,\Theta\rangle - (\partial_{\sigma}X+\nabla_{\sigma}\vartheta)\lrcorner B],
\end{equation}
where now $X:S^{1}\hookrightarrow\tilde{M}$ is a loop on the base $\tilde{M}$, $\vartheta\in C^{\infty}[X^{*}M]$ is the corresponding section of the pullback bundle over $S^{1}$ and $\nabla$ is the covariant derivative with respect to the pullback of the connection on $M\to\tilde{M}$. One can then use the fact that $B$ is a gerbe connection on $M$ to construct the prequantum bundle with the correct curvature.

The extended space is then constructed by the geometrization of the $2$-form $F'$ obtained from the contraction of $H$ with the vector fields $\partial/\partial\vartheta^{i}$, that is,
\begin{equation}\nonumber
F_{d}^{i}=\frac{\partial}{\partial\vartheta_{i}}\lrcorner H,
\end{equation}
which is a $2$-form in $\tilde{M}$ with values in $\text{Lie}(\mathbb{T}^{n})$ and integral periods. Specifically, one wants to think of it as the curvature of a connection $\Theta_{d}$ on a principal torus bundle $pr: N\to M$ over the total space of $M\to\tilde{M}$ with fibre $\mathbb{T}^{n}_{d}$, where we identify $\text{Lie}(\mathbb{T}^{n}_{d}) = \text{Lie}(\mathbb{T}^{n})^{*}$ (that is, $pr^{*}(F_{d}^{i})=d\Theta^{i}_{d}$). If $H$ satisfies some technical assumptions, this bundle will not only have an action of $\mathbb{T}^{n}_{d}$ but also one of the original torus $\mathbb{T}^{n}$ (note that, in principle, this acts only on $M$, but not necessarily on the total space of the new bundle $N\to M$). If this is the case, then $N$ itself is a principal double torus bundle over $\tilde{M}$ with fibre $\mathbb{T}^{n}\times\mathbb{T}^{n}_{d}$. We assume this is the case and denote by $\Theta_{D}$ a connection on $N$ which is compatible with the action of $\mathbb{T}^{n}\times\mathbb{T}^{n}_{d}$ and $F_{D}$ its curvature.

Up to some obstructions on $H$, one can proceed by substituting the phase space by $T^{*}(\mathcal{L}N)$ with symplectic structure $\omega_{N}=\delta\theta_{N}$, where
\begin{equation}\nonumber
\theta_{N} = \delta z + \oint_{S^{1}}d\sigma[p_{\alpha}\delta X^{\alpha} + \langle p,\Theta\rangle + \langle\Theta_{D},p_{D}\rangle - (\partial_{\sigma}X+\nabla_{\sigma}\vartheta)\lrcorner (B-\langle\Theta_{D},\Theta\rangle)].
\end{equation}
Now, if $\partial/\partial\vartheta_{i}$ and $\partial/\partial\vartheta^{i}_{D}$ generate the actions of $\mathbb{T}^{n}$ and $\mathbb{T}^{n}_{d}$ on $N$, then $\delta/\delta\vartheta_{i}$ and $\delta/\delta\vartheta^{i}_{D}$ generate the actions of $\mathbb{T}^{n}$ and $\mathbb{T}^{n}_{d}$ on $T^{*}(\mathcal{L}N)$. These actions are Hamiltonian
\begin{equation}\nonumber
\begin{cases}
\frac{\delta}{\delta\vartheta_{i}}\lrcorner\omega_{N}+\delta(p^{i}+\nabla_{\sigma}\vartheta^{i}_{D})=0\\
\frac{\delta}{\delta\vartheta^{i}_{D}}\lrcorner\omega_{N}+\delta(p_{D,i}+\nabla_{\sigma}\vartheta_{i})=0
\end{cases}\,\,\,,
\end{equation}
which, in particular, implies that they generate canonical flows on $T^{*}(\mathcal{L}N)$. We see that $\omega_{N}$ projects to a well-defined closed two-form on the quotient of $T^{*}(\mathcal{L}N)$ by the flow of the $\delta/\delta\vartheta^{i}_{D}$, though the projection is degenerate. Reducing the resulting pre-symplectic manifold gives back $T^{*}(\mathcal{L}M)$ with the symplectic structure $\omega_{M}$. However, one might just as well reduce with respect to the flows of the $\delta/\delta\vartheta_{i}$, which gives a manifold $M_{d}$ with symplectic structure $\omega_{M_{d}}$. The two are then automatically symplectomorphic, the symplectic diffeomorphism being generated by
\begin{equation}\nonumber
S = \frac{1}{2}\oint_{S^{1}}d\sigma(\vartheta^{i}_{D}(\sigma)\partial_{\sigma}\vartheta_{i}(\sigma)-\vartheta_{i}(\sigma)\partial_{\sigma}\vartheta^{i}_{D}(\sigma)).
\end{equation}
Additionally, this includes the $O(n,n,\mathbb{Z})$ structure, which acts by changing the subtorus of $\mathbb{T}^{n}\times\mathbb{T}^{n}_{d}$ with respect to which one does the symplectic reduction.

The function $S$ on $\mathcal{L}M\times\mathcal{L}M_{d}$ is to be understood as Hamilton's two-point function from classical mechanics. That is, one considers the graph of the one-form $\delta S$, which is $\Lambda = \{(p,q)\in T^{*}(\mathcal{L}M\times\mathcal{L}M_{d})|p=\delta S(q)\}$. Since $\delta S$ is closed, this submanifold is Lagrangian which, in turn, implies that the transformation $\rho:T^{*}(\mathcal{L}M)\to T^{*}(\mathcal{L}M_{d})$ defined implicitly by
\begin{equation}\nonumber
\Lambda = \{(p,p_{d},q,q_{d})\in T^{*}(\mathcal{L}M\times\mathcal{L}M_{d})=T^{*}(\mathcal{L}M)\times T^{*}(\mathcal{L}M_{d})|\rho(p,q)=(-p_{d},q_{d})\}
\end{equation}
is a symplectic diffeomorphism. We conclude that the transformation from $T^{*}(\mathcal{L}M)$ to $T^{*}(\mathcal{L}M_{d})$ is given by
\begin{equation}\nonumber
p = \frac{\delta S}{\delta q}, \,\,\,\, \text{ and } \,\,\,\, p_{d} = -\frac{\delta S}{\delta q_{d}}.
\end{equation}
If we restrict to one isometry, in the direction of $\vartheta$, this will then give
\begin{equation}\label{canon}
\begin{cases}
p_{\vartheta} = \frac{\delta S}{\delta\vartheta} = -\partial_{\sigma}\vartheta_{D}\\
p_{\vartheta}^{D} = -\frac{\delta S}{\delta\vartheta_{D}}=-\partial_{\sigma}\vartheta.
\end{cases}
\end{equation}

Let then $\rho$ be the composition of a Legendre transform on the $\vartheta$ variables, followed by the above transformation, and then by an inverse Legendre transform on the $\vartheta_{D}$ variables. An explicit computation shows that, if the fields $(G,B,\Phi)$ on $M$ and $(G_{D},B_{D},\Phi_{D})$ on $M_{D}$ are related by (\ref{buscher}) and $L$ and $L_{D}$ are the corresponding Lagrangians, then $L = L_{D}\circ\rho = \rho^{*}L_{D}$, which then implies that the symplectic structures are related by $\omega_{M} = \rho^{*}\omega_{M_{D}}$. Hence the connection between this type of canonical transformation and the Buscher rules.

As an imediate example, let us restrict to the $\vartheta$ coordinate, so that the Lagrangian becomes simply
\begin{equation}\nonumber
L=\frac{1}{2}G_{00}(\dot{\vartheta}^{2}-(\vartheta ')^{2}),
\end{equation}
where $\vartheta ' = \partial_{\sigma}\vartheta$. The Legendre transform gives
\begin{equation}\nonumber
H=\frac{P^{2}}{2G_{00}}+\frac{G_{00}(\vartheta ')^{2}}{2}, \,\,\,\,\,\,\,\,\,\,\,\,\,\,\,\, P = G_{00}\dot{\vartheta}.
\end{equation}
The canonical transformation (\ref{canon}) then gives
\begin{equation}\nonumber
H_{D} = \frac{(\vartheta_{D} ')^{2}}{2G_{00}}+\frac{G_{00}P_{D}^{2}}{2}.
\end{equation}
From Hamilton's equation, $\dot{\vartheta}_{D} = \delta H_{D}/\delta P_{D} = G_{00}P_{D}$, so the inverse Legendre transform gives
\begin{equation}\nonumber
L_{D} = \frac{1}{2G_{00}}(\dot{\vartheta}_{D}^{2}-(\vartheta_{D}')^{2}).
\end{equation}
So, indeed, if we had defined $L_{D}$ using $(G_{D})_{00} = G_{00}^{-1}$, we would have found $L=\rho^{*}L_{D}$.

Until now, all the symplectic geometry has appeared on the infinite-dimensional symplectic manifold $T^{*}(\mathcal{L}M)$. The dilaton, however, is a field on $M$, so it doesn't seem at first to be related to the inclusion of half-forms on the quantization procedure. The intermediate step is to look at the Fourier decomposition of the loops $X:\Sigma\to M$ \cite{kugo1992target}. Looking at the action of one of the $S^{1}\subset\mathbb{T}^{n}$, generated by translation in $\vartheta$, we write
\begin{equation}\nonumber
\begin{cases}
\vartheta(\sigma) = \vartheta_{0}+w\sigma + \text{ oscillators},\\
2\pi p_{\vartheta}(\sigma) = p + \text{ oscillators}.
\end{cases}
\end{equation}
Since the $\vartheta_{0}$ variable lies in the image circle $\vartheta(S^{1})$, one must have $w\in\mathbb{N}$ (we parametrize the coordinate on $S^{1}$ from $0$ to $2\pi$). We shall neglect the oscillators, as only the zero modes transform in a non-trivial way under T-duality. This has the effect of substituting $T^{*}(\mathcal{L}M)$ by a finite-dimensional manifold, because we parametrize each of the admissible loops simply by the coefficients of the zero modes. Substituting these in (\ref{canon}) and solving for $\vartheta_{D}$, $P_{D}$,
\begin{equation}\nonumber
\begin{cases}
\vartheta_{D}(\sigma) = \int_{0}^{\sigma}Pd\sigma = \vartheta_{D,0}+p\sigma+\text{ oscillators}\\
P_{D} = \frac{d}{d\sigma}\vartheta(\sigma) = w + \text{ oscillators},
\end{cases}
\end{equation}
where we have included an inversion $\sigma\mapsto -\sigma$. Hence it acts by swapping $p$ and $w$. Since we want to interpret $\vartheta_{D,0}$ as the coordinate conjugate to $w$, which has integer spectrum, it should take values in the interval $[0,2\pi]$ as well (note that this implies $p\in\mathbb{N}$ as well). One says that it is a coordinate on the dual circle. The action of the canonical transformation together with the Buscher rules preserves the Hamiltonian, which becomes
\begin{equation}\nonumber
H = \frac{p^{2}}{2G_{00}}+\frac{G_{00}w^{2}}{2} = \frac{p_{D}^{2}}{2(G_{D})_{00}}+\frac{(G_{D})_{00}w_{D}^{2}}{2} = H_{D}.
\end{equation}
Hence we take the phase space to be $T^{*}(\mathbb{T}^{2})=\{(p,w,\vartheta_{0},\vartheta_{D,0})\}$, where $\mathbb{T}^{2}$ is the torus made of the original and dual circles, with symplectic structure $G_{00}d\vartheta_{0}\wedge dp+(G_{D})_{00}d\vartheta_{D,0}\wedge dw$, where we treat $p$ and $w$ as continuous variables, with the understanding that their quantizations should be present in the resulting quantum theory.

Suppose now that the dilaton field defines a half-form on each of the symplectomorphic reductions, and that the two half-forms are related by one of the elements of the metaplectic group which correspond to the canonical transformation relating the two spaces. More specifically, we assume
\begin{equation}\nonumber
e^{\Phi}(d\vartheta_{0}\wedge dp)^{1/2} = \rho^{*}[e^{\Phi_{D}}(d\vartheta_{D,0}\wedge dw)^{1/2}],
\end{equation}
and use the fact that the two reductions of the symplectic structure on $T^{*}(\mathbb{T}^{2})$ should be mapped to each other, so that
\begin{equation}\nonumber
\begin{split}
e^{\Phi(\vartheta_{0})}(d\vartheta_{0}\wedge dp)^{1/2} &= \rho^{*}[e^{\Phi_{D}}(d\vartheta_{D,0}\wedge dw)^{1/2}] = \left[\rho^{*}\left(\frac{e^{2\Phi_{D}}}{(G_{D})_{00}}(G_{D})_{00}d\vartheta_{D,0}\wedge dw\right)\right]^{1/2} \\
&= \frac{e^{\Phi_{D}(\rho(\vartheta_{0}))}}{(G_{D})_{00}^{1/2}}[G_{00}d\vartheta_{0}\wedge dp]^{1/2} = \frac{e^{\Phi_{D}(\rho(\vartheta_{0}))}}{(G_{D})_{00}}(d\vartheta_{0}\wedge dp)^{1/2},
\end{split}
\end{equation}
where the commutation of the pullback sign with the square-root should be understood as choosing one of the two correponding elements of the metaplectic group. We conclude that one should have
\begin{equation}\nonumber
\Phi_{D} = \Phi - \frac{1}{2}\ln G_{00},
\end{equation}
which is the correct transformation law.

\end{stringdual}

\section{Acknowledgements}
DSB thanks Pierre Andurand for partial support during this work. DSB has benefitted from discussions on this topic with Luigi Alfonsi, Gary Gibbons, Malcolm Perry and Alan Weinstein. GC thanks support from NSF through Grant NSF DMR-1606591. GC wishes to thank Llohann Speran\c ca and Alexander Abanov for useful discussions and the mathematics department of UFPR for the invitation to lecture a short course based on these notes in their 2022 Graduate Summer School.

\pagebreak

\printbibliography

\end{document}